\documentclass[
onecolumn,
12pt,
nofootinbib, 						 
prd,
floatfix,
preprintnumbers,
amsmath,
amssymb,
groupedaddress,					
superscriptaddress 				
]{revtex4-2}

\usepackage{amsthm}			
\usepackage{mathtools}
\usepackage{tikz-cd} 			
\usepackage[sans]{dsfont} 	 
\usepackage{eucal} 				 
\usepackage[Euler]{upgreek} 
\usepackage{xcolor}				  
\usepackage[backref=page]{hyperref}  
\definecolor{linkCC}{RGB}{0, 0, 0}   
\definecolor{citeCC}{RGB}{0,50,50}   
\definecolor{fileCC}{RGB}{0,0,0}   
\definecolor{urlCC}{RGB}{0,50,50}    
\hypersetup{
    colorlinks=true,
    linkcolor=linkCC,       
    citecolor=citeCC,       
    filecolor=fileCC,       
    urlcolor=urlCC          
}
\def\thesection{\arabic{section}} 
\renewcommand\thesubsection{\thesection.\arabic{subsection}} 

\makeatletter
\renewcommand{\p@subsection}{}
\renewcommand{\p@subsubsection}{}

\makeatother
\newcommand{\nocontentsline}[3]{}
\newcommand{\tocless}[2]{\bgroup\let\addcontentsline=\nocontentsline#1{#2}\egroup}

\makeatletter
\newcommand{\pagetarget}[2]
{\hypertarget{#1}{#2}\protected@write\@auxout{}{%
   \string\newlabel{#1}{{#2}{\thepage}{page.\thepage}{#1}{}}}}
\makeatother

\newcommand{\pagelink}[2]
{\hyperlink{\getrefbykeydefault{#1}{name}{Doc-Start}}{#2}}

\newcommand{\pagenumber}[1]
{\getrefbykeydefault{#1}{page}{Doc-Start}}
\newcommand{\Fun}{\mathcal{F}\mathrm{un}}
\newcommand{\Map}{\mathrm{Map}} 
\newcommand{\bbC}{\mathbb{C}} 
 
\newcommand{\bbR}{\mathbb{R}} 
\newcommand{\Aut}{\mathrm{Aut}}
\newcommand{\R}{\mathbb{R}}
\newcommand{\C}{\mathbb{C}}
\newcommand{\Z}{\mathbb{Z}}
\newcommand{\ssR}{\mathds{R}}
\newcommand{\ssC}{\mathds{C}}
\newcommand{\ssZ}{\mathds{Z}}
\newcommand{\X}{\mathcal{X}}
\newcommand{\Cart}{\mathcal{C}\mathrm{art}}
\newcommand{\Sm}{\mathcal{S}\mathrm{m}}
\newcommand{\esh}[1]{{#1^{\delta}}}
\newcommand{\tge}[3]{\tau_{\geq #1}(#3,#2)}
\newcommand{\terminal}{\ast}
\newlength{\perspective}
\makeatletter
\begingroup
\catcode`\&=13
\gdef\and{&}
	\gdef\squareRaw#1#2#3#4#5#6#7#8{
	\begin{tikzcd}
	#1 \ar["{#5}"]{r} \ar["{#6}",swap]{d}\and #2\ar["{#7}"]{d} \\
	#3 \ar["{#8}",swap]{r} \and #4
	\end{tikzcd}
	}
	\gdef\square#1#2#3#4#5#6#7#8{
	\[\squareRaw{#1}{#2}{#3}{#4}{#5}{#6}{#7}{#8}\]
	}
	\gdef\squareE#1#2#3#4#5#6#7#8{
	\[\begin{tikzcd}
	#1 \arrow["{#5}",two heads]{r} \arrow["{#6}",swap]{d}\and #2\arrow["{#7}"]{d} \\
	#3 \arrow["{#8}",swap,two heads]{r} \and #4
	\end{tikzcd}\]
	}
	
	\gdef\pastingVE#1#2#3#4#5#6{
	\[\begin{tikzcd}
	#1 \ar[r,two heads]\ar{d}\and #2\ar{d} \\
	#3 \ar[r,two heads] \ar{d}\and #4\ar{d} \\
	#5 \ar[r,two heads] \and #6
	\end{tikzcd}\]
	}
	
	\gdef\pastingHE#1#2#3#4#5#6{
	\[\begin{tikzcd}
	#1 \arrow[two heads]{r} \arrow{d}\and #2\arrow{d}\arrow[two heads]{r} \and #3\arrow{d}\\
	#4\arrow[two heads]{r} \and #5 \arrow[two heads]{r} \and #6
	\end{tikzcd}\]
	}
	\gdef\twosimplex#1#2#3#4#5#6#7{
	\[\begin{tikzcd}
	\and {#2} \ar[dr,"{#6}"] \\
		{#1} \ar[ur,"{#5}"]\and \and {#3}
		\arrow[""{name=0, anchor=center, inner sep=0}, "{#7}"', from=2-1, to=2-3]
		\arrow["{#4}"{description}, draw=none, from=1-2, to=0]
	\end{tikzcd}\]
	}
\endgroup
\makeatother

\begin{document}   
\title{Smooth generalized symmetries of quantum field theories} 
\date{\today}
\author{Ben Gripaios}  
\email{gripaios@hep.phy.cam.ac.uk}
\affiliation{
Cavendish Laboratory, University of Cambridge,
JJ Thomson Avenue,
Cambridge, UK}
\author{Oscar Randal-Williams}  
\email{o.randal-williams@dpmms.cam.ac.uk}
\affiliation{
  DPMMS, University of Cambridge, Wilberforce Road, Cambridge, UK}
\author{Joseph Tooby-Smith}  
\email{j.tooby-smith@cornell.edu}
\affiliation{Department of Physics, LEPP, Cornell University, Ithaca, NY 14853, USA}
\begin{abstract} 
Dynamical quantum field theories (QFTs), such as those in which spacetimes are equipped with a metric and/or a field in the form of a smooth map to a target manifold, can be formulated axiomatically using the language of  $\infty$-categories. 
According to a geometric version of the cobordism hypothesis, such
QFTs collectively assemble themselves into objects in an
$\infty$-topos of smooth spaces. We show how this allows one to define
and study generalized global symmetries of such QFTs. The symmetries
are themselves smooth, so the `higher-form' symmetry groups
can be endowed with, {\em e.g.}, a Lie group structure.

Among the more surprising general implications for physics are, firstly, that QFTs in
spacetime dimension $d$, considered collectively, can have $d$-form
symmetries, going beyond the known $(d-1)$-form symmetries of
individual QFTs and, secondly, that a global symmetry of a QFT can be anomalous even before we try to gauge it,
due
to a failure to respect either smoothness (in that a symmetry of an
individual QFT does not smoothly
extend to QFTs collectively) or locality (in that a symmetry of an
unextended QFT does not extend
 to an extended one). 

Smoothness anomalies are shown to occur  
even in 2-state systems in quantum mechanics (here
formulated axiomatically by equipping $d=1$ spacetimes with a metric,
an orientation, and perhaps some unitarity structure). Locality anomalies are shown to
occur even
for invertible QFTs defined on $d=1$ spacetimes equipped with an
  orientation and a smooth map to a target
manifold. These correspond in physics to topological actions for a
particle moving on the target
and the relation to an earlier classification of such actions using invariant
differential cohomology is elucidated. 
\end{abstract}   
\maketitle
\tableofcontents
\begingroup
\catcode`\&=13
\gdef\and{&} 
\endgroup

\section{Introduction\label{sec:Intro}}
Previously \cite{GRWTS23}, we gave an interpretation of generalized
symmetries of topological quantum field theories (TQFTs) using homotopy theory. In a
nutshell, the idea there was that
TQFTs 
with a given tangential structure assemble
themselves, according to the cobordism hypothesis \cite{LurTFT}, into a
space $X$ and that a 
  generalized global\footnote{We also discussed generalized gauge
  symmetries in \cite{GRWTS23}, but will not do so here.} symmetry of a
particular TQFT can be understood as a fixed point of an external group
$G$ acting on $X$, provided that all of these notions are interpreted in a
suitably homotopy-theoretic sense. So by a `space' we mean a homotopy
type,  and by a `group' we mean 
not a
group in the usual sense of a set with an associative and invertible multiplication
law, but rather a space with a notion of associative
invertible multiplication up to homotopy.

As noted in \cite{GRWTS23}, one upshot of this interpretation is that there is rather more structure associated to a generalized
symmetry than physicists had previously observed. The known tower
of abelian `$n$-form
symmetry' groups, for $n \geq 1$, lying above the usual `$0$-form
symmetry' group correspond here to the homotopy groups $\pi_n(G)$, for
$n \geq 0$.
 So we get, for free, an action of $\pi_0(G)$ on $\pi_n(G)$ for each
 $n$, a  Samelson product $\pi_n(G)
\times \pi_m(G) \to \pi_{n+m}(G)$ for each $n$ and $m$, \&c. In
fact, every $G$ is
equivalent to the loop space of some pointed connected space, so the
structure associated to a generalized symmetry group $G$ is completely described
by its delooping $BG$. Moreover, a group
action of $G$ on $X$ is equivalent to a morphism to $BG$ whose
fibre over the basepoint is $X$ and a homotopy fixed point is equivalent
to a section of that morphism; in favourable cases, such as TQFTs in
spacetime dimension one or two, this enabled us to give a more-or-less explicit description
of all possible actions and their homotopy fixed points.

To give a
concrete example of the power of this approach, consider the global
symmetries of gauge theory
on oriented spacetimes with dimension $d=2$ and gauge group a
 discrete group
$K$. A gauge field is a connection on a principal
$K$-bundle over the spacetime manifold, but, since $K$ is discrete, such
a connection is unique, so we can instead think of the gauge
fields simply as maps from spacetime to $BK$. 
A classical action can then be formed, following
Dijkgraaf and Witten \cite{DW90}, by taking a class in 
$H^2(BK,\bbC^\times)$, pulling back to spacetime along the
gauge field, and pairing with the fundamental class.
In the case where $K=(\Z / p)^2$, such that the classical action is
specified by $q \in H^2(B (\Z / p)^2,\bbC^\times) \simeq \Z/p$, semiclassical arguments
were used in \cite{GKSW15} to argue that the quantized theory has 0- and 1-form symmetry
groups both isomorphic to $(\Z/r)^2$,
where $r$ is the greatest common divisor of $p$ and $q$. In fact, it
was shown in \cite{GRWTS23} that these are mere subgroups of much larger 0- and 1-form
symmetries, isomorphic to $S_{r^2}$ and $(\bbC^\times)^{r^2}$,
respectively.
The underlying reason for this larger symmetry is duality: given an arbitrary
classical gauge theory as above the resulting quantum theory depends,
up to equivalence,
only on the dimensions of the irreducible projective representations
of $K$ corresponding to the class in $H^2(BK,\bbC^\times)$. (In our
specific example, the theories specified by $(p,q)$ have $r^2$ inequivalent irreducible
projective representations, each of which has dimension
$p/r$.) To spell
it out, many
theories that are classically distinct
are quantumly equivalent  and so the symmetries of any one QFT are consequently
larger. This is, presumably, hard to see using semi-classical
arguments, but is child's play once one knows the corresponding $X$:
one needs only to find the $x \in X$ that corresponds to the given
classical action. Once one has done so, one knows the full
generalized symmetry structure (a connected homotopy 2-type); in our
specific example, it is
specified by the further information that $S_{r^2}$ acts on
$(\bbC^\times)^{r^2}$ by permutation with Postnikov invariant given by the
trivial class in $H^3\left( BS_{r^2}, (\bbC^\times)^{r^2}\right)$. 

In this work, we wish to play a similar game in a set-up that is
rather closer to the real world in two ways. Firstly, we shall
generalize from TQFTs to genuine 
quantum field theories (QFTs), equipped with geometric structures such
as
maps to some target manifold ({\em i.e.} quantum fields), spacetime metrics, and
so on. Secondly, we shall ask that the generalized symmetries of
these theories be equipped with a suitable smooth structure. This
allows us, for example, to consider the $n$-form symmetry groups not
just as abstract groups, but as Lie groups, as physicists are wont to
do. In a nutshell, the idea here will be that QFTs with a given
geometric structure assemble themselves, via a geometric version of
the cobordism hypothesis \cite{GP21}, into an object in a certain 
$\infty$-topos whose objects we generically call smooth
spaces. Many of the relevant constructions from homotopy theory
can be transferred to any $\infty$-topos \cite{Lur09,NSS14}. In particular, not only do we have analogous notions
of a group object $G$, an action of $G$ on another object $X$, and
a fixed point thereof (up to a notion of homotopy), but we
also have equivalences (via a notion of delooping) sending
these, respectively, to 
a pointed connected object $BG$, a morphism to $BG$ with fibre $X$
over the basepoint, and a section thereof. As a result, we are able to carry
over the discussion in \cite{GRWTS23}, {\em
  mutatis mutandis}, to the $\infty$-topos of smooth spaces. Here, a
group object $G$ is equipped with a smooth structure, as desired,
giving us a new, and hopefully powerful,
perspective on smooth generalized symmetries of dynamical ({\em i.e.}
non-topological) QFTs.

Among the juicier titbits, we find that QFTs in spacetime
dimension $d$ have a non-trivial notion of an $n$-form symmetry for $n
\leq d$ (rather than
$n < d$, as is the common lore) and that a global
symmetry can be anomalous even before one tries to gauge it, either
due to a failure to respect smoothness, or due to a failure to respect
locality. To be explicit, a smoothness anomaly (which occurs even in
quantum mechanics) arises when a symmetry of a particular QFT fails to
extend smoothly to the smooth space of QFTs, while a locality anomaly
arises when a symmetry of an unextended QFT fails to extend when we extend
the QFT.  These anomalies reinforce the view we espoused in \cite{GRWTS23} that
anomalies should be considered not
 so much as a problem associated to gauging a global symmetry, but rather to a failure
 to define global symmetries properly in the first place, namely in a
 way that is fully consistent with the sacred physical principles of
 locality and smoothness. Such anomalies arise frequently in practice, because
physicists have a nasty habit of defining QFTs by writing them down not only one at a
time, so smoothness is obscured, but also by specifying their values
only on closed spacetime manifolds (and often on only one, such as the
$d$-dimensional sphere)
so locality is
obscured as well.

We illustrate all of this using known examples of smooth
spaces of QFTs. Unfortunately, our collective ignorance as to which higher category
to take as the target for QFTs in $d>1$ means that most known examples
of $X$ are in $d=1$, but the principles
apply quite generally. The holy grail would be, given $X$ and $G$, to 
classify all possible $G$-actions on the whole of $X$ and their corresponding homotopy fixed
points, as we did for TQFTs in \cite{GRWTS23}, but it will become clear that this is a daunting task in
general, even in $d=1$. For example, for quantum mechanics with a
1-dimensional state space, the hamiltonian describing the time
evolution along an interval is given by a real number; if $G$
corresponds to a Lie group, there is a $G$-action for every smooth
action of the Lie group on $\bbR$ and part of the data of a homotopy
fixed point is the fixed point subset of the smooth action.

Though dispiriting, it is important to remark
that finding all actions on $X$ and their homotopy fixed points goes
way beyond what the typical physicist does. As we have already remarked, such a physicist,
is, by-and-large, ignorant of what $X$ is and instead
contents themself with
studying symmetries of QFTs considered one at a time. In our
framework, a QFT is represented
by a point $x$ of $X$ and the problem of 
finding all physical symmetries of that QFT considered in
  isolation can be phrased here as the problem of
finding all $G$-actions and homotopy fixed points of the connected
component of $X$ at $x$. As we shall see, connectedness makes the
problem of finding $G$-actions and homotopy fixed points more straightforward.
It is much less straightforward to
establish whether or not actions on a connected component extend to actions
on all of $X$; if not, we have the aforementioned smoothness
anomaly. 

Explicitly, our main examples are QFTs in $d=1$
equipped with an orientation and either a smooth map to
a target manifold or a metric. The most striking result in the latter
case (which upon addition of a suitable unitarity structure leads us
to quantum mechanics) is that one can have a smoothness
anomaly already for a quantum system whose state space is
two-dimensional. In the former case, we find a
locality anomaly even for a system whose state space is
one-dimensional. Such a QFT is a so-called invertible QFT, and we can think
of it as a semi-classical version of the mechanics of a point particle
moving on the target manifold. Each such QFT defines a topological
physics action, by evaluating the QFT on closed
spacetime manifolds ({\em i.e.} disjoint unions of circles). For
group actions that are induced by a smooth action of a Lie group on
the target manifold, we can compare with the classification of
invariant topological actions proposed in \cite{DGRW20} in terms of invariant differential
cohomology. The fact that a QFT here is defined in a fully-extended fashion, {\em i.e.} we specify its value not just
on closed spacetime manifolds ({\em i.e.} disjoint unions of circles),
but also on manifolds with boundary ({\em i.e.} disjoint unions of
intervals and circles), leads to very different
classifications. Indeed, the map from the extended symmetric theories to the
unextended ones defined by restriction is neither surjective (meaning
that a symmetry of the physics action does not extend, 
giving a locality anomaly, as defined above) nor injective (meaning
that a symmetry of the physics action can be extended in multiple
ways), in general. Two simple examples from physics may serve to illustrate the general situation:
firstly, for a particle moving in a plane in the presence of a uniform magnetic field,
the action for motion in a loop is invariant under translations of the
plane, but
the fact that the lagrangian shifts by a total derivative means that
it cannot be extended to an interval;
secondly, the physics action describing motion of a particle moving on a point is obviously symmetric under
the trivial action of any group on the point. But this symmetry
can be extended to the theory defined on an interval in different
ways, by assigning different characters of the group
to the 
one-dimensional state space assigned by the theory to a spacetime point.

In case the basic point gets lost amidst all the mathematical heavy machinery,
let us consider further the example of quantum mechanics. This is
certainly a `real world' example (even though we consider here only the simplified
case where the state space is finite-dimensional). The smooth space
$X$ that we obtain in this case corresponds to (for each finite
dimension $n$ of the state space) the smooth manifold of
$n \times n$ hermitian matrices together with the smooth action on it by
conjugation of the Lie
group of $n \times n$ unitary matrices. Physically, the hermitian matrices represent the possible
hamiltonians of a quantum-mechanical
system. Every physicist knows that hamiltonians
  related by conjugation by unitary matrices are equivalent, but
  identifying them and choosing one from each equivalence
  class ({\em e.g.} by diagonalizing and quotienting out by permutations), as
  a physicist usually does, not only results in a set which is not a
  smooth manifold, but also entails a loss of information. 
 It is better, in general, to 
keep track manifestly of the ways in which
hamiltonians are equivalent and this is what the smooth space $X$
does. It is better, in particular, for describing symmetries of quantum mechanical
theories
  because it allows us to keep track of transformations that send a hamiltonian not to
itself, but to a conjugate one, in a smooth fashion.

In contrast to the physics, we lay no great claim to mathematical novelty. (In
  particular, the notion of an $\infty$-topos is now
  well-established, largely thanks to \cite{Lur09}, and the notion of group
  actions therein was developed in \cite{NSS14}.) But one or
  two minutiae may be of mathematical interest. To give one example, we
  introduce in \S\ref{sec:orbstab}
  notions, for each non-negative integer $m$, of the
$m$-orbit and $m$-stabilizer of a point under a group
action in an $\infty$-topos. In addition, there are a number of new
results regarding group actions in $\infty$-topoi and their associated
homotopy fixed points.

{\em A brief review of existing literature.} It is worth
  highlighting a few papers relevant to our discussion here. Locality
  anomalies have previously been discussed in
  \cite{TW20,M20}. Anomalies associated with smoothness are discussed
  in \cite{CFLS19a,CFLS19b}, although it is unclear how they are
  related to the smoothness anomalies discussed here. The way we enforce the
  condition of a field theory having a $d$-form symmetry has been used
  to enforce the spin-statistics relation in \cite{MS23}.

{\em A note on style and notation.} Since our target audience consists
of physicists,
we have opted for an informal presentation, emphasing pedagogy over
details.
To this end, we eschew definitions, theorems, and
proofs whenever possible; such impedimenta, when not supplied in the
Appendix or given
an explicit reference, can mostly be found in
\cite{Lur09}. We mostly follow the
nomenclature and notation 
used in \cite{Lur09} with one significant exception: an `$\infty$-category'
there will henceforth be called simply a `category' here (and denoted,
as there, using calligraphic typeface); the same goes for all
$\infty$-categorical paraphernalia, such as functor, adjoint, (co)limit, topos, \&c. On the rare occasions where we need to
refer to a bog-standard category, we will call it a `$1$-category',
\&c.\footnote{Lest readers who know their $1$-categorical onions be lulled into a
false sense of security, we stress
 that their intuition may fail spectacularly on
occasion. For example, consider a point \( x: \terminal \rightarrow X
\) within an object \( X \), where \( \terminal \) is a terminal
object. The pullback of this point over itself is not a terminal
object, end of story.
On the contrary, it is the very beginning of the story of homotopy
theory in a topos, defining as it does the 
object
of loops in $X$ based at $x$. Those who are {\em au fait} with classical homotopy theory
should be on their guard too. For example, an object $X$ being
connected is neither a necessary nor a sufficient condition for it to
be what we shall call path-connected, namely such that any two points $x,x^\prime:
\terminal \to X$ in it are homotopic to one another. In the topos of
  smooth spaces, being connected implies being path-connected, but not the
  converse. The classifying space of principal bundles with connection,
  described in \S\ref{sec:smoothspaces}, provides a counterexample.}
As usual, by a pointed
object, we mean one equipped with the datum of a point, {\em i.e.} a morphism $x: \terminal \to X$ from a terminal object. As in \cite{Lur09}, by `space' we mean `$\infty$-groupoid' and
denote the archetypal topos in which they live by ${\cal
  S}$. Similarly, we too
allow ourselves to use `the' or `an essentially unique' to
mean `a' in situations where the gamut of possibilities forms a contractible
space. When we come to consider the specific topos $\Sm$ of smooth spaces in
\S\ref{sec:smoothspaces}, we shall frequently need to pass back and
forth between classical constructs in differential geometry, such as
manifolds and Lie groups, and their corresponding objects in the topos
of smooth spaces. To distinguish the two, we use sans serif typeface
({\em e.g.} $\mathsf{M}$, $\mathsf{G}$ or $\ssR$) for the former and serif
typeface for the latter ({\em e.g.} $M$,  $G$ or $\R$).

The outline is as follows. In the next Section, we recall some
relevant details of topos theory and describe
notions of groups,
actions and homotopy fixed points in an arbitrary topos. In \S\ref{sec:smoothspaces}, we
describe the particular topos of smooth spaces. In \S\ref{sec:examples}, we describe some
examples in smooth spaces that are relevant to physics.
\section{Group actions in a topos\label{sec:groups}}
Our goal in this Section is to sketch the relevant details of topos
theory and 
introduce various notions and constructions related to group actions
therein. The notion of a topos and a group object therein were given
in \cite{Lur09} and group actions were discussed in \cite{NSS14}; the notions of orbits,
stabilizers, and various results that will be useful in our later
analysis of examples are perhaps new. 
To keep in the spirit of things, we mostly express
these in a functorial fashion, which requires a certain amount
of abstraction. Happily, 
unwinding the definitions at the level of objects usually results in something straightforward.

A topos is a special kind of category, which in turn is
designed to be a homotopy-theoretic
generalization of a 1-category. So let us begin by describing how the latter
is achieved. 

A category is a special kind of simplicial set.
Like a 1-category, it has objects and
morphisms (given respectively by the 0-simplices, a.k.a. vertices, and
1-simplices, a.k.a. edges, a.k.a. 1-morphisms, of
the simplicial set), along with an identity morphism $\mathrm{id}_X: X
\to X$ for each object $X$
(given by its image under the degeneracy map), but unlike a 1-category it also has $n$-morphisms
for every $n >1$ (given by the $n$-simplices). The r\^{o}le of these
is to provide a notion of the associative
composition of morphisms defined in a 1-category that is weaker in
that it holds only up to coherent
homotopy. To wit, 
given
a morphism into some object and a morphism out of that same object,
there is a contractible space of 2-simplices that are sent by the
appropriate face maps to the input morphisms. Each such
2-simplex gives the data of a possible composition (the image under
the remaining face map) along with a coherent homotopy witnessing it
(the 2-simplex itself). The 3-simplices provide a homotopy-coherent
notion of associativity, and so on. Unlike a 1-category, where the
morphisms between two objects form a set, in a category the morphisms form a space, whose 0-simplices are the 1-morphisms.

A functor is simply a map of simplicial sets, so boils down to a map
of sets for each $n$, preserving the relations between the face and
degeneracy maps.

The archetypal category $\cal{S}$ has objects given by spaces up to weak homotopy equivalence. A topos $\cal{X}$ is a generalization of
$\cal{S}$ which maintains many of its desirable properties, admitting 
many notions akin to those in homotopy theory. There are (small)
limits and colimits, so we have the terminal object $\terminal$ along with
the terminal morphism $\terminal_X: X \to \terminal$ for each object $X$.
In addition, we have notions of $n$-truncated and
$n$-connected morphisms, for each $n \geq -2$, defined as follows. Firstly, recall that a space is called $n$-truncated if its $i$-th homotopy group vanishes for every $i>n$ and every choice of basepoint.\footnote{A space is
  $-2$-truncated iff it is equivalent to the point, while a space is
  $-1$-truncated iff it is equivalent to the point or to the empty space.} An object $X$ of the topos $\X$ is called $n$-truncated if for any $Y\in
\X$ the mapping space $\X(Y,X)$ is $n$-truncated. The inclusion of the full subcategory $\tau_{\le n}\X$ of $n$-truncated objects in $\X$ has a left adjoint $\tau_{\leq n}$ and we say that an object $X\in \X$ is $n$-connected if $\tau_{\leq n} X\simeq \terminal$. A morphism $f:Y\to Z$
  in $\X$ is called $n$-truncated (respectively, $n$-connected) if it is an
  $n$-truncated ($n$-connected) object in $\X_{/Z}$. 
   For each $n$, the $n$-connected/$n$-truncated morphisms form an orthogonal factorization system.
Among other things, this implies that any morphism 
$f: X \to Y$ in $\X$ canonically factorizes as the composition of
an $n$-connected morphism 
followed by an $n$-truncated one and that, given a commutative square
\[\begin{tikzcd}
    X \ar[r] \ar{d}[swap]{n-\mathrm{conn}} & Y \ar{d}{n-\mathrm{trun}}\\
    Z \ar[r] \ar[ur,dashed] & W
\end{tikzcd}\]
in which the morphism on the left is $n$-connected and the morphism on
the right is $n$-truncated, there is an essentially unique dashed
arrow that makes the diagram commute. 
For $n=-2$, these notions are not particularly exciting,
since any morphism is $-2$-connected and a morphism is $-2$-truncated
iff it is an equivalence; for $n=-1$ we have the notions of
monomorphism ($-1$-truncated) and effective epimorphism ($-1$-connected), generalizing the notions in
$\cal{S}$ of morphisms that induce respectively an injection on $\pi_0$ (and
bijections on higher homotopy groups)
and a surjection on $\pi_0$.  The reader is encouraged to think of
an effective epimorphism $X \twoheadrightarrow Y$ as a cover of $Y$ by
$X$.

Along similar lines, given a pointed object, {\em i.e.} an object $X
\in \X$ and a morphism $x:\terminal \to X$ in $\X$, we define its $k$-connected cover to be the object $\tau_{\geq (k+1)}(X,x)$ defined via the cartesian square
\begin{equation}\label{eq:ConnectedCover}
\squareRaw{\tau_{\geq (k+1)}(X,x)}{X}{\ast}{\tau_{\le k} X}{}{}{}{}
\end{equation}
where $X\to \tau_{\le k }X $ is the $k$-connected morphism in the
$k$-connected/$k$-truncated factorization of $X\to \ast$ and where
$\ast\to\tau_{\le k} X$ is the composition of $x$ and $X\to \tau_{\le
  k }X $. Since $k$-connectedness is stable under pullback,
$\tau_{\geq (k+1)}(X,x)$ is indeed $k$-connected, as the name
suggests. It is, moreover, canonically pointed, via the morphism $x$
and the universal property of 
(\ref{eq:ConnectedCover}). As is traditional in spaces, we call $\tau_{\geq
    1}(X,x)$ the connected component of $X$ at~$x$.
    
 Before going on to groups and their actions, let us define one finial bit of notation. For each topos $\X$ there is a special functor which takes an object $X\in \X$ to the mapping space $\X(\ast, X)$. We denote this functor $\Gamma : \X \to {\cal S}$, and will call the object  $\Gamma X$ the underlying space of $X$. 
\subsection{Groups and group actions as groupoids \label{sec:gpd}}
Group objects and group actions in a topos are both special cases of
groupoid objects; the latter may be viewed as the appropriate homotopical
generalization of the run-of-the-mill notion of an equivalence relation on
a set. This generalization is achieved in three steps. Firstly, since
the very essence of homotopy theory is that we care not just about
whether two things are equivalent, but rather about all the ways in which
they are equivalent, we generalize from an equivalence relation on a set to a
groupoid. Thus, the elements of the set become objects in a
1-category and the equivalence relations between them become morphisms
that are isomorphisms. (To go back, we demand that
there be at most one isomorphism between any two objects.) Secondly,
we generalize to groupoids valued in any 1-category by asking for a
1-simplicial object $U_\bullet$ in that 1-category obeying certain conditions due to
Segal \cite{Seg68} (and earlier Grothendieck \cite{Gro61}). (To go back, we
demand that the 1-category be the 1-category of sets.) 
Thus, denoting by $U_n$ the object of $n$-simplices, our original set
and equivalences are replaced respectively by $U_0$ and $U_1$ and the
identities are determined by the degeneracy map $s_0: U_0 \to
U_1$. The Segal conditions for $U_2$, for example, require that the
three squares
\[
\squareRaw{U_2}{U_1}{U_1}{U_0\mathrlap{,}}{d_0}{d_1}{d_0}{d_0}\qquad
\squareRaw{U_2}{U_1}{U_1}{U_0\mathrlap{,}}{d_2}{d_0}{d_0}{d_1}\qquad
\squareRaw{U_2}{U_1}{U_1}{U_0\mathrlap{,}}{d_2}{d_1}{d_1}{d_1}
\]
are 1-cartesian. Collectively, these
  define composition along with left- and right-
  inverses. Associativity is enforced by the Segal conditions for
  $U_3$.
  Thirdly, we generalize from 1-categories to categories by removing
  the `1-'.

  Having done this legwork, we may now easily define a group object
  as
  a groupoid object 
whose $U_0$ is the terminal object $\terminal$. As a result of the Segal conditions, we have
that $U_n$ is the $n$-fold product of $U_1$. We shall sometimes abuse
notation, by denoting both
the simplicial object $U_\bullet$ and the underlying object $U_1$ of a
group object by, {\em e.g.}, $G$.

Now we come to 
  a first miracle of topos theory: there is an equivalence of
  categories between groupoid objects in $\cal{X}$ and effective
  epimorphisms in $\cal{X}$,
  obtained in one direction by taking the colimit of the groupoid
  object and in the other direction by taking the \v{C}ech nerve.
Under the colimit, a group object $G$ is thus sent to an effective
epimorphism out of the terminal object and into an object that we call the
delooping of $G$ and denote $BG$. 
Since $\terminal \twoheadrightarrow BG$ is $-1$-connected
and is a section of $BG \to \terminal$, it follows \cite[Prop. 6.5.1.20]{Lur09} that $BG$ is
$0$-connected, a.k.a. connected. In short, there is an equivalence of
categories between group objects in $\cal{X}$ and pointed, connected
objects in $\cal{X}$. Going in the other direction, the first step of
the \v{C}ech nerve construction
sends a morphism to its fibre product with itself and so
we define, for any object $X$ with basepoint $x:\terminal \to X$, the object
$\Omega_x X$ (we write simply $\Omega X$ if the basepoint is clear) of
loops in $X$ based at $x$ via the cartesian square
\square{\Omega_x X}{\terminal}{\terminal}{X\mathrlap{.}}{}{}{x}{x}

Using this equivalence, we may define, following
  \cite[Defn. 3.1]{NSS14} and~\cite[Prop. 3.2.76]{SS21}, a group action of a group object
$G$ on an object $X$ as a morphism to $BG$ together with an
identification of the fibre over the basepoint $\terminal \twoheadrightarrow BG$ with
$X$. In other words, it is the data of an object $X/\!\!/G$ and a cartesian square
\squareE{X}{X/\!\!/G}{\terminal}{BG\mathrlap{.}}{}{}{}{}
Since effective epimorphisms are stable under
pullback, the morphism $X \twoheadrightarrow X/\!\!/ G$ is an effective
epimorphism, as indicated. To see that this is a sensible definition
of a group action, take the
\v{C}ech nerve of $X\twoheadrightarrow X/\!\!/G$. This results upstairs in a groupoid, in which moreover
$U_n$ is
equivalent to $X \times G^n$. We have a diagram of the schematic
form
\[\begin{tikzcd}
 \cdots & X \times G\times G  \ar{r} \ar[shift left=2]{r} \ar[shift
 right = 2]{r}& X \times G \ar[shift left]{r}{d_0} \ar[shift
 right]{r}[swap]{d_1}& X \ar[r,two heads] & X/\!\!/G\mathrlap{,}
\end{tikzcd}\]
where the face morphisms 
$d_{0}, d_1: X \times G  \to X$ correspond to
projection on the right hand factor and the group action respectively, but the latter
satisfies the usual properties only up to higher coherent homotopies
specified by the subsequent stages.

The holy grail for the purposes of physics would be to classify all possible
actions by $G$ and to find the corresponding homotopy fixed points. To do so,
we need a notion of when two $G$-actions are equivalent and it is
convenient to do so by assembling a category of $G$-actions on
arbitrary objects in $\X$. A suitable category (which is, in fact, a topos) is
given by the slice category $\X_{/BG}$ (see {\em e.g.} \cite[Prop 3.2.76]{SS21}). Given an object therein, {\em
  i.e.} an object $E$ in $\X$ along with a morphism $E \to BG$ in
$\X$, we recover the object in $\X$ on which $G$ acts by pulling back
$E \to BG$ along the basepoint $\terminal \twoheadrightarrow BG$ of $BG$. In
cases where there is little risk of confusion, we will often write simply $E$ to denote the object $E
\to BG$ in $\X_{/BG}$ in the following.

As usual when studying group actions, it is useful to consider the
notion of an equivariant morphism, which we define to be a morphism in
$\X_{/BG}$. As we shall
see in the next Subsection, this notion is convenient for discussing
homotopy fixed points.
Occasionally, we abuse nomenclature and use the same term to refer to the underlying nonequivariant morphism in $\X$ induced by pullback along $\terminal\twoheadrightarrow BG$.
\subsection{Homotopy fixed points \label{sec:fix}}
We now wish to define a homotopy-theoretic version of a fixed
point for a group object $G$ acting on an object $X$ in a topos ${\cal
  X}$. As we have seen, a group action is a fibre sequence and so it is
natural to guess that a homotopy fixed point should be a section of the
corresponding morphism $X/\!\!/G \to BG$. To see that this makes sense, we
note that we may
pull back a section $s: BG
\to X/\!\!/G$ to get a commutative diagram
\[\begin{tikzcd}
    \terminal \ar[d,"x",swap] \ar[r,two heads] & BG \ar[d, "s"]\\
    X \ar{d}\ar[r,two heads]& X/\!\!/G \ar{d}\\
    \terminal \ar[r,two heads]& BG \mathrlap{,}
\end{tikzcd}\]
in which all squares are cartesian and the vertical composite
morphisms are identities. In this way, we associate to every section
the data of a morphism $x:\terminal \to X$ in ${\cal X}$, {\em i.e.} a point in $X$, that is moreover
equivariant with respect to the given action of $G$ on $X$ and the
trivial action on $\terminal$, reproducing the usual 1-categorical notion of a fixed
point of a group action. We say that $s$ is a homotopy fixed point
  through the underlying fixed point $x$.
Note that there may be more than one section through a given $x$, meaning that there is more data
associated to a homotopy fixed point than its underlying fixed point. This reflects the fact
that to be
a {\em homotopy} fixed point is not merely a property, but a structure: to be fixed means to be fixed up to
homotopy, and the structure records the possible homotopies that do
the job.

Since we shall be interested for physics in the question of finding
all homotopy fixed points, it is behoves us to find a home where
homotopy fixed points can live. As in \cite{GRWTS23}, a
useful strategy for physics (we shall discuss why shortly) is
to assemble the homotopy fixed points into an object $X^{hG}$ in $\X$,
which we furthermore equip with a morphism
to $X$, corresponding to sending a homotopy fixed point to its
  underlying fixed point and forgetting that it is a
fixed point.\footnote{Again, because to be a homotopy fixed point is a
  structure not a property, we cannot expect the morphism $X^{hG} \to X$ to be a
monomorphism, {\em i.e.} $-1$-truncated, unlike
the case of a common or garden group action on a set
(where it corresponds
to the
the inclusion of the fixed point subset).} 
This is achieved as follows. Given any morphism $f:Y \to Z$ in a topos $\X$
we have a triple adjunction~\cite{Lur09} 
\[\begin{tikzcd}
	\X_{/Y}
	\ar[r, "{f_!}", shift left=4.5, ""{name=0, anchor=center, inner sep=0}] 
	\ar[r, "{f_*}", shift right=4.5, swap,""{name=2, anchor=center, inner sep=0}] 
	& \X_{/Z}\mathrlap{.}
	\ar[l, "{f^*}"{description},""{name=1, anchor=center, inner sep=3}]
	\ar[draw=none, from=1, to=2, "\dashv"{anchor=center,rotate=-90}]
	\ar[draw=none, from=0, to=1, "\dashv"{anchor=center,rotate=-90}]
\end{tikzcd}\]
Given $X\in \X$, we construct the object $X^{hG}$ by
noting that when we set $Y=X$, $Z=\terminal$ and $f$ to be a terminal morphism
$\terminal_X:X \to \terminal$  in the above, the right-adjoint $(\terminal_X)_*$  
sends an object in $\X_{/X}$, {\em viz.} a
morphism $g: W \to X$, to an object in $\X$ whose points ({\em
  i.e.} its morphisms from $\terminal$) correspond to morphisms 
from
$\mathrm{id}: X \to X$ to $g: W \to X$, {\em i.e.} to sections of $g$.
Thus, it is reasonable to make the 
definition $X^{hG} := (\terminal_{BG})_\ast  (X/\!\!/G \to BG)$. As a check, the adjunction implies that
  $\X_{/BG}(BG,X/\!\!/G ) \simeq \X(\terminal,X^{hG})$, so points
  of $X^{hG}$ indeed correspond to sections of $X/\!\!/G \to
  BG$.

To get the morphism $X^{hG}\to X$, we need to use the basepoint of
$BG$, which we write explicitly as $b_G:\terminal\to BG$. The unit of the
adjunction $b_G^\ast\dashv (b_G)_\ast$ defines a functor $\Delta^1\times
\X_{/BG}\to \X_{/BG}$ (or, in other words, a natural transformation of
functors $\X_{/BG}\to \X_{/BG}$) sending the morphism
$(\{0,1\},\mathrm{id}_{X/\!\!/G})$ to $X/\!\!/G\to (b_G)_\ast b_G^\ast (X/\!\!/G)$. By
postcomposing this with $(\terminal_{BG})_\ast:\X_{/BG} \to \X$ we get a
functor $\Delta^1\times \X_{/BG} \to \X$ which sends the morphism
$(\{0,1\},\mathrm{id}_{X/\!\!/G})$ to a morphism $X^{hG}\to X$. 
(Alternatively, we may think of this as a functor $\X_{/BG}\to
\Fun(\Delta^1,\X)$
sending $X/\!\!/G \to BG$ to $X^{hG}\to X$.)

The construction of the morphism $X^{hG}\to X$ enjoys the following
functoriality properties.
Firstly, given $G$-actions on two objects $X$ and $Y$ and a morphism $f:X/\!\!/G\to Y/\!\!/G$ in $\X_{/BG}$ we get a commutative square
 \[
 \begin{tikzcd}[column sep=2cm]
 	X^{hG}\ar[d]\ar[r,"{(\terminal_{BG})_\ast f}"] & Y^{hG} \ar[d]\\
 	X \ar[r,swap,"{b^\ast_G f}"]&Y\mathrlap{,}
 \end{tikzcd}
\]
where $b^\ast_G f$ is a $G$-equivariant morphism from $X$ to $Y$.

Secondly, given a 2-simplex
\[	\begin{tikzcd}
		& BG\ar[dr,"{h}"] & \\
		 \terminal\ar[ur,"{b_G}"]\ar[rr,"{b_H}",swap] & &BH
	\end{tikzcd}
\]
      in $\X$ and a $X/\!\!/H\to BH \in \X_{/BH}$ representing an $H$-action on $X$, we may form $X /\!\!/G = h^*(X /\!\!/H) \to BG \in \X_{/BG}$ as the pullback of $X/\!\!/H\to BH$ along $h : BG \to BH$, which represents the restriction of the $H$-action on $X$ to a $G$-action, and we get another 2-simplex
      \begin{equation} \label{eq:hfpOverX}
	\begin{tikzcd}
		& X^{hG}\ar{dr} & \\
		 X^{hH}\ar{ur}\ar{rr}& & X\mathrlap{.}
	\end{tikzcd}
      \end{equation}
      in $\X$, expressing functoriality with respect to the group objects. The construction goes as follows. Letting $\eta_h$,
        $\eta_{b_H},$ and $\eta_{b_G}$ denote, respectively, the units
        of the adjunctions $h^\ast \dashv h_\ast$, $b_H^\ast \dashv
        (b_H)_\ast$, and $b_G^\ast \dashv (b_G)_\ast$, we get a
        functor $\Delta^2\times \X_{/BH}\to \X_{/BH}$ (or, in other
        words, a modification) that, regarded as a 2-simplex in
        $\Fun(\X_{/BH},\X_{/BH})$, takes the form
   \twosimplex{\mathrm{id}}{h_\ast h^\ast}{(b_H)_\ast
     b_H^\ast\mathrlap{,}}{}{\eta_h}{\mathrm{id}_{h_\ast}
     \eta_{b_G}\mathrm{id}_{h^\ast}}{\eta_{b_H}}
 Postcomposing with $(\terminal_{BH})_\ast:\X_{/BH}\to \X$ gives a functor
    $\Delta^2\times \X_{/BH}\to \X$ which, evaluated at $X/\!\!/H \to BH \in
 \X_{/BH}$, returns the 2-simplex in Eq.~\ref{eq:hfpOverX}.

Before going further, let us pause to discuss how these notions
  should be applied to physics. The discussion will be general, but we
  will use the familiar example of an Lie group symmetry of quantum mechanics, discussed
  in detail in \S\ref{sec:UnitQM}, to illustrate.

The
idea will be that the object $X$ represents a collection of QFTs with
some specified structure, with a particular QFT given by a point $x: \terminal
\to X$. In the case of quantum mechanics, an $x$ is specified by
a unitary representation of $\ssR$, or equivalently a
hamiltonian.

The group object $G$ represents a possible 
symmetry group of a QFT; to see if it really is a symmetry, however, we must
first specify more data, namely an action of $G$ on $X$. In quantum
mechanics, for example, the trivial action of 
on $X$ will give rise to symmetries acting via a smooth unitary
representation on the state space, while non-trivial actions will give
rise to smooth representations 
that are twisted in some way, such as projective or antiunitary representations.

Having specified the data of a group object $G$ and an action of
  it on $X$, we are now in a position to ask if there
is a  homotopy
fixed point through $x$; if there is, we say that the QFT $x$ admits a
$G$-symmetry (the action of $G$ is left implicit in the notation).
The associated homotopy fixed point is to be interpreted as the QFT $x$
equipped with the symmetry structure; we call it a $G$-symmetric
QFT. In quantum mechanics, for example, with the trivial action of $\mathsf{G}$, a
homotopy fixed point is given by a unitary representation of $\mathsf{G} \times
\ssR$ that restricts along $\ssR \hookrightarrow \mathsf{G} \times
\ssR$ to the representation specifying $x$ (equivalently, the
unitary operators assigned to each $\mathsf{g} \in \mathsf{G}$ commute with the hamiltonian).
We stress again that a symmetry of a QFT is a structure rather than a
property, because (for fixed $G$ and fixed action thereof) there may
be multiple homotopy fixed points through $x$. In quantum mechanics,
for example, there may be multiple inequivalent representations of $\mathsf{G} \times
\ssR$ that restrict suitably.
The $G$-symmetric
QFTs in $X$ assemble themselves into the object
$X^{hG}$ and the morphism $X^{hG} \to X$ sends a $G$-symmetric QFT
to its underlying QFT, forgetting the $G$-symmetry structure.

Since our
construction is functorial, it follows that group actions that
are isomorphic in $\X_{/BG}$ lead to isomorphic $X^{hG} \to X$ in
$\Fun(\Delta^1,\X)$, so are physically
equivalent.

It is to be stressed here that $G$ represents a symmetry
group that is extrinsic rather than intrinsic to $X$ (in particular, $G$ can be
chosen freely). There also exists a notion of an
intrinsic symmetry group of a particular QFT $x: \terminal
\to X$, given by the group object $\Omega_x X$ of loops based at
$x$. Furthermore, there is even a notion \cite{NSS14} of 
an intrinsic symmetry group $\mathrm{Aut}\, X$ of the
object $X$ of QFTs, at least when the object $X$, or rather the
morphism $X \to \terminal$ is suitably compact.
In practice, actually giving an explicit description of
$\mathrm{Aut}\, X$ starting from an explicit description of $X$ is a difficult
task. Moreover, the interplay between $\mathrm{Aut}\, X$ and $\Omega_x
X$ is subtle. Passing to an extrinsic notion of symmetry is convenient
in that it allows us to avoid having to do
any of this.

To remove some of the mystery of these assertions, consider the simple example of a
quantum mechanical system whose dynamics is invariant under
rotations of 3-dimensional space. There is an intrinsic $\mathsf{SO}(3)$
symmetry, and one possibility is for the states of the system to form a
representation thereof. But the fact that $\mathrm{Aut}\, X$ is
non-trivial for the object of quantum-mechanical theories (cf. \S\ref{sec:NonUnitQM})
means that one can also have projective representations of
$\mathsf{SO}(3)$.

Carrying on our discussion of the physics interpretation, an important remark is that
physicists are presumably interested only in objects in $X$ that admit
a point (since otherwise there is no QFT in the `object of QFTs') and
only in $G$-actions on $X$ that admit a
homotopy fixed point (since otherwise there will be no QFT that admits a $G$-symmetry under the given action). Moreover, given a $G$ and an $X$, it
  would be desirable to have a means of considering, in one go, all possible actions admitting
  fixed points, along with their homotopy fixed points.
We can achieve this by considering 
not just the topos $\X_{/BG}$ in which an
object corresponds to a $G$-action, but also the related category of pointed
objects. 

Given any topos $\X$, we define the corresponding category of pointed
objects $\X_\terminal$ as the slice category $\X_{\terminal/}$.\footnote{Note that $\X_\terminal$ is not a topos (unless it is trivial), so the usual theorems for topoi do not apply.} We claim that
$(\X_{/BG})_\terminal$ is a category whose objects are equivalent to homotopy
fixed points of $G$-actions. Indeed, the terminal object $\terminal$ in $\X_{/BG}$ is
$\mathrm{id}_{BG}$, so a section is manifestly an object in
$(\X_{/BG})_\terminal$. 
To recover the $G$-action (which is necessarily one that admits a
homotopy fixed point), we simply follow the functor
$(\X_{/BG})_\terminal \to \X_{/BG}$ that forgets the point, while to recover
the underlying fixed point $\terminal \to X$, we 
follow the functor $(b^\ast_G)_\terminal:(\X_{/BG})_\terminal \to \X_\terminal$
induced by $b^*_G: \X_{/BG} \to \X$ ({\em i.e.} pullback along the
basepoint $b_G: \terminal \to BG$). 

As well as recovering the individual homotopy fixed point
corresponding to an object in $(\X_{/BG})_\terminal$, we can recover the space
(not smooth space) of homotopy fixed points of a given action as the
 fibre of the functor $(\X_{/BG})_\terminal\to \X_{/BG}$ above a
 $G$-action $X/\!\!/G\to BG$. Indeed, this is equivalent via \S2.1.2 of \cite{Lur09}
 to the space $\X(\terminal,X^{hG})$.

It is sometimes
useful to exploit
a relation, described in \pagetarget{page:Actions-HFPs}{Appendix}~\ref{proof:Actions-HFPs}, between
the category $(\X_{/BG})_\terminal$ and
the topos $\Fun(\Delta^1,\X)_{/\mathrm{id}_{BG}}$. Since the latter is evidently
equivalent to the topos of actions on objects given by arrows in $\X$ by the group
given by looping the arrow $\mathrm{id}_{BG}$,
the relation shows that the problem of
finding group actions with homotopy fixed points
is really no different to the problem of finding bare group actions --
it simply takes place in a different topos. The problem of finding
bare group actions will be discussed in \S\ref{sec:findact}.
\subsection{Orbits and stabilizers \label{sec:orbstab}}

For better or
for worse, some physicists seem to be not so much 
interested in (or perhaps are
entirely unaware of) the space $X$ of
QFTs, but rather tend to focus on their pet
theory $x: \terminal \to X$ within it. We shall have frequent need to refer to
such a physicist in the sequel, so we shall give
them the allegorical moniker Simplicio.\footnote{We hope that no physicist (Galileo included)
 will be unduly offended by our doing so.}
Now Simplicio, parochial as they are, is 
presumably not (or not yet, at least) interested in
arbitrary $G$-actions on $X$, nor even $G$-actions on $X$ admitting (or
equipped
with) a homotopy fixed point as above, but rather only on $G$-actions
on $X$ admitting (or equipped
with) a homotopy fixed point through their favourite
$x$.

For Simplicio's benefit, it is convenient to formulate certain
notions associated to a point $x: \terminal \to X$. Firstly, by the orthogonal
factorization property we have a Postnikov decomposition
({\em cf}. the Whitehead tower of a space) 
\[ x: \terminal \to \dots \to \tge{m}{x}{X}\to \dots \to \tge{2}{x}{X} \to  \tge{1}{x}{X} \to X,\]
where $\tge{m}{x}{X}$ is the $(m-1)$-connected cover of
$(X,x)$ defined by Eq.~\ref{eq:ConnectedCover}. To see this,
  observe that since $\tge{m}{x}{X}$ is $(m-1)$-connected,  \cite[Prop. 6.5.1.20]{Lur09} tells us that the  
morphism $\terminal \to \tge{m}{x}{X}$ is $(m-2)$-connected. The
morphism $\tge{m}{x}{X} \to X$ appears at the top of
the diagram in Eq.~\ref{eq:ConnectedCover}, and is the pullback of
$\ast \to \tau_{\le m-1}X$. This latter morphism is $(m-2)$-truncated, a
consequence of~\cite[Lem. 5.5.6.15]{Lur09} and the fact that
  looping increases truncatedness. Since truncatedness is
  preserved by pullback, $\tge{m}{x}{X} \to X$ is $(m-2)$-truncated
  too.

Although it is not true in a
general topos, in the topos $\Sm$ of smooth spaces being 0-connected implies\footnote{The
  converse does not hold, a counterexample being the classifying space
  of bundles with connections introduced in \S\ref{sec:smoothspaces}.} being path-connected,
{\em i.e.} any two points are homotopic to one
another, meaning
physically that the theories are equivalent.
Thus we can sensibly
interpret, in our framework, what Simplicio
does as studying homotopy fixed points of group actions on the
connected component $\tge{1}{x}{X}$, rather than $X$ itself,
(with the QFT $x$ replaced by the point $\terminal \to
\tge{1}{x}{X}$ defined by the decomposition above.) This
modus operandi, which has both features and a bug, will be of
some significance in the sequel. The first feature,  as we shall see in the next Subsection, is that all symmetries of $\tge{1}{x}{X}$ ({\em i.e.} all possible homotopy fixed
points of all possible actions) can be found, as we shall see, in a way analogous to what is done in spaces. The second feature, is that every $G$-action on $X$
equipped with a homotopy fixed point, $s_x$, through $x$ descends to
a $G$-action on $\tge{m}{x}{X}$, for each $m$, so that
  Simplicio will not miss any of the possible symmetries of the QFT $x$. To see this, consider the Postnikov decomposition of $s_x$,
\[\begin{tikzcd}
    \terminal \ar[d] \arrow[bend right=50, swap]{ddddd}{x} \ar[r, two heads] & BG \ar[d] \arrow[bend left=50]{ddddd}{s_x}\\
    \vdots \ar[d] & \vdots \ar[d]\\
    \tge{m}{x}{X} \ar[d] \ar[r,two heads]& \tge{m}{s_x}{X/\!\!/G} \ar[d]\\
    \vdots \ar[d]& \vdots \ar[d]\\
    \tge{1}{x}{X} \ar[d] \ar[r,two heads]& \tge{1}{s_x}{X/\!\!/G} \ar[d]\\
    X \ar[d] \ar[r,two heads]& {X/\!\!/G} \ar[d] \\
	 \terminal \ar[r,two heads]& BG\mathrlap{,}
\end{tikzcd}\]
in which all squares are cartesian. Thus $\tge{m}{s_x}{X/\!\!/G}\to
BG$ defines a $G$-action on $\tge{m}{x}{X}$. 

The bug is that not every
$G$-action on $\tge{1}{x}{X}$ admitting a homotopy fixed
  point extends to a $G$-action on $X$.
Thus we exhibit a new
kind of anomaly that is associated purely to global symmetries,
arising because there are group actions of individual QFTs that do not
respect the smooth structure of the whole space of QFTs. We call these
anomalies `smoothness anomalies'. As we shall see in \S\ref{sec:UnitQM}, an example arises even when $X$
corresponds to the smooth space of unitary quantum mechanical
theories, so these smoothness anomalies seem likely to play a r\^{o}le
in the real world.

The other notions that we wish to formulate generalize the
usual notions of the orbit and stabilizer of a group acting on a set.
Given a point $x:\terminal\to X$ and an action on $X$
by a group object $G$, we can form the morphism $[x]:\terminal \to X/\!\!/G$ obtained by composing $x$ with the inclusion $X \twoheadrightarrow X/\!\!/G$
of the fibre. We call $[x]$ the stabilizer morphism of $x$. The Postnikov decomposition of the stabilizer morphism for each $m\geq 1$ gives the pointed, $(m-1)$-connected
object $\tge{m}{[x]}{X/\!\!/G}$, whose $m$th looping $\Omega^m\tge{m}{[x]}{X/\!\!/G}$ we call
the $(m-1)$-stabilizer of $x$ under the $G$-action. In particular,
$\Omega\tge{1}{[x]}{X/\!\!/G}$ generalizes the usual notion of the stabilizer subgroup at
a point of
a group acting on a set.
 Pulling the stabilizer morphism $[x]$ back along $X\to  X/\!\!/G$ we
 obtain a commutative diagram of the form
 \[\begin{tikzcd}
    G  \ar[d,"{o_x}"] \ar[r,two heads] & \terminal \ar[d, "{[x]}"]\\
    X \ar[d] \ar[r,two heads]& {X/\!\!/G} \ar[d] \\
	 \terminal \ar[r,two heads]& BG\mathrlap{,}
\end{tikzcd}\]
in which both squares are cartesian; we call the resulting  
morphism $o_x:G\to X$ the orbit morphism, since it
generalizes the usual orbit map for a group acting on a set. Its
$(m-2)$-connected/$(m-2)$-truncated factorization for each
$m\geq 1$  defines an object $\mathrm{im}_m \, o_x$ which
we call the $(m-1)$-orbit of $x$ under the $G$-action. The object
$\mathrm{im}_1 \, o_x$ generalizes the usual notion of the orbit of a point in a
set under a
group action.\footnote{As an aside, the usual notions of
  transitive and free actions of group actions on sets are encoded, respectively, by
  $-1$-connectedness or $-1$-truncatedness of the orbit morphism,
  leading to obvious (though $x$-dependent) generalizations of $m$-transitive and $m$-free
  group actions in any topos, but we make no use of them here.}

Universality of the
pullbacks implies, furthermore, that we have a diagram of the form
\[\begin{tikzcd}
    G  \ar[d] \ar[r] & \terminal \ar[d]\\
    \vdots \ar[d] & \vdots \ar[d]\\
\mathrm{im}_{m} \, o_x\ar[d] \ar[r,two heads]& \tge{m}{[x]}{X/\!\!/G} \ar[d]\\
    \vdots \ar[d]& \vdots \ar[d]\\
    \mathrm{im}_{1}  \,o_x\ar[d] \ar[r,two heads]& \tge{1}{[x]}{X/\!\!/G} \ar[d]\\
    X \ar[d] \ar[r,two heads]& {X/\!\!/G} \ar[d] \\
	 \terminal \ar[r,two heads]& BG\mathrlap{,}
\end{tikzcd}\]
in which all squares are cartesian.

Let us now discuss what can be said about the $m$-orbits and
  $m$-stabilizers at $x$ when we have a homotopy fixed point at $x$,
  {\em i.e.} a section $s_x$ through $x$. A first observation is that
  when we have a homotopy fixed point through $x$, the orbit morphism $G \to X$
  factors through $x: \terminal \to X$, as can be seen by considering the diagram
\[\begin{tikzcd}
	{G} & \terminal \\
	\terminal & BG \\
	X & {X/\!\!/G} \\
	\terminal & BG
	\arrow[from=4-1, to=4-2]
	\arrow[from=3-1, to=3-2]
	\arrow[from=3-2, to=4-2]
	\arrow[from=3-1, to=4-1]
	\arrow["{s_x}", from=2-2, to=3-2]
	\arrow[from=1-2, to=2-2]
	\arrow["{[x]}", bend left=40, from=1-2, to=3-2]
	\arrow[from=1-1, to=1-2]
	\arrow["x"', from=2-1, to=3-1]
	\arrow[from=2-1, to=2-2]
	\arrow[from=1-1, to=2-1]
\end{tikzcd}\]
in which all squares are cartesian. (An action having an orbit morphism that factors through $x:\ast \to X$ does not necessarily imply the existence of a
  homotopy fixed point through $x$, however.)

The observation that the orbit morphism $G \to X$
  factors through $x:\terminal \to X$ when we have a homotopy fixed point $s_x$, naturally
  leads us to compare the Postnikov decompositions of the morphisms
  $s_x$ and $[x]$. The sources of these morphisms, respectively $BG$
  and $\terminal$ disagree, but their
  targets agree, so we can hope that their Postnikov decompositions
  agree far from the sources but close to the targets, and moreover
  that the region in
which they coincide will be larger the more $BG$ resembles $\terminal$, {\em i.e.}
 the more connected it is. This intuition leads to the
following theorem (\pagetarget{page:OrbitSpace}{Appendix}~\ref{proof:OrbitSpace}): when $BG$ is
$p$-connected and $0\le m\le p+1$, we have that $ \tge{m}{s_x}{X/\!\!/G}\simeq
 \tge{m}{[x]}{X/\!\!/G}$,  that $s_x:BG\to X/\!\!/G$ has a
  $(m-2)$-connected/$(m-2)$-truncated factorization $BG\to \tge{m}{[x]}{X/\!\!/G}\to
  X/\!\!/G$, where the morphism $\tge{m}{[x]}{X/\!\!/G}\to X/\!\!/G$ is as above, and
  that consequently  $\tge{m}{x}{X}\simeq \mathrm{im}_m \, o_x$.

The usefulness of this result for our purposes is that the Postnikov
decomposition of any section $BG \to X/\!\!/G$ depends, toward its far end,
only on its underlying fixed point $x$.
So if we try to reconstruct
  sections by piecing together their decompositions, part of our
  work can be done in a wholesale fashion.
Since $BG$ is always $0$-connected, we can make use of
this at least for $\tge{1}{x}{X}$, which, as we have argued, is Simplicio's
{\em domus}.

Even more is true when the object $X$ is path-connected. 
For then the morphism of spaces
$\X_{/BG}(BG,  \tge{1}{[x]}{X/\!\!/G}) \to \X_{/BG}(BG, X/\!\!/G)$, which formalizes the
construction of sections of $X/\!\!/G \to BG$ from sections of $ \tge{1}{[x]}{X/\!\!/G} \to
BG$ described above, is not only
$-1$-truncated, but is also $-1$-connected, as will be shown
  momentarily. Thus it is an equivalence and we can find the whole space of homotopy fixed points
of the $G$-action on path-connected $X$ from the
space of homotopy fixed points of the corresponding action on a single
$\tge{1}{x}{X}$. 
This will be of use when we consider group actions on the smooth
  space classifying principal bundles with connection.
      To see that the morphism $\X_{/BG}(BG, \tge{1}{[x]}{X/\!\!/G}) \to
    \X_{/BG}(BG, X/\!\!/G)$ is $-1$-connected, {\em i.e.} essentially
    surjective, we note that, since $X$ is path-connected, for any section $s$ in $\X_{/BG}(BG, X/\!\!/G)$, the composite $\terminal\to BG\xrightarrow{s}X/\!\!/G$ is homotopic to $[x]$. Thus, from the discussion above, $s$ factors through $ \tge{1}{[x]}{X/\!\!/G}$, and is consequently in the essential image of  $\X_{/BG}(BG,  \tge{1}{[x]}{X/\!\!/G}) \to
    \X_{/BG}(BG, X/\!\!/G)$.
\subsection{Finding group actions and homotopy fixed points \label{sec:findact}}
We now discuss the practical business of finding group actions and
homotopy fixed points.

We have already introduced the notions
of $n$-truncated and $n$-connected morphism, for $n \geq -2$.For QFTs in spacetime
dimension $d$ and typical choices of target space,\footnote{Namely,
  the target category should be $d$-discrete, meaning that it is
  equivalent to a category in which the $k$-morphisms for $k>d$ are
  identities. Other choices of target can lead to spaces of field theories that are not $d$-truncated; in these cases we 
expect that even-higher-form symmetries could occur.} the corresponding smooth space $X$ is 
$d$-truncated and we will see in physical examples
that $X$ sometimes enjoys various connectedness
properties as well. As for the extrinsic symmetry group object $G$, choosing $BG$ to
be $n$-connected amounts to insisting that the `symmetry be at least
$n$-form or higher', in the physics lingo, while we shall shortly
show that it can be taken to be $d$-truncated without loss of
generality. This nearly, but not quite, corresponds to the physics
lore that QFTs in dimension $d$ can have at most $(d-1)$-form
symmetries. In fact, we will see that $d$-form symmetries do occur, but they manifest themselves in a more subtle way.

These facts about truncatedness and connectedness motivate us
to study their interplay 
with the notions of fibre sequences and sections that we
used earlier to define group actions and homotopy fixed points.
As we shall see, they
dovetail nicely. As ever, we shall try to emphasize the physical
motivations underlying the mathematical results. All of those results
are known in homotopy theory, but the fact that they generalize to any
topos is sometimes new. 

Let us begin by showing that when $X$
is $n$-truncated, one can replace any action of a group on $X$ by an action of the
$n$-truncation of that group, without loss of generality, but that truncating
further implies a loss of generality.
This immediately implies the claim above that QFTs in spacetime dimension $d$
can have $d$-form symmetries, but not $(d+1)$-form symmetries. 

The first part follows immediately from the fact \cite[Cor. 3.62]{Ras20} that the functor $ \tau_{\leq n}
({\cal X}_{/B\tau_{\leq n}G}) \to \tau_{\leq n} ({\cal X}_{/BG})$ induced by pullback along $BG
  \to B\tau_{\leq n} G$ is an equivalence of categories. On objects this
  implies that, given an action of $G$ on an $n$-truncated $X$, we have a diagram 
\pastingHE{X}{X/\!\!/G}{\tau_{\leq n+1} (X/\!\!/G)}{\terminal}{BG}{B\tau_{\leq n} G\mathrlap{,}}
where the left cartesian square is the fibre sequence corresponding to the
given action and the right square is also cartesian. This tells us
that any $G$-action can be obtained from an action by the $n$th
truncation of $G$, $\tau_{\leq n} G$. Moreover, the object $X^{hG}$
is isomorphic over $X$ to the object $X^{h\tau_{\leq n}
  G}$, in the sense of Eq.~\ref{eq:hfpOverX}.

That truncating further implies a loss
of generality  follows by observing an ordinary group (a.k.a. a $0$-truncated group object in
spaces) can act non-trivially 
on a set (a.k.a. a $0$-truncated object in spaces).

So how do we recover the physics lore that one can at most have
$(d-1)$-form symmetries? As we have already remarked, Simplicio is interested in a single QFT $x: \terminal\to X$, and their study of symmetries
of $x$ amounts here to finding group actions on the $0$-connected object
$\tge{1}{x}{X}$. Moreover, Simplicio is interested in $G$-actions that
admit a homotopy fixed point through $x$. As such, they find themself in
the situation of the following theorem, which we prove in
\pagetarget{page:ConnectedTruncatedX}{Appendix}~\ref{proof:ConnectedTruncatedX}:
if $X$ is both $n$-truncated and $0$-connected, then every $G$-action on $X$ equipped with a homotopy fixed point is in the essential image of the functor $(\X_{/B\tau_{\leq n-1}G})_\terminal\to (\X_{/BG})_\terminal$, and thus every $G$-action on $X$ admitting a homotopy fixed point is in the essential image of 
  the functor $\X_{/B\tau_{\leq n-1}G}\to \X_{/BG}$.

A weaker statement holds when $X$ is not connected. Let $\eta:BG\to B\tau_{\leq n-1}G$ be the unit of the truncation adjunction and let $\iota:\tau_{\leq n}(\X_{B\tau_{\leq n-1}G})\to \X_{/B\tau_{\leq n-1}G}$ be the inclusion of truncated objects. Unlike the connected case, not every $G$-action on $X$ that admits a homotopy fixed point is in the essential image of $\eta^\ast$. Nor is $\eta^\ast \circ \iota$ an equivalence, as it is for $B\tau_{\leq n}G$. However,  for those $X/\!\!/G$ that are in the essential image of $\eta^\ast$, \emph{i.e.}, of the form $\eta^\ast(X/\!\!/\tau_{\leq n-1}G)$  (with $X/\!\!/\tau_{\leq n-1}G$ corresponding to a $\tau_{\leq n-1}G$ action on $X$), then $X^{hG}$ is equivalent to $X^{h\tau_{\leq n-1}G}$ over $X$. This follows since the counit of the adjunction 
  \[\begin{tikzcd}
	{\tau_{\leq n}(\X_{/B\tau_{\leq n-1}G})} & {\X_{/B\tau_{\leq n-1}G}} & {\X_{/BG}\mathrlap{,}}
	\arrow[""{name=0, anchor=center, inner sep=0}, "{\tau_{\leq n}}"', shift right=2, from=1-2, to=1-1]
	\arrow[""{name=1, anchor=center, inner sep=0}, "\iota"', shift right=2, from=1-1, to=1-2]
	\arrow[""{name=2, anchor=center, inner sep=0}, "{\eta_!}"', shift right=2, from=1-3, to=1-2]
	\arrow[""{name=3, anchor=center, inner sep=0}, "{\eta^*}"', shift right=2, from=1-2, to=1-3]
	\arrow["\dashv"{anchor=center, rotate=-90}, draw=none, from=0, to=1]
	\arrow["\dashv"{anchor=center, rotate=-90}, draw=none, from=2, to=3]
\end{tikzcd}\]
is an equivalence of functors. Since both $\eta^\ast$ and $\eta_\ast$ preserve $n$-truncated objects and morphisms, they induce an adjunction
  \[\begin{tikzcd}
	{\tau_{\leq n}(\X_{/B\tau_{\leq n-1}G})} & {\tau_{\leq n}(\X_{/BG})\mathrlap{,}} 	\arrow[""{name=0, anchor=center, inner sep=0}, "{\eta^\ast}"', shift right=2, from=1-2, to=1-1]
	\arrow[""{name=1, anchor=center, inner sep=0}, "\eta_\ast"', shift right=2, from=1-1, to=1-2]
	\arrow["\dashv"{anchor=center, rotate=-90}, draw=none, from=0, to=1]
\end{tikzcd}\]
in which $\eta^\ast$ is full and faithful. Thus, the unit of this adjunction is an equivalence. The image of the components of this unit under $(\terminal_{B\tau_{\leq n-1}G})_\ast$ appear in (\ref{eq:hfpOverX}), manifesting the equivalence of $X^{hG}$ and $X^{h\tau_{\leq n-1}G}$ over $X$ in this case.

The upshot is that Simplicio sees only $(d-1)$-form
symmetries, but only because of their myopic focus on individual
QFTs. Once one accepts that the smooth space $X$ of QFTs is not
connected (as it rarely is), $d$-form symmetries can occur.

We remark here that it does not seem reasonable
to dismiss these $d$-form symmetries as
irrelevant for physics. Indeed, starting
from some QFT $x$ in a smooth space $X$, one can often reach an inequivalent
QFT $y$ (which therefore belongs to a distinct connected component), by
means of a smooth change in the coupling
constants.
(That this is so will become
obvious when we describe the smooth space of quantum-mechanical
theories in \S\ref{sec:NonUnitQM}, but for now let us give an even simpler example: 
the
smooth space representing a manifold $M$ has a distinct connected component for every
distinct point in $M$.\footnote{Mathematically, this is best understood that there
are two distinct notions of connectedness in the topos of smooth spaces. A connected
manifold in ${\Sm}$ is disconnected in the sense of connectedness
described here, but is connected in the other sense \cite{Sch13}.}) So
$d$-form symmetries could play a r\^{o}le
in any
phenomena where a smooth family or ensemble of QFTs are considered,
such as Berry phases.

A related remark is that there is a significant difference between the
topos of spaces and a generic topos when it comes to connectedness. In
the former, every space is a coproduct ({\em i.e.} a
disjoint union) of connected spaces. So while there certainly exist
non-connected spaces, they are generated in a certain sense by
connected ones. In a
generic topos this is not so. (For the example of the smooth space
representing a manifold $M$ given above, the coproduct of the
connected components returns not $M$ with its original smooth structure, but the set underlying $M$
equipped with the smooth structure corresponding to the discrete
topology.)

This difference will cause headaches when we start trying to actually find
group actions on smooth spaces. Since actions correspond to fibre sequences
and since the Postnikov decomposition can be carried out fibrewise, for any action of a group object $G$ on an object $X$ we have an
associated diagram
\[\begin{tikzcd}
    X & {X/\!\!/G} \\
      \vdots & \vdots \\
    	 {\tau_{\leq m}X} & {\tau_{\leq m}X/\!\!/G} \\
	 {\tau_{\leq m-1}X} & {\tau_{\leq m-1}X/\!\!/G} \\
          \vdots & \vdots \\
	 \terminal & BG\mathrlap{,}
         \arrow[from=1-1, to=2-1]
         \arrow[from=2-1, to=3-1]
         \arrow[from=3-1, to=4-1]
         \arrow[from=4-1, to=5-1]
         \arrow[from=5-1, to=6-1]
          \arrow[from=1-2, to=2-2]
         \arrow[from=2-2, to=3-2]
         \arrow[from=3-2, to=4-2]
         \arrow[from=4-2, to=5-2]
         \arrow[from=5-2, to=6-2]
         \arrow[twoheadrightarrow,from=1-1, to=1-2]
         \arrow[twoheadrightarrow,from=2-1, to=2-2]
         \arrow[twoheadrightarrow,from=3-1, to=3-2]
         \arrow[twoheadrightarrow,from=4-1, to=4-2]
         \arrow[twoheadrightarrow,from=5-1, to=5-2]
           \arrow[twoheadrightarrow,from=6-1, to=6-2]
\end{tikzcd}\]
in which all squares are cartesian. The left vertical arrows denote the Postnikov decomposition of \(X \to \terminal\), whereas the right vertical arrows denote that of \(X/\!\!/G \to BG\). For the latter, we adopt the notation \(\tau_{\leq m}X/\!\!/G\), indicating that \(\tau_{\leq m}X/\!\!/G \to BG\) represents a \(G\)-action on \(\tau_{\leq m}X\). By pasting of pullbacks, and
noting that for physics the object $X$ will be $n$-truncated for some
$n$ (so we may as well take $G$ to be $n$-truncated as well), we see
that one
can find $G$-actions via a finite bootstrap algorithm, provided one can find
the individual cartesian squares of the form
\[\begin{tikzcd}
    Y \ar[d] \ar[r,dashed] & ? \ar[d,dashed]\\
    V \ar[r] & W\mathrlap{,}
  \end{tikzcd}\]
where the morphism $Y \to V$ comes as close as it can to being an
isomorphism, in that it is both $m$-truncated and $m-1$-connected, for
some $m$.

 As we show in \pagetarget{page:ConnectedActions}{Appendix}~\ref{proof:ConnectedActions}, one can do this provided one makes
 certain technical assumptions about the objects $X$, $X /\!\!/ G$ and the
 morphism $X \to X/\!\!/G$. These assumptions
involve
no loss of generality in the topos ${\cal S}$, where every object is
isomorphic to a coproduct of connected objects, but they severely
hamstring us in smooth spaces, because the smooth spaces $X$
representing QFTs rarely take such a form, and even if they do we are
only able to find the fibre sequences satisfying the other assumptions in
this way. So we cannot generally hope to be able to find all
$G$-actions in this way, as we
did for TQFTs in \cite{GRWTS23}.

Happily for Simplicio, the assumptions hold in any topos when $X$ is
connected and admits a point (as $\tge{1}{x}{X}$ surely is and
does). Indeed, suppose that in the 
process of constructing the possible $X/\!\!/G$ via Postnikov decomposition we
have succeeded in constructing the possible ${\tau_{\leq
    m-1}X/\!\!/G}$. The latter
object is also connected and also admits a point, so there exists an
effective
epimorphism $\terminal \twoheadrightarrow \tau_{\leq m-1} X$. To try
to construct the possible ${\tau_{\leq m}X/\!\!/G}$, we consider the
diagram
\[\begin{tikzcd}
  F_m  \ar[d] \ar[r,two heads] &	 {\tau_{\leq m}X} \ar[d] \ar[r,two heads,dashed] & {?} \ar[d,dashed]\\
\terminal \ar[r,two heads] &	 {\tau_{\leq m-1}X} \ar[r,two heads]& {\tau_{\leq m-1}X/\!\!/G}\mathrlap{,} 
\end{tikzcd}\]
  in which the left square is cartesian and in which we wish to find the
  possible right-hand cartesian squares (for each of which the entry indicated
  by `$?$'
  corresponds to a cromulent ${\tau_{\leq m}X/\!\!/G}$).
The reverse pasting lemma, \cite[Lem. 3.7.3]{ABFJ20}, shows that finding
the right-hand cartesian squares is equivalent to finding the cartesian
rectangles, which have the special form of a fibre sequence and so are
hopefully easier to find.\footnote{In particular, they
are classified by the space $\X(
\tau_{\leq m-1}X/\!\!/G,B\mathrm{Aut}\, F_m)$
\cite{NSS14}, which may be interpreted
\cite{NSS14} as degree-one cohomology of $\tau_{\leq m-1}X/\!\!/G$ with local coefficients
in $\Aut \, F_m$.}

We remark for later purposes that an $m$-truncated and $(m-1)$-connected object such
as $F_m$ with $m \geq 1$ is called an $m$-gerbe in \cite{Lur09}, or an
Eilenberg-Mac Lane object of degree $m$ if it is pointed. Since any
point in a connected object is homotopic to any other  in smooth
spaces, it will do us little harm to blur the distinction between the
two.

We now discuss the business of finding the homotopy fixed points
corresponding to a given group action. Again, we cannot expect this to
be an easy endeavour in a general topos and will have to rely on
applying {\em ad hoc} results to favourable cases. Our main useful
result exploits the notions of the $m$-orbit and $m$-stabilizer of a given
point $x: \terminal \to X$ of
$X$ that we introduced in \S\ref{sec:orbstab} and
simplifies the finding of the homotopy fixed points through
$x$.  Indeed, we saw in \S\ref{sec:orbstab} that when $BG$ is $(m-1)$-connected, any such
section factorizes as $s: BG \to \tge{m}{[x]}{X/\!\!/G} \to X/\!\!/G$ (and furthermore
this pulls back along $X \to X/\!\!/G$ to $x: \terminal \to \tge{m}{x}{X} \to X$). Thus,
instead of finding the sections of $X /\!\!/ G \to BG$ over $x$, we can attack the
problem of finding
sections of $\tge{m}{[x]}{X/\!\!/G}  \to BG$. Since the fibre of $\tge{m}{[x]}{X/\!\!/G}
\to BG$, {\em viz. } $\tge{m}{x}{X}$, is $(m-1)$-connected, finding its sections is
usually easier. We still have to find the morphism $\tge{m}{[x]}{X/\!\!/G} \to X/\!\!/G$, of
course.
We stress that this trick always works for $m=1$, so is of general
applicability.

To introduce our final box of tricks, we recall that we have already seen that finding fibre sequences  becomes simpler when $X$
is connected and admits a point. We close this section by
discussing further simplifications that occur in finding both fibre sequences  and
sections when $X$ is $k$-connected, for large enough $k$. 

To discuss this will take us into the
topos-theoretic version of stable homotopy theory. This is most
conveniently done in the category $\X_\terminal$ of pointed objects. Since for
physics we are interested in $X$ admitting a point ({\em i.e.} a
QFT), and since $X$ will be connected here all points are homotopic
to one another, there is no loss of generality in doing so.

We have already described the looping operation $\Omega$ in \S\ref{sec:findact}
on a pointed object as the pullback of the point along itself.
We now regard 
this as a functor $\X_\terminal \to \X_\terminal$. It is right-adjoint to the
suspension functor $\Sigma: \X_\terminal \to \X_\terminal$, which on an object in $\X_\terminal$ is given
by the pushout along the terminal morphism. When $X$ is pointed and $k$-connected, the unit
$X \to \Omega \Sigma  X$ (which is defined on objects by the
universality of the pullback defining the looping) is $2k$-connected as a morphism in $\X$ (this generalizes the
Blakers-Massey theorem to any topos \cite{ABFJ20}), as is 
the counit $\Sigma \Omega X \to X$. 

Letting $\X_\terminal^{n,k}$ denote the full subcategory of $\X_\terminal$ on the
pointed, $n$-truncated and $k$-connected objects of $\X$, it follows that we induce an
adjunction \begin{tikzcd}
	\tau_{\leq n+1} \Sigma:\X_\terminal^{n,k} &  \X_\terminal^{n+1,k+1}:\Omega
	\arrow[""{name=1, anchor=center, inner sep=0}, shift left=1.5, from=1-1, to=1-2]
	\arrow[""{name=2, anchor=center, inner sep=0}, shift left=1.5, from=1-2, to=1-1]
	\arrow["\dashv"{anchor=center, rotate=-90}, draw=none, from=1, to=2]
\end{tikzcd}. We claim that
the unit and counit of this adjunction are $n$-truncated and
$2k$-connected as morphisms in $\X$. In particular, the adjunction defines an equivalence
of categories for $2k \ge n$ (see \cite[Thm. 5.1.2]{MS16} for a similar result).

To see that this is so, observe that the unit, considered as a morphism in $\X$, is a morphism of $n$-truncated objects, so is
itself $n$-truncated. It factors, moreover, as
\(
 	X\to \Omega \Sigma X \to  \Omega \tau_{\leq n+1} \Sigma X.
     \)
 where we have seen that the first morphism is $2k$-connected and where the
 second morphism is $n$-connected as discussed in \S2.4 of \cite{BL21}. So either $2k < n$ and the
 unit is $n$-truncated and manifestly $2k$-connected, or $2k \ge n$ and the unit is
 $n$-truncated and $n$-connected, in which case it is an isomorphism
 and so $2k$-connected as well. It is, therefore, also an isomorphism in $\X_\terminal$. A similar argument holds for the
 counit. 

 The power of this result, for our purposes, is that when $X$ is
 sufficiently connected, we can find the group actions on it and
 sections thereof from the group actions and sections of either its
 looping or its suspension
 and that this
 process can be iterated provided we remain in the stable range $2k
 \geq n$. To be explicit, let us discuss the most favourable case where $X$ is an Eilenberg-Mac Lane object of degree $n$ in $\X$, meaning a pointed, $n$-truncated and
 $(n-1)$-connected object in $\X$.
 
We have already seen in \S\ref{sec:fix}
 that a group action of $G$ on any $X$ with a homotopy fixed point
 corresponds to an object $BG \to X/\!\!/G$ in $(\X_{/BG})_\terminal$. It will
 comes as no great surprise that when $X$ is an
 Eilenberg-Mac Lane object of degree $n$ in $\X$, then  $BG \to X/\!\!/G$
 is itself an Eilenberg-Mac Lane object in $\X_{/BG}$. Indeed, the
 morphism $X \to \terminal$ in $\X$ is $n$-truncated and
 $(n-1)$-connected, so the same is true of any morphism $X/\!\!/G \to BG$
 defining a group action via \S\ref{sec:gpd}. But this means that $X/\!\!/G \to
 BG$ is an $n$-truncated and
 $(n-1)$-connected object in $\X_{/BG}$ and thus, as above, the section  $BG \to X/\!\!/G$ defines an Eilenberg-Mac Lane object in $\X_{/BG}$.
 
Letting $\Omega_{BG}:(\X_{/BG})_\terminal\to (\X_{/BG})_\terminal$ and
$\Sigma_{BG}:(\X_{/BG})_\terminal\to (\X_{/BG})_\terminal$ denote looping and
suspension in the topos $\X_{/BG}$, then since $b^*_G$ preserves colimits and limits we have that
 $b^*_G\Omega_{BG} (BG\to X/\!\!/G)$ is equivalent to $\Omega (\terminal\to X)$, where $(\terminal\to X)\simeq b^\ast_G(BG\to X/\!\!/G)$ is an object in $\X_\terminal$,
  and
 likewise $b^*_G\tau_{\leq n+1} \Sigma_{BG} (BG\to X/\!\!/G)$ is equivalent to $\tau_{\leq n+1} \Sigma (\terminal\to
 X)$. So looping sends a section of an action on $X$ to a section of
 an action on $\Omega X$, \&c.
\section{The topos of smooth spaces \label{sec:smoothspaces}}
According to the geometric version of the cobordism hypothesis \cite{GP21}, QFTs
assemble themselves into objects in the topos ${\Sm}$ of smooth spaces, given by sheaves on the site $\Cart$ whose objects are nonnegative integers, with $p$ representing $\ssR^p$, whose morphisms are smooth maps 
$\phi: \ssR^p \to \ssR^q$, and whose covers are given by
good open covers (meaning a cover for which every finite non-empty
intersection is diffeomorphic to some $\ssR^p$).\footnote{The topos
  ${\Sm}$ is, in fact, equivalent to the topos of sheaves on the site
  of manifolds with open covers~\cite{Bun22}.}
A sheaf here is a presheaf,
that is a functor $\Cart^{op} \to {\cal S}$ satisfying descent
(See Defn.~6.2.2.6 of~\cite{Lur09}).\footnote{We remark that, as in 1-categories, limits in
  sheaves can be computed in presheaves, which in turn can be
  computed objectwise.}

The data of an object in ${\Sm}$ are specified by a map $\Cart^{op}
\to {\cal S}$ of simplicial sets,
so consist of: a space for each $p$, a morphism of spaces for
each smooth map $\phi: \ssR^p \to \ssR^q$, a 2-morphism of spaces for
each pair $(\phi: \ssR^p \to \ssR^q, \phi^\prime: \ssR^q \to \ssR^r)$
of composable smooth maps, and so on.

The terminal object $\terminal$ in $\Sm$ corresponds to the locally constant sheaf with value
  $\ast$. Perhaps the next simplest objects in ${\Sm}$ correspond to smooth
manifolds. Given a smooth manifold $\mathsf{M}$,\footnote{We remind the
reader that we use sans serif typeface for classical notions in
differential geometry and serif typeface for the corresponding
notions in ${\Sm}$.} we define a
corresponding smooth space $M$ by assigning: to $p$, the 0-truncated
space corresponding to the
set $\mathrm{Map}(\ssR^p, \mathsf{M})$ of smooth maps from $\ssR^p$
to $\mathsf{M}$; to $\phi$, the morphism of spaces corresponding to the map
of sets $\mathrm{Map}(\ssR^q,
\mathsf{M}) \to \mathrm{Map}(\ssR^p, \mathsf{M}): \alpha \mapsto \alpha \circ
\phi$; and so on. The smooth space $M$ is 0-truncated, but
not 0-connected (unless $\mathsf{M}$ is a point):  the first follows since a sheaf is truncated precisely when its underlying presheaf is and truncatedness of presheaves can be tested objectwise, and the second then follows since $\tau_{\leq 0} M = M \neq \terminal$ if $\mathsf{M}$ is not a point.

The data of a morphism in ${\Sm}$ 
are specified by giving a map $\Delta^1 \times \Cart^{op}\to
{\cal S}$ of simplicial sets.
Given two manifolds $\mathsf{M}$ and $\mathsf{N}$ and a smooth map
$\mathsf{f}: \mathsf{M} \to \mathsf{N}$, there is a morphism  $f:M
\to N$ that restricts to $M$ and $N$, respectively, on the two
non-degenerate 0-simplices $\{0\},\{1\} \in \Delta^1_0$ and that sends the
1-simplex  $(\{0,1\},\phi) \in \Delta^1_1 \times \Cart^{op}_1$ to $\alpha
\mapsto \mathsf{f} \circ \alpha \circ \phi$. Every morphism
from $M$ to $N$ is homotopic to one of this form and in fact something
stronger is true ({\em cf.} Example 1.3.32 of~\cite{Sch13}):
the full
subcategory of ${\Sm}$ on such objects is equivalent
to the category given by the nerve of the familiar 1-category of manifolds and
smooth maps.

Just like any other topos, ${\Sm}$ has group objects and these are
equivalent via delooping to pointed connected smooth spaces. Perhaps the most basic group
object corresponds to a Lie group $\mathsf{G}$. We define the
associated delooping $BG$ by assigning: to $p$,
the
1-truncated space corresponding to the 1-groupoid with a single object
whose morphisms are smooth maps $\ssR^p \to \mathsf{G}$, composed by multiplying pointwise in $\mathsf{G}$; 
to $\phi$, the morphism of spaces that on morphisms in the
aforementioned 1-groupoid is given by
$\Map(\ssR^q,\mathsf{G})\to\Map(\ssR^p,\mathsf{G}): f\mapsto f\circ \phi$;
and so on. The smooth space $BG$ is pointed in an obvious way and the
fact that it is
$0$-connected and $1$-truncated follows from an objectwise
observation. The underlying object
$G$ is obtained via the pullback
\square{G}{\terminal}{\terminal}{BG\mathrlap{,}}{}{}{}{}
and an objectwise computation shows that it corresponds to the
manifold underlying the Lie group $\mathsf{G}$.\footnote{In more detail, this follows from the observation that
for any $0$-truncated group
$K$ in spaces we can work in the 2-category whose objects are
1-groupoids, whose 1-morphisms are 1-functors, and whose 2-morphisms are
1-natural transformations. We then have a commuting square
\square{K}{\terminal}{\terminal}{BK\mathrlap{,}}{}{}{}{}
in which $K$ is the 1-groupoid whose objects are elements of $K$ and
whose morphisms are identities and $BK$ is the 1-groupoid with one
object whose morphisms are elements of $K$. The filling 2-simplices
here are both given by the 1-natural transformation from the unique
1-functor $K \to BK$ to itself whose component on the object in $K$ given by
$k \in K$ is the morphism in $BK$ given by $k$. Since any cone from any
space $X$ into $\terminal \to BK \leftarrow \terminal$ is specified by a map of sets
from $\tau_{\leq 0} A$ to $K$, this square is a universal cone.}

We have defined $BG$ as the delooping of the group object $G$, but it
is useful to see explicitly that it is sensible to regard it as a
smooth space analogue of the usual classifying space of smooth
principal $\mathsf{G}$-bundles. Indeed, if it does play such a r\^{o}le, then
we expect the $G$-bundles over any smooth space $X$ to be classified
(as a space) by $\X(X,BG)$. But by the Yoneda lemma, when
$X=\mathbb{R}^p$, $\X(X,BG)$ is just the space assigned by $BG$ to the
object $\ssR^p$ on the site, which when $G$ corresponds to
$\mathsf{G}$ is the nerve of
the 1-groupoid with a single object and (iso)morphisms given by smooth maps
from $\ssR^p$ to $\mathsf{G}$. But this 1-groupoid is
equivalent to the 1-groupoid whose objects are principal
$\mathsf{G}$-bundles on $\ssR^p$ and whose (iso)morphisms are the
usual morphisms of principal $\mathsf{G}$-bundles: every such bundle is
trivializable and the morphisms of the trivial bundle are in bijective
correspondence with smooth maps $\ssR^p \to \mathsf{G}$. So
  indeed $BG$ plays the r\^{o}le of a classifying space, at least for
  $G$-bundles on $\ssR^p$.

For later purposes, it will be convenient to describe both the pointed
and unpointed versions of the
full subcategories on the smooth spaces corresponding to deloopings of
Lie groups. In
the pointed case, we have that the full subcategory of ${\Sm}_\terminal$
is equivalent to the nerve of the
usual 1-category of Lie groups and smooth homomorphisms. In the
unpointed case, the full subcategory of ${\Sm}$ (which for later
purposes we denote ${\cal L}$)
is equivalent to the category
in which a 0-simplex is a Lie
group $\mathsf{G}$, a 1-simplex is a smooth homomorphism $\mathsf{f}: \mathsf{G}\to
\mathsf{H}$, and a 2-simplex
\twosimplex{\mathsf{G}}{\mathsf{H}}{\mathsf{K}\mathrlap{,}}{\mathsf{k}}{\mathsf{f}}{\mathsf{f}^\prime}{\mathsf{f}^{\prime\prime}}
is an element $\mathsf{k}\in \mathsf{K}$  such that 
$\mathsf{f}^{\prime\prime}=\mathsf{k}(\mathsf{f}^\prime \circ \mathsf{f})\mathsf{k}^{-1}$.
Given 2-simplices specified by $\mathsf{k} \in \mathsf{K}$ and
$\mathsf{l}_{1,2,3} \in \mathsf{L}$ forming a tetrahedron 
\[\begin{tikzcd}[row sep={1.5cm,between origins}, column sep={2.5cm,between origins},labels={inner sep = .5pt}]
	& {\mathsf K}\ar[ddr,bend left=20,"{\mathsf f}"{pos=0.4}]& \\
	& { \mathsf H} \ar[dr,""{name=HL}]\ar[u] &\\
	{ \mathsf G}\ar[rr,bend right=20,""{name=GL}]\ar[uur,bend left=20,""{name=GK,pos=0.6}] \ar[ur] && { \mathsf L\mathrlap{,}}
	\ar["{\mathsf k}"{description, pos=0.3}, from=2-2,to=GK, Rightarrow,shorten >=2pt,shorten <=-4pt]
	\ar["{{\mathsf l}_1}"{description,pos=0.4}, from=2-2,to=GL, Rightarrow,shorten >=2pt,shorten <=0pt]
 	\ar["{{\mathsf l}_2}"{description,pos=0.5}, bend left=10, from=1-2,to=HL,Rightarrow,shorten >=2pt,shorten <=5pt]
	\ar["{{\mathsf l}_3}"{description,pos=0.5}, bend left=15, from=1-2,to=GL,Rightarrow,shorten >=10pt,shorten <=5pt,shift left=0.5]
\end{tikzcd}\]
we have a $3$-simplex
if and only if
$\mathsf{l}_1\mathsf{l}_2=\mathsf{l}_3\mathsf{f}(\mathsf{k})$, in which case the 3-simplex is unique.

The notions of manifold and Lie group are subsumed in the notion of a
Lie groupoid $\mathsf{L}$
\cite{Ehr59},
namely a 1-groupoid in which both the objects and morphisms form
manifolds, which we denote $\mathsf{L}_0$ and $\mathsf{L}_1$,
respectively,
such that the source $\mathsf{s}$ and target $\mathsf{t}$ maps are smooth surjective
submersions and the multiplication $\mathsf{m}$, the unit $\mathsf{u}$ and inverse $\mathsf{i}$
are all smooth maps.\footnote{We recover a manifold when $\mathsf{L}_1
  = \mathsf{L}_0$ and a Lie
  group when $\mathsf{L}_0$ is a point.} We
define a corresponding object $L$ in ${\Sm}$ by assigning: to
$p$, the 1-truncated space corresponding to the 1-groupoid $\Map(\ssR^p, \mathsf{L}_1)\rightrightarrows
\Map(\ssR^p, \mathsf{L}_0)$ (where the source, target, multiplication, unit, and inversion
are induced by $\mathsf{s},\mathsf{t},\mathsf{m}, \mathsf{u}$, and
$\mathsf{i}$ in the obvious way), and so on.

The simplest example of a Lie groupoid that is neither a manifold nor
a Lie group arises when we act with a Lie group on a manifold and will
be described in the next Subsection.

A description of the full subcategory of ${\Sm}$ on smooth spaces
  corresponding to Lie groupoids is somewhat complicated and not
  needed in what follows. Thus, we content ourselves with a
  description of how to construct \emph{some} morphisms in $\Sm$
  between Lie groupoids. In particular, given a 1-functor
  $\mathsf{f}:\mathsf{L}\to \mathsf{J}$ inducing smooth maps,
  $\mathsf{f}_0$ and $\mathsf{f}_1$, on the manifolds of objects and
  morphisms, respectively, we construct a 1-morphism $f:L\to J$ in
  ${\Sm}$ as follows. For each $p$, $f$ assigns to the
  morphism $(0\to 1,\mathrm{id}_{\ssR^p})$ in $\Delta^1\times \Cart$ the
  functor $L(\ssR^p)\to J(\ssR^p)$ sending  the object $\mathsf{h}\in
  \Map(\ssR^p,\mathsf{L}_0)$ of $L(\ssR^p)$ to the object
  $\mathsf{f}_0\circ \mathsf{h} \in \Map(\ssR^p,\mathsf{J}_0)$ of
  $J(\ssR^p)$ and the morphism  $\mathsf{r}\in \Map(\ssR^p,\mathsf{L}_1)$
  of $L(\ssR^p)$ to the morphism $\mathsf{f}_1\circ \mathsf{r} \in
  \Map(\ssR^p,\mathsf{J}_1)$ of $J(\ssR^p)$. The remaining data of $f$ are taken such that the homotopies involved (which in this
  case are natural transformations of functors) are in fact strict
  identities. Every morphism in between objects in $\Sm$ corresponding to manifolds is homotopic
  to such a morphism, and likewise for morphisms between objects corresponding to Lie groups. However, not every morphism in $\Sm$ from an object corresponding to a manifold to
   an object corresponding to a Lie group can be constructed in this way up to homotopy.

Finally, we wish to describe two other smooth spaces constructed from a Lie group $\mathsf{G}$ that will play
    an important r\^{o}le in examples. 
    These smooth spaces, which we will denote $B_\nabla G$ and $G\mathrm{Conn}(M)$ for a manifold $\mathsf{M}$, are close relatives of the smooth space $BG$ that classifies smooth
    principal $\mathsf{G}$-bundles on a smooth space $M$ corresponding
    to a manifold. 
    
 The first, $B_\nabla G$, 
     classifies smooth principal $\mathsf{G}$-bundles with the added
    data of a connection. For discussions in the literature of this
    and related objects, see \cite{Sch13,SW14,FH13}.
 Generalizing our
discussion of principal bundles without connection above,
we define $B_\nabla G$ to be the sheaf that assigns to $p$
the nerve of the 1-groupoid whose objects are smooth
principal $\mathsf{G}$-bundles with connection
on $\ssR^p$ and whose morphisms are
connection-preserving morphisms of principal
$\mathsf{G}$-bundles. (On $\phi: \ssR^p \to \ssR^q$, we get the usual pullbacks of
the bundles and connections.) Equivalently, since any such bundle on
$\ssR^p$ is trivializable and since a connection on the trivial
bundle on $\ssR^p$ is given by a $\mathfrak{g}$-valued 1-form $\mathsf{A}$ on
$\ssR^p$, we get the 1-groupoid in which an object is an $\mathsf{A}$
and an isomorphism from $\mathsf{A}$ to $\mathsf{A}^\prime$ is a
smooth map $\mathsf{g}: \ssR^p \to \mathsf{G}$ such that
$\mathsf{A}^\prime = \mathsf{g}^{-1} \mathsf{A} \mathsf{g} +
\mathsf{g}^{-1} d \mathsf{g}$.

The smooth space $B_\nabla G$ has some curious properties, which will be relevant for our
discussion of the symmetries of QFTs. Since its value on $p$ is
  the nerve of a $1$-groupoid, $B_\nabla G$ is $1$-truncated, just
like $BG$. Unlike $BG$, however, it is not
connected, so it cannot be the delooping of a
group object in ${\Sm}$. This follows from the fact that there is more than one equivalence class
(under bundle morphisms) of
connections on any $\ssR^p$ for $p > 0$. 
Nevertheless,
  like any connected smooth space (such as $BG$),  $B_\nabla G$ admits a point
$\terminal \to B_\nabla G$ and moreover any one point is homotopic to any other
one, so $B_\nabla G$ is path-connected. Indeed, via the Yoneda lemma, the space of points ${\Sm}(\terminal,B_\nabla G) \in {\cal S}$ is given by the value of
$B_\nabla G$ on the site object $\ssR^0$, and is connected, since
  every bundle over a point is trivializable and the trivial bundle
  over a point
  has a unique connection.

To give a description of the connected component of $B_\nabla G$
  at any point, we need to introduce the notion of the discretization
  $\esh X$ of a smooth space $X$. The functor $\Gamma:{\Sm} \to {\cal S}$
  that sends a smooth space $X$ to its space ${\Sm}(\terminal,X)$ of points
  has a left adjoint, $\mathrm{Disc}$, that sends a space $S$ to the
  smooth space given by the locally constant sheaf with value $S$.
 We define $\esh X =
  \mathrm{Disc} \, \Gamma
  X$ and note that the counit of the adjunction gives us a
  natural morphism $\esh X \to X$.
Because both $\mathrm{Disc}$ and
  $\Gamma$ preserve finite products this construction extends naturally to group objects in 
  ${\Sm}$~\cite[Prop. 3.5]{Bun20}. Furthermore, since $\mathrm{Disc}$ and
  $\Gamma$ also preserve colimits (they are both left-adjoints),
  they commute with delooping $B$~\cite[Prop. 3.5]{Bun20}. When $G$ is a group object in smooth spaces corresponding to a Lie group $\mathsf{G}$, we have that $\esh G$ corresponds to the discrete group, which we denote $\esh{\mathsf{G}}$, obtained by replacing the given smooth structure on $\mathsf{G}$ with the discrete one.

We now show that the connected component of $B_\nabla G$ at any point
  is equivalent to the smooth space $B\esh G$, which can be thought of
  as classifying principal $\mathsf{G}$-bundles with flat connection.
Since $BG^\delta$ is the delooping of a group object, $\terminal \to BG^\delta$ is $(-1)$-connected. There is a canonical morphism
  $B\esh{G}\to B_\nabla G$ that for each $p$ is the inclusion of the $1$-groupoid
  $\{$principal $\mathsf{G}$-bundles on $\ssR^p$ with flat connection$\}$ into the 1-groupoid
  $\{$principal $\mathsf{G}$-bundles on $\ssR^p$ with connection$\}$.
This is an inclusion of a path-connected component in spaces, so is a $(-1)$-truncated morphism. Truncatedness may be tested objectwise, so $B\esh{G}\to B_\nabla G$ is $-1$-truncated in $\Sm$. Since any
 point $\terminal\to B_\nabla G$ is unique up to homotopy, it must 
 factor through $B\esh G$, exhibiting $\terminal\to B\esh G\to B_\nabla G$ as the
($-1$-connected,$-1$-truncated) factorization of the point. Thus
 $B\esh G$ is indeed the connected component. There is, of course, also a canonical morphism $B_\nabla G \to BG$
that forgets the connection. 

The second smooth space $G\mathrm{Conn}(M)$~\cite{FRS16} classifies smooth families of
  principal bundles over $\mathsf{M}$. We define it analogously to
  $B_\nabla G$, except that we assign to $p$ the (nerve of the)
  1-groupoid whose objects are $G$-principal bundles on $\mathsf{M}
  \times \ssR^p$ together with a fibrewise connection with respect to
  the projection $\mathsf{M}
  \times \ssR^p \to \ssR^p$.

The space  $\Sm(\ast, G\mathrm{Conn}(M))$ is equivalent to the
  space of $\mathsf{G}$-principal bundles on $\mathsf{M}$, or $\Sm(M,
  B_\nabla G)$ if one wishes. That is, $G\mathrm{Conn}(M)$ and
 $(B_\nabla G)^M$ are equivalent as spaces ($\Gamma (G\mathrm{Conn}(M) )\simeq \Gamma (B_\nabla G)^M$) (but not as smooth spaces). The connected component of  $G\mathrm{Conn}(M)$ at each point corresponds to the smooth family of fibrewise connnection preserving gauge transformations of that bundle. As an example, when $\mathsf{M}$ is taken as the point $\ast$, we recover  $G\mathrm{Conn}(\ast) \simeq BG$ .

    \subsection{Group actions in smooth spaces\label{sec:grpSS}}
    \subsubsection{Lie groups acting on manifolds \label{sec:grpSSm}}
Having defined the smooth spaces corresponding to
the simplest kinds of (group) objects
${\Sm}$, namely manifolds and Lie groups, we can now define the simplest kind of
group action in ${\Sm}$, which corresponds to a smooth action of a
Lie group on a manifold. Namely, 
given a Lie group $\mathsf{G}$ acting on a smooth manifold
$\mathsf{M}$, we define the action Lie groupoid $\mathsf{M} /\!\!/
\mathsf{G}$ by taking $(\mathsf{M} /\!\!/
\mathsf{G})_0 = \mathsf{M} $, $(\mathsf{M} /\!\!/
\mathsf{G})_1 = \mathsf{M} \times
\mathsf{G}$, $\mathsf{s}$ to be the
projection on the left-hand factor, $\mathsf{t}$ to be the group action, and defining
$\mathsf{m}$, \&c,  using the corresponding maps in $\mathsf{G}$.
Denoting the corresponding smooth space by $M/\!\!/G$, we have a
morphism $M/\!\!/G \to BG$ in ${\Sm}$, induced by the morphism of Lie groupoids that
on $(\mathsf{M} /\!\!/
\mathsf{G})_1 = \mathsf{M} \times
\mathsf{G}$ is projection on the right-hand factor, whose fibre over
the basepoint is $M$, 
exhibiting a
$G$-action on $M$ in ${\Sm}$.

The observation that a Lie group action
on a manifold defines a group action in ${\Sm}$ can be elevated into the following
equivalence of categories \cite[Prop. 6.4.44]{Sch13}. Let $\mathsf{G}$ be a Lie group and 
$BG$ the corresponding delooped object in ${\Sm}$. 
Then the full subcategory of ${\Sm}_{/BG}$ on
objects whose fibre is a smooth space corresponding to a manifold
is equivalent to the nerve of the usual 1-category of
$\mathsf{G}$-manifolds.  

So every $G$-action on $M$ in ${\Sm}$ is equivalent to a
$\mathsf{G}$-action on the corresponding $\mathsf{M}$. The associated
object $M^{hG}$ and the morphism $M^{hG} \to M$ realise the fixed
point subset of the
$\mathsf{G}$-action on $\mathsf{M}$ as a smooth space along with its inclusion
in $\mathsf{M}$.  
 The points of $M^{hG}$ are given by ${\Sm}_{/BG}(BG, M/\!\!/G)$. But
 by the equivalence of categories given above, this is the set of
 equivariant maps $\ast \to \mathsf{M}$, with the trivial
 $\mathsf{G}$-action on the point manifold $\ast$, giving us the expected fixed points.
    \subsubsection{Lie groups acting on deloopings of Lie groups \label{sec:grpSSl}}
 The next simplest action to consider is that of a $G$ which
 corresponds to a Lie group $\mathsf{G}$ acting on $BK$, the delooping
 of another Lie group $\mathsf{K}$.
As we will now see, such an action, which will play a prominent r\^{o}le in
our examples, corresponds to a short exact sequence of Lie groups.

 To show it, recall that such an action corresponds to a cartesian square 
 \squareE{BK}{E}{\terminal}{BG\mathrlap{,}}{}{}{}{}
 for some smooth space $E$. {Let us begin by showing that $E$ is 
 $0$-connected, $1$-truncated, and admits a point, so that we can legitimately call it
 $BH$ for some $0$-truncated group object $H$.}
The morphism $BK\to \terminal$ is {$0$-connected and $1$-truncated} and
therefore so is $E\to BG$. Now $BG$ is also {$0$-connected and
  $1$-truncated}, so $E$ is $0$-connected
(from~\cite[Prop. 6.5.1.16]{Lur09}) and
$1$-truncated too (from~\cite[Lem. 5.5.6.14]{Lur09}). The smooth space $E$ certainly
admits a point, since $BK$ does, and thus is equivalent to a $BH$
for some $0$-truncated group $H$.

We next want to show that the $0$-truncated group $H$ is equivalent to
a Lie group.
Looping, we get a cartesian square  
\squareE{H}{G}{\terminal}{BK\mathrlap{,}}{p}{}{}{}
in $\Sm$ exhibiting $p: H \to G$ as an principal $K$-bundle (of smooth
spaces), as defined in \cite{NSS14}.
Then~\cite[Prop. 6.4.39]{Sch13} tells us that $H$ is equivalent as an object in $\Sm$ to a
(smooth space corresponding to a) manifold.\footnote{In fact,
   \cite[Prop. 6.4.39]{Sch13} tells us more, namely that $p: H \to G$
  is equivalent to a smooth principal $\mathsf{K}$-bundle in the usual
  sense, which allows for an alternative proof of what follows here.} The maps
corresponding to multiplication and inversion defined by the Segal
conditions on a manifold correspond to smooth maps of that manifold, so $H$
is equivalent to a group object corresponding to a Lie group $\mathsf{H}$. 

The group action $E\to BG$ can therefore be replaced by 
 \squareE{BK}{BH}{\terminal}{BG\mathrlap{,}}{Bi}{}{Bp}{}
 where $\mathsf{i}:\mathsf{K}\to \mathsf{H}$ and
 $\mathsf{p}:\mathsf{H}\to \mathsf{G}$ are smooth homomorphisms and
 so, in particular, have constant rank. Looping again, we get the cartesian square
 \[
 \begin{tikzcd}
 	K \ar[d,two heads] \ar[r, "i"] & H \ar[d,two heads,"p"] \\ 
 	\terminal \ar[r] & G\mathrlap{,}
 \end{tikzcd}
\]
in $\Sm$.
Since smooth spaces corresponding to manifolds form a full-subcategory of ${\Sm}$, it follows
that
\squareE{\mathsf{K}}{\ast}{\mathsf{H}}{\mathsf{G}\mathrlap{,}}{}{\mathsf{i}}{}{\mathsf{p}}
is a 1-pullback in the 1-category of manifolds. The smooth map $\mathsf{p}$ is
surjective and therefore, by the Global Rank Theorem of~\cite{Lee03}, also
a submersion. Furthermore, the image of $\mathsf{i}$, {\em viz.}
$\mathrm{ker}\, \mathsf{p}$, is embedded in $\mathsf{H}$, so the
restriction of $\mathsf{i}$ to it is both smooth~\cite{Lee03} and bijective, {\em ergo} (again by the Global
Rank Theorem) a diffeomorphism.

It follows that every action of $G$ on $BK$ corresponds to a Lie group extension
\[
	\mathsf{K}\xrightarrow{\mathsf{i}} \mathsf{H} \xrightarrow{\mathsf{p}} \mathsf{G}\mathrlap{,}
\]
as defined, {\em e.g.}, in \cite{DGRW20}. (Equivalently, we
have that $\mathsf{p}: \mathsf{H} \to \mathsf{G}$ is a smooth
principal $\mathsf{K}$-bundle in the usual sense with respect to the
action of $\mathsf{K}$ on $\mathsf{H}$ induced via $\mathsf{i}$ and
multiplication in $\mathsf{H}$.)
Again, this lifts to an equivalence between the full subcategory of
${\Sm}_{/BG}$ on objects whose fibre is a smooth space
corresponding to the delooping of a Lie group and the full subcategory
of $\mathcal{L}_{/\mathsf{G}}$ on those objects for which the smooth
homomorphism $\mathsf{p}:
\mathsf{H}\to \mathsf{G}$ is surjective (and, {\em ergo}, a submersion). Let us denote this last category as $\mathcal{L}^{\twoheadrightarrow}_{/\mathsf{G}}$. 
 
Unsurprisingly, the homotopy fixed points of such actions correspond
to sections of the surjection (or equivalently
smooth splittings of the smooth short exact sequences). To be precise, the category $({\cal L}_{/\mathsf G}^\twoheadrightarrow)_\terminal$ is equivalent to the
nerve of the $1$-category whose objects are pairs of Lie group
homomorphisms $(\mathsf{p}:\mathsf{H}\to \mathsf{G}, \mathsf{s}:\mathsf{G}\to
\mathsf{H})$ such that
$\mathsf{p}\circ \mathsf{s}$ is the identity homomorphism on $\mathsf{G}$, and
whose morphisms from  $(\mathsf{p}:\mathsf{H}\to \mathsf{G}, \mathsf{s}:\mathsf{G}\to
\mathsf{H})$ to
$(\mathsf{p}^\prime:\mathsf{H}^\prime\to \mathsf{G}, \mathsf{s}^\prime:\mathsf{G}\to
\mathsf{H}^\prime)$
are smooth homomorphisms $\mathsf{q}:\mathsf{H}\to \mathsf{H}^\prime$
such that $\mathsf{p}^\prime \circ  \mathsf{q}=\mathsf{p}$ and
$\mathsf{s}^\prime=\mathsf{q} \circ \mathsf{s}$. The forgetful functor
$({\cal L}_{/\mathsf G}^\twoheadrightarrow)_\terminal\to{\cal L}_{/\mathsf G}^\twoheadrightarrow$ sends
$(\mathsf{p}, \mathsf{s})$ to $\mathsf{p}$ and sends $\mathsf{q}$ to the
 2-simplex whose edges are $\mathsf{q}$, $\mathsf{p}$ and
$\mathsf{p}^\prime$ and whose filler is the identity of the group $\mathsf{G}$.

Making the Lie group $\mathsf{K}$ explicit, we have that the full
subcategory of $(\mathcal{L}_{/\mathsf{G}}^{\twoheadrightarrow})_\terminal$ on
objects $\mathsf{p}$ whose kernel is isomorphic to $\mathsf{K}$ is 
equivalent to the nerve of the $1$-category whose objects are smooth
homomorphisms $\mathsf{e}:\mathsf{G}\times \mathsf{K}\to \mathsf{K}$
that 
are automorphisms for all $g\in \mathsf{G}$, and whose morphisms from
$\mathsf{e}$ to $\mathsf{e}^\prime$ are smooth automorphisms $\mathsf{f}:\mathsf{K}\to
\mathsf{K}$ such that $\mathsf{f}(\mathsf{e}  (g,k))=\mathsf{e}^\prime(g, \mathsf{f}(k))$. The
equivalence sends $\mathsf{s}:\mathsf{G}\to \mathsf{H}$ to $\mathsf{e}: \mathsf{G} \times
\mathsf{K} \to \mathsf{K} : (g,k) \mapsto \mathsf{i}^{-1}\left(
  \mathsf{s}(g) \mathsf{i}(k) \mathsf{s}(g^{-1})\right)$. 

The
corresponding space of
homotopy fixed points, {\em i.e.} the fibre of the forgetful functor
$(\mathcal{L}_{/\mathsf{G}}^{\twoheadrightarrow})_1\to
\mathcal{L}_{/\mathsf{G}}^{\twoheadrightarrow}$, is then equivalent to
the nerve of the 1-groupoid whose objects are smooth maps $\mathsf{r}:\mathsf{G}\to \mathsf{K}$ such that
\begin{equation}\label{eq:TwistedReps}
	\mathsf{r}(\mathsf{g}_2) \mathsf{e} (\mathsf{g}_2,\mathsf{r}(\mathsf{g}_1))=\mathsf{r}(\mathsf{g}_2\mathsf{g}_1),
\end{equation}
 for all $\mathsf{g}_1,\mathsf{g}_2\in \mathsf{G}$, and whose morphisms are elements
 $\mathsf{k}\in \mathsf{K}$ such that 
\[
	\mathsf{r}^\prime(\mathsf{g})=\mathsf{k}\mathsf{r}(\mathsf{g}) \mathsf{e} (\mathsf{g},\mathsf{k}^{-1}),
\]
for all $\mathsf{g}\in \mathsf{G}$.

In some
  of our examples, the Lie group $\mathsf{K}$ will be
  connected. This
  enables us to give its group of automorphisms the structure of a Lie
  group,\footnote{More generally, this can be done if the group of
    components is finitely generated \cite{Hoc52}.} $\mathsf{Aut} \, \mathsf{K}$, such that we can equivalently
  view $\mathsf{e}$ as a smooth homomorphism $\mathsf{G} \to
  \mathsf{Aut} \, \mathsf{K}$ and $\mathsf{f}$ as a conjugation in $\mathsf{Aut} \, \mathsf{K}$. This construction goes as follows. The
  derivative defines a functor from the 1-category of Lie groups to
  the 1-category of Lie algebras. This functor induces a homomorphism from the group of
  automorphisms of any Lie group $\mathsf{K}$ to those of its Lie
  algebra $\frak{k}$. This map is a bijection if the Lie group $\mathsf{K}$ is
  simply-connected. It follows that the map is an injection if the
Lie group $\mathsf{K}$ is
  merely connected, since any such $\mathsf{K}$ is a quotient of
  its simply-connected universal cover $\tilde{\mathsf{K}}$ by some
  discrete subgroup $\mathsf{Z}$ of
  its centre. (To be explicit, the map lands on the automorphisms of
  $\tilde{\mathsf{K}}$ that induce automorphisms on $\mathsf{Z}$, or
  equivalently that map $\mathsf{Z}$ onto itself.) Now, giving the
  group of automorphisms of $\mathsf{K}$ the compact-open topology
  makes it into a closed subgroup of the group of automorphisms of
  $\mathfrak{k}$ which is itself a closed subgroup of the general
  linear group on $\mathfrak{k}$, with its usual Lie group
  structure. So we have a Lie group $\mathsf{Aut} \, \mathfrak{k}$ and
  a Lie subgroup $\mathsf{Aut} \, \mathsf{K}$ thereof.

  This construction also gives us a means to give an explicit
  description of $\mathsf{Aut} \, \mathsf{K}$, provided we can do the
  same for $\mathsf{Aut} \, \mathfrak{k}$. The latter task is
  complicated by the fact that the field $\ssR$ over which we are
  working is not algebraically closed. These difficulties can be
  circumvented if $\mathsf{K}$ is compact
  as well as connected, and a description of $\mathsf{Aut} \,
  \mathsf{K}$ in that case in terms of root diagrams can be found in
  {\em e.g.} \S IX.4.10 of~\cite{Bor08}.

  In other examples, $\mathsf{K}$ will be a discrete Lie group. Then
  we can replace $\mathsf{G}$ by its group of path-components, since the
  smooth maps $\mathsf{e}$ and $\mathsf{r}$ both have targets that are
  discrete, so factor through $\pi_0\mathsf{G}$. Thus the analysis of actions on $BK$ and
their homotopy fixed points collapses to the analysis of short exact
sequences of groups and their splittings, with $\mathsf{G}$ replaced by $\pi_0 \mathsf{G}$.
  
  With these results in hand, we are ready to consider some examples of group actions on smooth
spaces corresponding to QFTs.

    \subsubsection{Lie groups acting on non-connected smooth spaces}

In the examples, we will frequently encounter examples of smooth
spaces $X$ that are
not connected. It is useful to relate group actions on these to
actions on related connected smooth spaces. As shown in \S\ref{sec:orbstab}, one useful
result here is that every homotopy fixed point
through $x$ of a group action on $X$ descends to a homotopy fixed
point of a corresponding group action on the connected component
$\tge{1}{x}{X}$. Assuming one can find the latter actions and homotopy fixed
points, the remaining difficulties are, firstly, that one needs to
consider all $x$, and, secondly,  that one needs to determine which group
actions on $\tge{1}{x}{X}$ extend to the whole of $X$ (note that the homotopy
fixed points then necessarily
extend).

We are unable to surmount these difficulties for general $X$, but we
can make some progress in special cases. One in particular is where $X =B_\nabla K$ for some Lie group
$\mathsf{K}$. We discuss this example further in  Appendix~\ref{sec:grpSSc}. 
\subsection{Comparison with the topological case}

It is worth pausing here to compare with what happens in the case of
topological quantum field theories, discussed in~\cite{GRWTS23}.

The two cases are similar in that, in both cases, a
symmetry of a QFT is a homotopy fixed point of a group object acting
on some object in an underlying topos.

The only difference is that,
in the topological case, the underlying topos is that of spaces,
$\mathcal{S}$, whilst in the smooth case we have the topos $\Sm$ of
smooth spaces.

So how do the topoi  $\Sm$ and $\mathcal{S}$ differ? Given any topos
$\X$, as introduced above, there is a functor $\Gamma$ to $\mathcal{S}$, defined by
evaluating  $\X(\ast, -)$. This functor forms the left half of
an adjunction (whose right adjoint produces `discrete' objects in
$\X$) and the failure of this adjunction to be an equivalence
allows us to characterize, in a canonical way, the difference between $\X$ and
$\mathcal{S}$.

Now, the topos $\Sm$ is rather special, in that this
adjunction is much closer to being an equivalence than it is for a
generic topos $\X$. To wit, the adjunction extends to a
quadruple adjunction, making $\Sm$ into what
is called in \cite{Sch13} a cohesive topos. In particular, $\Gamma$ is also
a right adjoint, whose left adjoint produces `concrete' objects in
$\Sm$, of which manifolds (and more generally diffeological spaces)
are examples.

Via these adjunctions, one can often translate statements about smooth
symmetries of QFTs into statements about symmetries of TQFTs and {\em
  vice versa}, though some care is needed. Suppose, for example, that one finds smooth $G$-actions
on some smooth space $X$ in ${\Sm}$. Any such action leads to an
action in $\mathcal{S}$ of the underlying discrete group $\Gamma G$ on
$\Gamma X$, but not every action of  $\Gamma G$ on $\Gamma X$ can be
got in this way, since some may not act on $\Gamma X$ smoothly. So
even if we are considering topological quantum field theories, the
physical requirement that things be smooth ({\em i.e.} that we work in
${\Sm}$), leads to a more restrictive notion of symmetry. Explicitly,
as we shall see in \S\ref{sec:TFTs1dTake1}, some of the symmetries of TQFTs that we
found in~\cite{GRWTS23} do not correspond to smooth actions.

The change from $\mathcal{S}$ in~\cite{GRWTS23} to ${\Sm}$ here leads
to another important new feature, namely that of a smoothness
anomaly. To see that these cannot occur in the topological case, it is
enough to observe, as we did in \S\ref{sec:findact}, that in $\mathcal{S}$ every space is a coproduct of
connected components. Thus, every action (with or without a homotopy
fixed point)  on a connected component
of a space $X$ can be extended to an action on the whole of $X$ by
demanding that it act trivially on the other components.

\section{Examples from physics \label{sec:examples}}
\subsection{QFTs in $d=1$ with a smooth map to a target manifold \label{sec:QMtoM}}
For a first example, consider QFTs in
spacetime dimension $d=1$ in which the spacetimes in the bordism
category are equipped with both an orientation and a smooth map to a
manifold $\mathsf{M}$ and the target category is the nerve of the
1-category whose objects are complex vector spaces and whose morphisms
are linear maps. These can be intepreted physically as 1-dimensional
QFTs with a background field. In the invertible case,
these QFTs correspond to topological actions (see {\em e.g.} \cite{DGRW20}) in classical mechanics
formed using only the smooth map and the orientation 
(which allows us to do integrals, and so obtain actions by integrating
local lagrangians). An example is the coupling of a particle
(which in 3-d space would correspond to choosing $\mathsf{M} = \ssR^3$) to
an external magnetic field. We note however that here such actions are not
required to be real (a.k.a. hermititan),
since we have imposed no requirement of unitarity yet.\footnote{We will do when we
consider quantum mechanics in the next Subsection.}

According to \cite{BEP15} such QFTs correspond to smooth complex
vector bundles with connection. The explicit correspondence is as follows. The
0-dimensional manifold given by a point equipped with a smooth map to
$\mathsf{M}$ defines a point $\mathsf{m} \in \mathsf{M}$; to this we
associate the fibre of the vector bundle over $\mathsf{m}$. The 1-dimensional
manifold with boundary given by an interval equipped with a smooth map
to $\mathsf{M}$ defines a smooth path in $\mathsf{M}$; to this we
associate the parallel transport of the connection.

So what is the smooth space $X$ representing such QFTs? The results of \cite{BEP15} only give us a description of the underlying space, so there is a possible ambiguity. Let us describe two likely candidates.
Instead of
dealing with vector bundles with connection, we can consider (for each
$n \geq 0$) the
principal
$\mathsf{GL}_n$-bundles with connection to which they are associated. These, as we
have seen in \S \ref{sec:smoothspaces}, can be assembled either into
the smooth space $(B_\nabla GL_n)^M$, where by $Z^Y$ we denote the
smooth space ({\em i.e.} an object in ${\Sm}$) of morphisms from a
smooth space $Y$
to a smooth space $Z$, or into the smooth space $GL_n\mathrm{Conn}(M)$. Both of these have the same underlying space, so are consistent with the results of \cite{BEP15}. But now consider a path in either smooth space, which ought to correspond physically to a smooth 1-parameter family of theories. A path in $(B_\nabla GL_n)^M$ is precisely a morphism $M \times [0,1] \to B_\nabla GL_n$, {\em i.e.} a $GL_n$-principal bundle on $M \times [0,1]$ with connection. But on physical grounds, what we want is a $GL_n$-principal bundle on $M \times [0,1]$ with a notion of connection that is fibrewise over $[0,1]$, in that it only allows parallel transport along vectors tangent to $M$. This corresponds to a path in $GL_n\mathrm{Conn}(M)$.

As such, we shall proceed here with the assumption that $GL_n\mathrm{Conn}(M)$ is the correct smooth space. Putting everything together, we
take $X$ to be
given by $\coprod_n GL_n\mathrm{Conn}(M)$. In Appendix~\ref{sec:TFTs1dTake3} we discuss what would happen if we had taken $B_\nabla GL_n$, and show the presence of a smoothness anomaly in that case.

\subsubsection{Smooth TQFTs in $d=1$  \label{sec:TFTs1dTake1}}

As a warm up, let us consider the case where $M$ is a point, such that $X \simeq \coprod_n  BGL_n$. The general arguments in \S
\ref{sec:grpSSl} thus show that, at least when the group $G$ acting
corresponds to a Lie group $\mathsf{G}$, the problem reduces to
finding smooth splittings of short exact sequences of Lie groups
extending $\mathsf{G}$ by $\mathsf{GL}_n$. The (splittable) short
exact sequences are given by
smooth homorphisms $\mathsf{G} \to \mathsf{Aut}\, \mathsf{GL}_n$
as in \S\ref{sec:grpSSl} and the corresponding homotopy fixed
points are given by smooth maps $\mathsf{G} \to \mathsf{GL}_n$ as in Eq.~\ref{eq:TwistedReps}.

The physics interpretation is as follows. Such a TQFT is specified by
the dimension $n$ (necessarily finite) of the state space, which is a
mere vector space since we have no notion of unitarity yet. The
hamiltonian is zero, as one expects for a TQFT, corresponding to the
fact that the
linear operator on the state space associated to time evolution
along an interval is the identity operator. A symmetry of the
TQFT is na\"{\i}vely a linear (not unitary, since this
notion makes no sense) operator on the state space that
commutes with the hamiltonian, {\em i.e.} any linear operator. A $\mathsf{G}$-symmetry of the
TQFT is an assignment of a linear operator to each element  in $\mathsf{G}$. This assignment need not form a {\em bona fide} representation of
$\mathsf{G}$, because a symmetry need not send a TQFT to itself, but rather
could send it to an equivalent one. The twisted representations as
defined above give a complete description of the ways in which this
can happen.

\subsubsection{The general case}
Now let us return to the case of general $\mathsf{M}$. The smooth space $GL_n \mathrm{Conn}(M)$ is then neither connected nor path-connected, which
  complicates our analysis of symmetries.

To make things tractable, let us consider just a special class of
  group actions on $GL_n\mathrm{Conn}(M)$, namely those where a group
  object $G$ corresponding to a Lie group $\mathsf{G}$ acts on
  $GL_n\mathrm{Conn}(M)$ via the usual notion of a smooth action of
  $\mathsf{G}$ on $\mathsf{M}$. The induced action on $GL_n\mathrm{Conn}(M)$ is given by pulling smooth families of bundles back along a a smooth family of elements in $\mathsf{G}$.
Doing so
 allows us to make contact with the way in which physicists usually
 study symmetries of QFTs, namely by considering transformations of
 fields leaving some classical action (or strictly speaking, the
exponentiated action) invariant.

So let us  compare with what physicists do. To begin with, we ignore
symmetry considerations. Since $d=1$, the physics
here is that of a particle moving on the smooth manifold
$\mathsf{M}$. In the case where $n=1$ ({\em i.e.} invertible QFTs),
evaluating a QFT on the closed spacetime
manifold $\mathsf{S}^1$ gives a complex number, which represents a possible
contribution to the exponentiated action governing the dynamics. Since
this action depends only on the orientation and the smooth map to 
$\mathsf{M}$ with which $\mathsf{S}^1$ is equipped, a physicist calls
this a `topological' action term.

Now, it has long been known (for a recent discussion, see \cite{DGRW20}) that
in any spacetime dimension $d$, a source of such action terms is given by the differential
cohomology \cite{CS85} in degree $d+1$ of the manifold
$\mathsf{M}$, which refines the ordinary integral cohomology of
$\mathsf{M}$ with information about the differential forms on
$\mathsf{M}$. (Here, since no requirement of unitarity is imposed on
the physics side, we should consider differential forms valued in
$\ssC$ rather than $\ssR$.)

Since this differential cohomology in degree 2 is known to correspond
to the set of equivalence classes of principal $\ssC^\times$-bundles with connection on
$\mathsf{M}$ \cite{BB14}, which in turn is given by $\pi_0 \Sm (\ast, \bbC^\times\mathrm{Conn}(M))$, we see that in $d=1$, every such
topological action
term (that is compatible with locality in that it can be extended to
an invertible QFT) takes this form. So differential cohomology
classifies topological action terms in $d=1$.\footnote{In arbitrary $d>1$, it is
clear that differential cohomology cannot capture all topological action
terms; rather some differential refinement of cobordism is
needed.}

Now let us consider symmetry. What the physicist does is to consider a
smooth action of some Lie group $\mathsf{G}$ on $\mathsf{M}$ and
declares a physical theory to be symmetric if the physics action is
invariant under the $\mathsf{G}$ action. This requirement corresponds
to the mathematical notion of invariant differential cohomology defined in
\cite{DGRW20}. Our definition here of a $G$-symmetric QFT as a homotopy
fixed point of a group action on a smooth space of fully-extended QFTs
is rather more sophisticated, even in $d=1$, since it requires the
symmetry to be consistently defined not just on closed spacetime
manifolds, {\em i.e.} on disjoint unions of circles, but also
on compact spacetime manifolds with boundary, {\em i.e.} on disjoint unions of
circles and intervals. 

The upshot is that there is no reason to expect that the
classification of invariant topological action terms given by
invariant differential cohomology to agree with the more refined
classification based on fully-extended QFTs given here and indeed we shall now
see that, in general, they do not agree. In more detail, we expect that there
should be a map from the latter to the former obtained by
restricting a QFT to closed spacetimes, but we will show that this map
is neither surjective nor injective, in general. In physics terms, a
failure to be surjective represents a locality anomaly, in that a
symmetry of the classical action defined on closed spacetimes cannot
be extended consistently to spacetimes with boundary, while a failure
to be injective implies that there exist many ways to extend a
symmetry of the classical action on closed spacetimes to spacetimes
with boundary.

To construct the map, we begin by considering an action on an object
$X$ by a group
object $G$ in a general topos $\X$. We may pull this back along the
morphism $\esh G \to G$ to get an action of $\esh G$ on $X$, along
with a morphism $X^{hG} \to X^{h\esh G}$. 
Applying $\Gamma$ to the latter to pass to spaces (recall that $\Gamma: X \mapsto
{\X}(\terminal,X)$) and making use of the adjunctions
$\mathrm{Disc}\dashv \Gamma$ and those of the form $(\ast_X)^\ast\dashv (\ast_X)_\ast$
results in a morphism
\begin{equation}\label{eq:CohomPartial}
\Gamma (X^{hG}) \to (\Gamma X)^{h\Gamma G} ,
\end{equation}
in ${\cal S}$. 

Taking the Postnikov decomposition of the $\Gamma G$-action on $\Gamma X$
\pastingVE{\Gamma X}{(\Gamma X)/\!\!/\Gamma G}{\tau_{\leq 0} \Gamma X}{(\tau_{\leq 0}\Gamma X)/\!\!/\Gamma G}{\terminal}{B\Gamma G\mathrlap{,}}
implies
that we have a morphism $ ((\Gamma X)^{h\Gamma G}) \to (\tau_{\leq 0} \Gamma X)^{h\Gamma G}$ in ${\cal S}$ which, since the target is $0$-truncated, factors through $ \tau_{\leq 0}
((\Gamma X)^{h\Gamma G}) \to (\tau_{\leq 0} \Gamma X)^{h\Gamma G}$.
Combining this with the 0-truncation of Eq.~\ref{eq:CohomPartial} then
yields a sequence
\begin{equation}\label{eq:CohomSequence}
	\tau_{\leq 0}	\Gamma \left(X^{hG}\right) \to \tau_{\leq 0}
((\Gamma X)^{h\Gamma G}) \to (\tau_{\leq 0} \Gamma X)^{h\Gamma G} \mathrlap{,}
\end{equation}
of set maps.

Now specialize to the topos $\Sm$ of smooth spaces, setting $X =
 \mathbb{C}^\times\mathrm{Conn}(M)$, where $M$ corresponds to a smooth
manifold $\mathsf{M}$ and suppose that $G$ corresponds to a Lie group
$\mathsf{G}$ acting smoothly on $\mathsf{M}$. We claim that the
composition of the two maps in Eq.~\ref{eq:CohomSequence} gives the
desired map. Indeed, the left-hand term in
Eq.~\ref{eq:CohomSequence} is then the set of equivalence classes of
$\mathsf{G}$-symmetric invertible QFTs in $d=1$ equipped with an orientation and a smooth map to $\mathsf{M}$. In the right hand term, $\tau_{\leq 0} \Gamma
 \mathbb{C}^\times\mathrm{Conn}(M)$ is the differential cohomology in
degree two of $\mathsf{M}$. The
fact that $G$ acts on $\mathbb{C}^\times\mathrm{Conn}(M)$ via $M$ implies
that $\esh G$ also acts via $M$. Thus 
$\left( \tau_{\leq 0} \Gamma
\mathbb{C}^\times\mathrm{Conn}(M)\right)^{h \Gamma G}$ is the invariant
differential cohomology in degree two of the $\mathsf{G}$-manifold
$\mathsf{M}$, as defined in \cite{BB14}. Moreover, since there is a 1-1
correspondence between differential characters in degree two and
holonomy maps of bundles with connection \cite{BB14}, we see that the
composition of the two maps indeed corresponds to restricting the QFT
to closed spacetime manifolds and is the map we seek.

To see that the map is not surjective, in general, consider the case where the abelian Lie group $\mathbb{R}^2$  acts on itself by left (say)
translation. A calculation of the invariant differential cohomology in
degree two, as
in \cite{DGRW20}, shows that $(\tau_{\leq
0}\Gamma (
\mathbb{C}^\times\mathrm{Conn}(\mathbb{R}^2)))^{h\Gamma \mathbb{R}^2} \cong \mathbb{C}$. Physically this corresponds to the
classical mechanics of a particle moving in a plane in the presence of
a uniform magnetic field (which, since there is no requirement of
unitarity, is allowed to take complex values), whose physics action on
a spacetime $\mathsf{S}^1$ is
indeed invariant under translations of the plane. One of these fixed points corresponds to the bundle with connection $A = y dx$ (where $x$ and $y$ parameterise $\mathbb{R}^2$). However, the action of $\Gamma \R^2$ on the corresponding connected component  of $\Gamma  (
\mathbb{C}^\times\mathrm{Conn}(\mathbb{R}^2)))$ corresponds to the
non-trivial central extension of $\ssR^2$ by $\ssC^\times$ represented
by the cocycle $\epsilon((x_1,y_1), (x_2,y_2)) = e^{-i y_1 x_2}$. This
central extension does not split , meaning there is no fixed point in $ \tau_{\leq 0}
((\Gamma \mathbb{C}^\times\mathrm{Conn}(\mathbb{R}^2)))^{h\Gamma
  \mathbb{R}^2})$ corresponding to that in $(\tau_{\leq
0}\Gamma (
\mathbb{C}^\times\mathrm{Conn}(\mathbb{R}^2)))^{h\Gamma \mathbb{R}^2}
\cong \mathbb{C}$. In other words, the right-hand map
in (\ref{eq:CohomSequence}) is not surjective, so the composite of the
two maps in (\ref{eq:CohomSequence}) cannot be surjective either. This corresponds to the well-known fact that the
(exponentiated) physics action defined on an interval is not invariant under
translations, but rather shifts by a boundary term.

To see that the map is no injective, in general, it suffices to consider the action of a Lie group $\mathsf{G}$ on the
manifold consisting of a single point, for which the (invariant)
differential cohomology in degree two is a singleton set.
Since  $\mathbb{C}^\times\mathrm{Conn}(\ast)\simeq \mathbb{C}^\times$ the set $\tau_{\leq 0}	\Gamma \left(X^{hG}\right)$ is the set of equivalence classes of smooth splittings of the trivial extension of $\mathsf{G}$ by $\mathbb{C}^\times$. Thus choosing $\mathsf{G}$ to be $\ssZ/2$ already provides us with an example where the map fails to be injective.

The physical interpretation is as follows. The structure associated to
the homotopy fixed point is a 1-dimensional representation of $
\mathsf{G}$, which is carried by the state space attached to a
spacetime point. This structure is invisible when we evaluate the
theory only on closed spacetimes of dimension one.

\subsection{Non-Unitary Quantum Mechanics\label{sec:NonUnitQM}}
Now let us consider QFTs in $d=1$ on spacetimes equipped with both an orientation and a Riemannian metric. 
These theories are close to what a physicist would call quantum
mechanics, in that an oriented spacetime interval has associated to it
a 
length and an orientation allowing non-trivial
time evolution via a hamiltonian. They have, however, no notion yet of unitarity, so we call them `non-unitary quantum mechanical theories'.

According to \cite{GP21}, if we force the target category defining our theories to be the nerve of the $1$-category of finite-dimensional vector spaces and linear maps to avoid technical difficulties, the relevant smooth space is $X = \coprod_n BGL_n^{B\bbR}$. In what follows we shall consider just one component, so $X= BGL_n^{B\bbR}$.

So a point in $X$ is a morphism from $B\bbR$ to $BGL_n$. As we saw in
\S\ref{sec:smoothspaces}, every such morphism is homotopic to a
morphism corresponding to a smooth homomorphism from  
$\ssR$ to $\mathsf{GL}_n$, {\em i.e.}
to a smooth representation of $\ssR$ of dimension $n$. This corresponds physically to the time evolution.

In fact, the smooth space $BGL_n^{B\R}$ is equivalent to the smooth space $Mat_n/\!\!/GL_n$ corresponding to the action Lie groupoid
$\mathsf{Mat}_n/\!\!/\mathsf{GL}_n$, where $\mathsf{Mat}_n$ is the
smooth manifold of $n \times n$-matrices with values in $\ssC$ and the
action of the Lie group $\mathsf{GL}_n$ is by conjugation. To see
this, let us construct a morphism $Mat_n/\!\!/GL_n\to BGL_n^{B\R}$ as
follows: For every $p \in \mathbb{N}$, we must give a functor from
$(Mat_n/\!\!/GL_n)(\ssR^p)$  to  $BGL_n^{B\R}(\ssR^p)$. The source
groupoid $(Mat_n/\!\!/GL_n)(\ssR^p)$ has as objects smooth maps
$\mathsf{h}_1:\ssR^p\to \mathsf{Mat}_n$  and has as morphisms from
$\mathsf{h}_1$ to $\mathsf{h}_2$ smooth maps $\mathsf{m}:\ssR^p\to \mathsf{GL}_n$ such that
$\mathsf{h}_2(\mathsf{x})=\mathsf{m}(\mathsf{x})\mathsf{h}_1(\mathsf{x})\mathsf{m}(\mathsf{x})^{-1}$. The target groupoid $BGL_n^{B\R}(\ssR^p)$
has as objects smooth maps $\mathsf{r}:\ssR^p\times \ssR\to \mathsf{GL}_n$ which
are homomorphisms for all $\mathsf{x}\in \ssR^p$ and has as morphisms from $\mathsf{r}_1$
to $\mathsf{r}_2$ smooth maps $\mathsf{m}:\ssR^p\to \mathsf{GL}_n$ such that
$\mathsf{r}_2(\mathsf{x})=\mathsf{m}(\mathsf{x})\mathsf{r}_1(\mathsf{x})\mathsf{m}(\mathsf{x})^{-1}$. We choose the functor that takes
objects $\mathsf{h}$ in $(Mat_n/\!\!/GL_n)(\ssR^p)$ to $(\mathsf{x},\mathsf{t})\mapsto
e^{\mathrm{i} \mathsf{h}(\mathsf{x})\mathsf{t}}$ in $BGL_n^{B\R}(\ssR^p)$, and which takes
morphisms $\mathsf{m}$ in $(Mat_n/\!\!/GL_n)(\ssR^p)$ to the corresponding
morphism described by the same smooth map $\mathsf{m}$ in
$BGL_n^{B\R}(\ssR^p)$. The factor of $\mathrm{i}$ we introduced will be
convenient when we come to study unitary quantum mechanics. The
remaining unspecified data of  $Mat_n/\!\!/GL_n\to BGL_n^{B\R}$ are
taken such that the homotopies involved are in fact strict identities
between compositions of 1-functors. This morphism is an
  equivalence of smooth spaces whose inverse can be constructed by
  differentiating the smooth maps $\ssR^p\times \ssR\to \mathsf{GL}_n:(\mathsf{x},\mathsf{t})\mapsto e^{\mathrm{i}\mathsf{t}\mathsf{h}(\mathsf{x})}$
  with respect to $\mathsf{t}$. Physically, we can think of $\mathsf{Mat}_n$ as the smooth manifold of
hamiltonians of quantum-mechanical theories with a state space of
dimension $n$. For now, since we have no notion of unitarity yet, the
state space is a mere vector space, and there is no requirement for a hamiltonian to be hermitian. The action encodes the fact that conjugation by $\mathsf{GL}_n$ sends a hamiltonian to one that is physically equivalent.

We remark that it is crucial that we retain the information about the
physical equivalences between hamiltonians. If we simply na\"{\i}vely
pass to the set $\mathsf{Mat}_n/\mathsf{GL}_n$ of equivalence classes
of hamiltonians (which Simplicio does when they  `diagonalize the
hamiltonian'\footnote{We caution that in the situation without
  unitarity, one cannot in fact necessarily diagonalize the
  hamiltonian, {\em cf. e.g.} $\left(\begin{smallmatrix} 1 & 1 \\ 0 &
    0 \end{smallmatrix}\right) \in \mathsf{Mat}_2$; but one can pass to some
  canonical normal form, such as the Jordan normal form.}), we lose the
smooth structure. Even if we consider $\tau_{\leq 0} (Mat_n/\!\!/{GL}_n)$,
which is the smooth space corresponding to
$\mathsf{Mat}_n/\mathsf{GL}_n$, we will not be able to give a complete
description of the smooth symmetries.

Let us now make some relevant remarks about the smooth space
$Mat_n/\!\!/GL_n$. Firstly, it is not connected and has many
inequivalent points, for its
0-truncation is, as a space, that corresponding to the set of
equivalence classes of hamiltonians, of which there are many. (For
$n=1$, for example, we get the set $\bbR$; more generally we get the
set of possible
elementary divisors, {\em i.e.} the characteristic polynomials of the
Jordan blocks.)

Secondly, given a point $h$ in $Mat_n/\!\!/GL_n$ corresponding to a hamiltonian $\mathsf{h} \in
\mathsf{Mat}_n$, the connected component at $h$ is given by the
delooping of the group object $S_h$ corresponding to the Lie group
$\mathsf{S}_h$ given by the stabilizer of $h$ under the action of $\mathsf{GL}_n$.

From here, it is easy to satisfy Simplicio. For $G$-actions with
homotopy fixed point through a point corresponding to a hamiltonian
$\mathsf{h}$,
we may replace $G$ by $\tau_{\leq 0} G$, following the discussion in \S\ref{sec:findact}. Taking $G_0 := \tau_{\leq 0} G$ to be Lie for
simplicity, we find that the possible $G$-actions with homotopy fixed
point through $h$ are described by split short exact sequences of Lie
groups extending $\mathsf{G}_0$ by $\mathsf{S}_h$.
Suppose the automorphism $\mathsf{e}:\mathsf{G}_0\times
\mathsf{S}_h\to \mathsf{S}_h$ describes the splitting
(see~\S\ref{sec:grpSSl}). Then the homotopy fixed points are given by
smooth maps $\mathsf{r}:\mathsf{G}_0\to \mathsf{S}_h$ satisfying
Eq.~\ref{eq:TwistedReps}. 
 So, for example, if we take the trivial extension
$\mathsf{S}_h \times \mathsf{G}_0$, we get a {\em bona fide}
representation of $\mathsf{G}_0$ that commutes with $\mathsf{h}$, but if $G_0$ acts on $S_h$
by inner automorphisms, we get a projective representation of $\mathsf{G}_0$
that commutes with $\mathsf{h}$.
So physically, we get back the usual
expectation that symmetries correspond to linear (not unitary here)
operators that commute with the hamiltonian,
in the sense that 
$\mathsf{r}(\mathsf{g}) \mathsf{h}=\mathsf{h}\mathsf{r}(\mathsf{g}),$ 
but they need only form a representation of a group up to twisting
via $\mathsf{e}$. 

Here an explicit description of the automorphisms of $\mathsf{S}_h$,
is far from straightforward, not least because $\mathsf{S}_h$, though
connected, is not compact.\footnote{To see that $\mathsf{S}_h$ is
  connected, let $\mathsf{s},\mathsf{s}^\prime \in \mathsf{S}_h$ and consider the
  function $\mathsf{f}:\ssC\to \mathsf{Mat}_n: \mathsf{t}\mapsto \mathsf{s}\mathsf{t}+\mathsf{s}^\prime (1-t)$. By
  the fundamental theorem of algebra, the zero set $Z$ of
  $\mathsf{det} \mathsf{f}: \ssC \to \ssC$ consists of at most $n$ isolated points so $\ssC-Z$ is
  path connected.  Moreover, $0,1 \in \ssC-Z$. Since $\mathsf{f}\mathsf{h} = \mathsf{h}\mathsf{f}$,
  evaluating $\mathsf{f}$ along any path from $ \mathsf{t}=0$ to $ \mathsf{t}=1$ in $\ssC-Z$ exhibits
  a path in $\mathsf{S}_h$ connecting $s$ to $s^\prime$. To see that
 $\mathsf{S}_h$ is not compact, observe that it has a Lie subgroup isomorphic
 to $\ssC^\times$ given by the non-zero diagonal matrices.}

To give just the simplest example, if
  $\mathsf{S}_h=\mathsf{GL}_1 \simeq \ssC^\times$, all automorphisms
  are necessarily outer and the Lie group $\mathsf{Aut}\,  \mathsf{GL}_1$ is isomorphic to the subgroup of $\mathsf{GL}_2$
  consisting of $2\times 2$ matrices of the form $\left(\begin{smallmatrix} a &
    b \\ 0 & \pm 1 \end{smallmatrix}\right)$, which, in terms of the connected
  component at the identity (isomorphic to the
  affine group of the real line) and the group of components
  (isomorphic to $\ssZ/2$), can be
  expressed as a short exact sequence
  of Lie groups that is right- (but not left-) split.

Since the smooth space $X$ of non-unitary quantum-mechanical theories
is not connected, we cannot expect to be able to classify all group
actions on the full smooth space easily. We shall content ourselves
with describing some specific examples. Moreover, there can be smoothness anomalies, in that a group action defined on the
connected component at $h$ may not extend to all of $X$ and indeed the
example we give for unitary quantum mechanics in \S\ref{sec:UnitQM}
provides an
example for non-unitary quantum mechanics as well. 
\subsubsection{Smooth TQFTs in $d=1$ - Take two}\label{sec:1dTFTTake2}

If we focus on the special case of theories corresponding to the
point $h=0$ (of which up to equivalence there is one for each $n$), we
get another smooth incarnation of the oriented TQFTs in
$d=1$ whose symmetries we discussed in \cite{GRWTS23}.

To be precise, we should like to consider the homotopy fixed points
through $h=0$ of
group actions on
$Mat_n/\!\!/GL_n$. According to the general arguments in \S \ref{sec:orbstab}, we can
do this by first studying the group actions with homotopy fixed points
on $\tge{1}{x}{X}$, which here is equivalent to $BGL_n$ (since every element
of $\mathsf{GL}_n$ stabilizes $h=0$), and then worrying about which
actions on $\tge{1}{x}{X}$ extend to actions on $X$. As we show in the \pagetarget{page:ExtendToQM}{Appendix}~\ref{proof:ExtendToQM}, there is in
fact nothing to worry about here, in that all actions extend. 

The problem, and  the physical interpretation is now identical, as one would hope, to what we found in~\S\ref{sec:TFTs1dTake1}.

\subsubsection{Invertible non-unitary quantum mechanics}
Next consider the special case where $n=1$, for which  ${Mat}_1/\!\!/GL_1=\bbC\times
B\bbC^\times$.  While this is still not connected, its $0$-truncation
is simply the smooth space corresponding to the manifold $\ssC$, which
will allow us to make some progress.

For a generic group $G$ acting on $\C \times B\C^\times$, we have a diagram
\pastingVE{\bbC\times B\bbC^\times}{(\bbC\times
  B\bbC^\times)/\!\!/G}{\bbC}{\bbC/\!\!/G}{\terminal}{BG\mathrlap{,}}
in $\Sm$ in which all squares are cartesian. The bottom square defines an
action of $G$ on the smooth space corresponding to
the manifold $\ssC$, whose homotopy quotient we denote $\bbC /\!\!/ G$. The top square shows that the morphism $(\bbC\times
  B\bbC^\times)/\!\!/G \to \bbC /\!\!/ G$ has the structure of a
  $B\bbC^\times$-fibre bundle over $\bbC/\!\!/G$, as defined in \cite{NSS14}, in which the effective
  epimorphism $\bbC \to \bbC/\!\!/G$ provides a cover that trivializes the
  bundle. As shown there,  $B\bbC^\times$-fibre bundles over $\bbC/\!\!/G$ are classified by morphisms
  from $\bbC/\!\!/G\to B\Aut \, (B\bbC^\times)$, though we must take care that not all such bundles will be trivialised by $\C\to \bbC/\!\!/G$. Going in the other direction, we can
  classify all $G$-actions on $\C\times B\bbC$ by first finding all $G$-actions on
  $\bbC$. For each of these, we must then find the possible
  $B\bbC^\times$-bundles over the resulting homotopy quotient
  $\bbC/\!\!/G$, subject to the condition that they trivialize under $\C\to
  B\bbC/\!\!/G$. Such a bundle always exists, since we
  have the trivial bundle $B\bbC^\times \times \bbC/\!\!/G$. 
  
The richness of possible $G$-actions on smooth spaces of QFTs is now
clear, since even when $G$ corresponds to
a Lie group, in general there will be many ways that it can act smoothly
on $\bbC$, and each of these will induce at least one action on
$X$. So it is hard to say anything else in full generality. Thus, let us turn to constructing explicit examples of group actions. 

Suppose, as a first example, that $G$ acts trivially on $\bbC$, but that
we have a non-trivial action of $G$ on $B\bbC^\times$, denoted
$B\bbC^\times/\!\!/G$. Then we have a $G$-action on $\bbC\times B\bbC^\times$ given by 
\pastingVE{\bbC\times B\bbC^\times}{(\bbC\times
  B\bbC^\times)/\!\!/G}{\bbC}{\bbC\times BG}{\terminal}{BG\mathrlap{,}}
 This example is of interest because it shows that we have no
 smoothness anomalies. Indeed, the connected component of any point
 $\terminal\to \C\times B\C^\times$ is simply $B\C^\times$ and the above
 action leads to the action corresponding to $B\bbC^\times /\!\!/G$ on each connected component.
  
For our next example, we consider actions on $\C \times B\C^\times$ of
a group object that is itself a delooping, such that we may write it as
$G = B\hat G$. These are of interest, since they justify our earlier claim that
QFTs in spacetime dimension $d$ can have non-trivial $d$-form
symmetries. Moreover, they illustrate another surprising feature of
actions of such groups: their possible homotopy fixed points can pick
out a `proper subspace' of $\tau_{\leq 0} X$, even though they necessarily act
trivially on $\tau_{\leq 0} X$.

We consider a specific class of actions of such groups that are
obtained by pullback of a canonical example, in which $\hat G$ is the
abelian group object $(\C^\times)^\C$
corresponding to the abelian group
whose elements are
smooth maps from $\ssC$ to $\ssC^\times$, with multiplication given
by pointwise multiplication in the target, made into a smooth space in
the obvious way.\footnote{The reader will observe
  that the construction works for any $X$ of the form $Y \times BA$,
  provided $BA$ is deloopable.}

Evaluation on the target defines a morphism $\C \times (\C^\times)^\C
\to \C^\times$ and we can use this to build a morphism $\C \times B^2(\C^\times)^\C
\to B^2\C^\times$ as, for each $\ssR^p$, the functor
$(\C\times B^2(\C^\times)^\C)(\ssR^p)
\to B^2\C^\times(\ssR^p)$ that acts trivially on 0- and 1-simplices and by evaluation on 2-simplices. 
Using this morphism, we can form the diagram
\[\begin{tikzcd}
	\C \times B\C^\times & (\C \times B\C^\times) /\!\!/
        B(\C^\times)^\C & \terminal \\
	\C & {\bbC\times B^2(\C^\times)^\C} & {B^2\C^\times} \\
	\terminal & {B^2(\C^\times)^\C\mathrlap{,}}
	\arrow[from=1-3, to=2-3]
	\arrow[from=2-1, to=2-2]
	\arrow[from=2-2, to=2-3]
	\arrow[from=1-2, to=1-3]
	\arrow[from=1-2, to=2-2]
	\arrow[from=1-1, to=1-2]
	\arrow[from=1-1, to=2-1]
	\arrow[from=2-1, to=3-1]
	\arrow[from=3-1, to=3-2]
	\arrow[from=2-2, to=3-2]
\end{tikzcd}\]
in which all squares are cartesian and the top rectangle is the
trivial principal $B\C^\times$-bundle over $\C$. 
The left-hand rectangle then
exhibits an action of $B(\C^\times)^\C$ on $\C \times
B\C^\times$.\footnote{More explicitly, the object $(\C \times B\C^\times) /\!\!/
        B(\C^\times)^\C$ can be described as follows.  To $\ssR^p$, we
        assign the $2$-groupoid whose objects are smooth maps $\ssR^p\to
        \ssC$, while a 1-simplex $f_0\to f_1$ exists if and only if
        $f_0=f_1$ and is given by a smooth map $\mu:\ssR^p\to \ssC^\times$, and a
        2-simplex with all $0$-vertices $f$ and $1$-vertices
        $\mu_{01}$, $\mu_{12}$ and $\mu_{02}$ is a smooth map $\eta:\ssC\times \ssR^p\to \ssC^\times$ such that $\mu_{01}(x)\mu_{12}(x)=\mu_{02}(x)\eta(f(x),x)$.}

Given any abelian group object $B\hat G$ and a morphism
$B^2\hat G \to B^2(\C^\times)^\C$, we can construct a group action of
$B\hat G$ on $B(\C^\times)^\C$  by pulling back. We note that we get a
similar diagram for $B\hat G$, namely 
\[\begin{tikzcd}
	\C \times B\C^\times &(\C \times B\C^\times)/\!\!/B\hat G & \terminal \\
	\C & {\bbC\times B^2\hat G} & {B^2\C^\times} \\
	\terminal & {B^2\hat G\mathrlap{.}}
	\arrow[from=1-3, to=2-3]
	\arrow[from=2-1, to=2-2]
	\arrow[from=2-2, to=2-3]
	\arrow[from=1-2, to=1-3]
	\arrow[from=1-2, to=2-2]
	\arrow[from=1-1, to=1-2]
	\arrow[from=1-1, to=2-1]
	\arrow[from=2-1, to=3-1]
	\arrow[from=3-1, to=3-2]
	\arrow[from=2-2, to=3-2]
\end{tikzcd}\]
While the bottom square imposes no constraint on the action (since
$B\hat G$ always acts trivially on the $0$-truncation $\C$ of $\C
\times B\C^\times$), the action is strongly constrained by the top
squares. Namely, $(\C \times B\C^\times)/\!\!/B\hat G \to  {\bbC\times B^2\hat G}$
must not only have the structure of a principal ${B\C^\times}$-bundle, but also must trivialize under pullback along $\C \to \C \times
B\C^\times$.

Moreover, our original action of $B(\C^\times)^\C$ is apparently anyway
uninteresting for physics, since it has no homotopy fixed
points. Indeed,
  $\tge{2}{[z]}{(\C \times B\C^\times) /\!\!/
        B(\C^\times)^\C}=:B^2S_z$ at a point $z$ represented
  by $\mathsf{z} \in\ssC$ assigns to $\ssR^0$ the subgroup of smooth maps $\mathsf{f}:\ssC\to
  \ssC^\times$ for which $\mathsf{f}(\mathsf{z})$ is the
    identity in $\ssC^\times$. In spaces
  $B^2S_z(\ssR^0)\to B^2(\bbC^\times)^\bbC (\ssR^0)$ does not have a
  section for any $z$, thus neither does the morphism $B^2S_z \to
  B^2(\bbC^\times)^\bbC$ in smooth spaces. Hence, the action of
  $B(\C^\times)^\C$ on $\C\times B\C^\times$ has no sections
 {\em i.e.} no homotopy fixed points.

Nevertheless, after pulling back we can obtain a plentiful supply of interesting actions
that do have homotopy fixed points
using this construction, as the
following examples show. 

Consider a $B\hat G$-action of the above type. Then,  $\tge{2}{[z]}{(\C \times B\C^\times) /\!\!/
        B\hat G}=: B^2\hat S_z$ for a point  $z:\ast \to\C$
 is related to  $B^2S_z$ by the cartesian square
\square{B^2\hat S_z}{B^2 S_z}{B^2\hat G}{B^2(\C^\times)^\C\mathrlap{.}}{}{}{}{}
For this $B\hat G$-action
to have a homotopy fixed point, the left vertical morphism must have a section. This is
equivalent to saying that $B^2\hat G\to B^2(\C^\times)^\C$ must
factor, up to homotopy, through $B^2S_z \to
B^2(\C^\times)^\C$. Looking at this result in spaces, {\em i.e.} after
applying $\Gamma$, tells us that we only have homotopy fixed points
through $z$ if, under the double looping of $B^2\Gamma \hat G\to B^2\Gamma (\C^\times)^\C$, every object in $\Gamma \hat G$ is taken to a  smooth
map  $\ssC\to \ssC^\times$ in $\Gamma (\C^\times)^\C$ that evaluates to $1\in \ssC^\times$ at
$\mathsf{z}$. (If $\hat G$ happens to be discrete, this condition is
both necessary and sufficient.)

For a first example, let $B\hat G=B\R$, corresponding to the abelian
group of the real numbers under addition, considered as a discrete
group. Then let $B^2\R\to
B^2(\C^\times)^\C$ be the morphism corresponding to the homomorphism
sending $\mathsf{x}\in \ssR$ to the map $\ssC\to \ssC^\times: \mathsf{z} \mapsto
e^{\mathsf{z} \mathsf{x}}$. The induced action of $B\R$ on $\C\times B\C^\times$
has a single homotopy fixed point through $\mathsf{z}=0$.
More generally, one could take $\mathsf{x}$ to the map $\ssC\to \ssC^\times: \mathsf{z}
\mapsto e^{\mathsf{f}(\mathsf{z})\mathsf{x}}$ for some smooth map $\mathsf{f}:\ssC\to \ssC$; the homotopy
fixed points are then through the zeroes of $\mathsf{f}$.

As a second example, take $B\hat G=B\Z$, the abelian group of integers
under addition. Then let $B^2\Z\to B^2(\C^\times)^\C$ be the morphism corresponding to the map
taking $\mathsf{n}\in \ssZ$ to $\ssC\to \ssC^\times: \mathsf{z} \mapsto e^{\mathsf{z}\, \mathsf{n}}$. The points $z\in \C$ which are homotopy fixed points of this action are those corresponding to $\mathsf{z}\in2\pi \mathrm{i} \ssZ\subset \ssC$.

These group actions are analogous, in terms of the homotopy fixed
points of the induced action on $\tau_{\leq 0} X$, to the standard notion of group actions on sets. But in terms
of the actions themselves, they are very much unlike standard group
actions on sets, because they produce non-trivial homotopy fixed
points of $\tau_{\leq 0} X$ even though the action on $\tau_{\leq 0} X$ is
trival. The `explanation' for this is that although every group action
on $X$ with a homotopy fixed point induces a group action on $\tau_{\leq 0}
X$ with homotopy fixed point, there is no guarantee that all homotopy
fixed points of the latter arise in this way.

We stress that this phenomenon arises only because $X$ is not
connected; poor Simplicio will never observe such things.
\subsection{Unitary Quantum Mechanics\label{sec:UnitQM}}
We now discuss how to add a unitary structure to the theories
described, to get something which could genuinely be described as quantum
mechanics. In fact this can be done using symmetry
considerations,\footnote{For an alternative implementation, which
  generalizes to QFTs in higher $d$, see \cite{KS21}.}
making it all the more suitable for discussion here.

A key point is that we in fact need both a structure and a property
(corresponding respectively to reflection and postivity, respectively,
in the euclidean context \cite{FH21,MS23}), which
we can implement here by means of 0- and 1-form symmetries,
respectively.

We begin by observing that there is an action of the group object
corresponding to the Lie group $\ssZ/2$ on $\coprod_n BGL_n$ that sends the
non-trivial element of  $\ssZ/2$ to the automorphism $M\mapsto
(M^{\dagger})^{-1}$ of $\mathsf{GL}_n$. According to the discussion in
\S\ref{sec:grpSSl}, a homotopy fixed point of this action is given by an element
$\mathsf{A} \in \mathsf{GL}_n$ such that $\mathsf{A} = \mathsf{A}^\dagger$, {\em i.e.}  a
non-degenerate hermitian form. This $\mathsf{A}$ specifies where
the non-trivial element of $\ssZ/2$ is sent under the twisted representation as defined in Eq.~\ref{eq:TwistedReps}. The morphisms between homotopy fixed
points correspond to conjugation by $\mathsf{GL}_n$, so every homotopy fixed
point is equivalent to a diagonal matrix whose entries are $\pm 1$. 
More generally, the smooth space of homotopy fixed points assigns to
$p$ the groupoid with objects smooth maps $\mathsf{A}:\ssR^p\to
\mathsf{GL}_n$ taking values in non-degenerate hermitian forms and
morphisms from $\mathsf{A}_1$ to $\mathsf{A}_2$ given by smooth maps
$\mathsf{m}:\ssR^p\to \mathsf{GL}_n$ such that
$\mathsf{A}_2(\mathsf{x})=\mathsf{m}(\mathsf{x})
\mathsf{A}_1(\mathsf{x})\mathsf{m}(\mathsf{x})^{\dagger}$. This smooth
space is equivalent to $\coprod_n \coprod_{p+q=n} BU_{p,q}$, since the
Gram-Schmidt procedure can be carried out for smooth families.

This $\Z/2$-action on $BGL_n$ extends, of course, to an action on the
smooth space $\coprod_n (BGL_n)^{B\bbR}$ of non-unitary quantum
mechanics (where it corresponds to sending objects in the target
category of the QFT to their duals).\footnote{To get a reflection
  structure, as is required for reflection-positivity in euclidean
  QFTs, one should instead take the homotopy fixed points of the
  action which simultaneously sends $t \in \R$ to $-t$, {\em cf.}
  \cite{FH21}.} The corresponding homotopy quotient is given (for
  each $n$) by the mapping object
$(BGL_n/\!\!/\Z/2)^{B\R\times B\Z/2}$ in the slice topos $\Sm_{/B\Z/2}$ and the space of homotopy fixed points of this action
can be written in terms of the semidirect product associated to the
$\Z/2$ action on $BGL_n$ as ${\Sm}_{/B\Z/2}(B(\R\times \Z/2),
B(GL_n\rtimes \Z/2))$. A vertex in this space corresponds, via the arguments in \S\ref{sec:grpSSl}, to a smooth homomorphism $\ssR\times \ssZ/2 \to \mathsf{GL}_n\rtimes \ssZ/2$, such that the diagram
\twosimplex{\ssR\times  \ssZ/2}{\mathsf{GL}_n\rtimes \ssZ/2}{\ssZ/2\mathrlap{,}}{}{}{}{}
commutes (on the nose, since $\ssZ/2$ is abelian).

More prosaically, a homotopy fixed point of this action is, for some $n$,
an $n$-dimensional representation $e^{\mathrm{i}th}$ of $\ssR$ together with a
non-degenerate hermitian form $A$ of dimension $n$ such that $Ae^{\mathrm{i}th}
= e^{\mathrm{i}th^\dagger}A$.

Using the mapping space and $(\terminal_{B\Z/2})^\ast \dashv (\terminal_{B\Z/2})_\ast$
adjunctions, we find that the smooth space of homotopy fixed points is
equivalent to $\coprod_{p+q=n} (BU_{p,q})^{B\R}$. This smooth space of
QFTs, which we might call hermitian quantum mechanics, suffers from
the fact that probabilities computed using Born's rule can exceed 1.  
To remedy that, we need an extra property that singles out the
$p=n$, $q=0$ component of   $\coprod_{p+q=n} (BU_{p,q})^{B\R}$. One
can do this by hand or via a the action of a connected group object in smooth spaces (`a 1-form symmetry').\footnote{For example, the $B\Z/2$ action that is induced by the trivial action on
    $BU_{n,0}$ and the $B\Z/2$-action on $BU_{p,q}$, for which
    $BU_{p,q}/\!\!/B\Z/2$ assigns to, {\em e.g.}, $\ssR^0$ the 2-groupoid with a single 0-simplex, 1-simplices elements of $\mathsf{U}_{p,q}$, and 
   2-simplices with edges $U_{01},U_{12},U_{02}\in U_{p,q}$ that are labeled by $\pm
   1$ such that  $U_{01}U_{12} =\pm U_{02}$.} 

The smooth space corresponding to unitary QM is then $X=BU_n^{B\R}$. Similar to the non-unitary case, this is equivalent to the smooth space $Her_n/\!\!/U_n$ corresponding to the action Lie groupoid $\mathsf{Her}_n/\!\!/\mathsf{U}_n$, where $\mathsf{Her}_n$ denotes the smooth manifold of $n\times n$-hermitian matrices with values in $\ssC$ and the action of the Lie group $\mathsf{U}_n$ is by conjugation.
The smooth space $Her_n/\!\!/U_n$ is not connected; by
diagonalizing, we see that its
$0$-truncation, as a space, corresponds to the set in which an
element is a choice of $n$ unordered
real numbers. 
Labelling the degeneracies of these by $q_i \in \{0,1,2,\dots \}$
(so that $\sum_{i \in I} q_i = n$), we find that the connected component is given by
the delooping of the stabilizer under the $\mathsf{U}_n$ action, {\em i.e.} by
the Lie group $\prod_{i \in I} \mathsf{U}_{q_i}$. 

Ideally, Simplicio would now like to know the possible smooth automorphisms of
$\mathsf{S}_h = \prod_{i \in I} \mathsf{U}_{q_i}$ corresponding to the $h$ of
interest. We note that $\mathsf{S}_h$ is both connected and compact,
(though neither simply-connected  nor
semisimple unless $n=0$). It should therefore be possible to give an
explicit general description, but we refrain from doing so
here. To give just the simplest case, when $\mathsf{S}_h =
  \mathsf{U}_1$, we get that $\mathsf{Aut} \,  \mathsf{U}_1 \simeq
  \ssZ/2$, corresponding to complex conjugation.

Let us instead, finally, give an explicit example of a
smoothness anomaly. 
Recall that in this case, a smoothness anomaly is a $G$-action on
$BS_h$ which does not extend to a $G$-action on $X$. We have already
seen that there are no smoothness anomalies in non-unitary quantum
mechanics when $n=1$ and a similar argument applies here. So to make
things as simple as possible, we will
consider the case $n=2$.  For the same reason, we take $G=\Z/2$ and
$\mathsf{h}=\mathrm{diag}(1,-1)$, so that $\mathsf{S}_h\simeq \mathsf
{U}_1^{ 2}:=\mathsf{U}_1\times \mathsf{U}_1$ and $\mathsf{Aut} \, \mathsf{S}_h \simeq \mathsf{GL}_2(\ssZ)$. We take
the action in which the non-trivial element of $\ssZ/2$ is sent to the
smooth automorphism 
$\mathsf{c}: (e^{\mathrm{i}\uptheta},e^{\mathrm{i}\upphi}) \mapsto
(e^{\mathrm{i}\uptheta},e^{-\mathrm{i}\upphi})$ of $\mathsf{S}_h$. A homotopy fixed
  point of this $\Z/2$-action on $B {U}_1^{ 2}$ corresponds to a twisted representation, as per Eq.\ref{eq:TwistedReps}, which assigns an element of the form $(\pm 1,e^{\mathrm{i}\uptheta})\in \mathsf {U}_1^{ 2}$ to the non-trivial element of $\ssZ/2$. We denote by $j:B {U}_1^{ 2}\to Her_2/\!\!/U_2$ the
inclusion of $B {U}_1^{ 2}$ into $Her_2/\!\!/U_2$ which (following our discussion of Lie
groupoids in \S \ref{sec:smoothspaces}) is induced by the smooth map
$\mathsf{j}: \mathsf {U}_1^{ 2}\to\mathsf{Her}_2\times  \mathsf{U}_2:
(e^{\mathrm{i}\uptheta},e^{\mathrm{i}\upphi}) \mapsto \left(\mathsf{h}, (e^{\mathrm{i}\uptheta},e^{\mathrm{i}\upphi})\right)$.

Now let us suppose that there is an action on $Her_2/\!\!/U_2$ that leads to
  this action on $BS_h$ and show that it implies a
  contradiction. Taking the putative action, we have a diagram
  (analogous to the diagram \ref{eq:AutoDiagram})
  \[\begin{tikzcd}[row sep={45,between origins}, column sep={45,between origins}]
B {U}_1^{ 2}\ar{dd} \ar[rr,"j",pos=0.5]\and    \and  Her_2/\!\!/U_2 \ar[crossing over]{rr}  \and \and \terminal\\[\perspective]
 \and B {U}_1^{ 2} \ar[rr,"j",pos=0.2] \ar{dl} \ar[ul,"Ba",swap] \and     \and Her_2/\!\!/U_2\ar[ul,"f",swap] \ar{rr} \ar{dl} \and\and \terminal\ar[""{name=0, anchor=center, inner sep=0},dl]\ar{ul} \\[-\perspective]
B {U}_1^{ 2}/\!\!/\Z/2\ar{rr} \and  \and   (Her_2/\!\!/U_2)/\!\!/\Z/2 \ar{rr} \ar[from=uu,crossing over]\and \and B\Z/2\mathrlap{,}\ar[from=uu,crossing over]\and
\arrow[{description}, shorten >=3pt, Rightarrow, from=1-5, to=0]
\end{tikzcd}\]
in which all squares are cartesian and the morphisms denoted $Ba$ and $f$ must be equivalences, since they are
obtained by pulling back the equivalence $\terminal \to \terminal$. There are two possible choices for the filler of the horn $\terminal\to B\Z/2
\leftarrow \terminal$ denoted by the double arrow, corresponding to the two
elements of $\ssZ/2$; we shall show that for the non-trivial
element, no suitable $f$ exists. 

To do so,  we consider first the constraints on $a$ imposed by the
outer prism in the diagram. Since this prism only
involves objects that are deloopings of Lie groups, we may replace it
by the equivalent diagram
\[\begin{tikzcd}[row sep={45,between origins}, column sep={45,between origins}]
    \mathsf {U}_1^{ 2} \ar[crossing over]{rr}  \and \and \terminal\\[\perspective]
    \and \mathsf {U}_1^{ 2} \ar[ul,"\mathsf{a}",swap] \ar{rr} \ar{dl} \and\and \terminal\ar[""{name=0, anchor=center, inner sep=0},dl]\ar{ul} \\[-\perspective]
        \mathsf {U}_1^{ 2}\rtimes \ssZ/2 \ar{rr} \ar[from=uu,crossing over]\and \and \ssZ/2 \mathrlap{,}\ar[from=uu,crossing over]\and
\arrow[{description}, shorten >=3pt, Rightarrow, from=1-3, to=0]
\end{tikzcd}\] 
in the category ${\cal L}$ defined in \S\ref{sec:smoothspaces}, where the semidirect
product $\mathsf{U}_1^{ 2}\rtimes \ssZ/2$ is defined by the given action of
$\ssZ/2$ on $\mathsf{U}_1^{ 2}$ in the usual way. In order for this to
commute, we must have that $\mathsf{a} = \mathsf{b} \circ \mathsf{c} \circ
\mathsf{b}^{-1}$ for some automorphism $ \mathsf{b}$ of
$\mathsf{U}_1^{ 2}$.

Having constrained the allowed forms $\mathsf{a}$, our job will be done if we can show
that there is no valid $\mathsf{a}$ and no equivalence $f$ making a commutative square
\[\begin{tikzcd}
	{B {U}_1^{ 2}} & {B {U}_1^{ 2}} \\
	{Her_2/\!\!/U_2} & {Her_2/\!\!/U_2\mathrlap{.}}
	\arrow["Ba", from=1-1, to=1-2,dashed]
	\arrow["f"', from=2-1, to=2-2,dashed]
	\arrow["j", from=1-2, to=2-2]
	\arrow["j"', from=1-1, to=2-1]
\end{tikzcd}\]
To show it, we turn this putative commutative diagram into a
putative commutative triangle involving Lie groups,
for which we can use the explicit description of the full sub-category
$\mathcal{L}$ of $\Sm$ in \S\ref{sec:smoothspaces} to show that such a triangle cannot exist.

As an intermediate step, we form the commutative diagram
\[\begin{tikzcd}
	& {B {U}_1^{ 2}} & {B {U}_1^{ 2}} \\
	{\R\times B {U}_1^{ 2}} & {Her_2/\!\!/U_2} & {Her_2/\!\!/U_2} & {BU_1\mathrlap{.}} \\
	{B {U}_1^{ 2}} & {BU_2}
	\arrow["Ba", from=1-2, to=1-3,dashed]
	\arrow["f"', from=2-2, to=2-3,dashed]
	\arrow["j", from=1-3, to=2-3]
	\arrow["j"', from=1-2, to=2-2]
	\arrow[from=2-1, to=2-2]
	\arrow[from=3-2, to=2-2]
	\arrow[from=2-3, to=2-4]
	\arrow["Bd"', from=3-1, to=3-2]
	\arrow["{\{0\}\times \mathrm{id}}", from=3-1, to=2-1]
	\arrow["{\{1\}\times \mathrm{id}}"', from=1-2, to=2-1]
	\arrow["BT"', from=3-2, to=2-4,dashed]
	\arrow["BR", bend left=50, from=1-2, to=2-4,dashed]
\end{tikzcd}\]
Here, the morphism $Bd$ corresponds to the homomorphism $\mathsf{d}:
\mathsf{U}_1^{ 2}\to
\mathsf{U}_2:(e^{\mathrm{i}\uptheta},e^{\mathrm{i}\upphi})\mapsto
\mathrm{diag}(e^{\mathrm{i}\uptheta},e^{\mathrm{i}\upphi})$. The
morphism $\R \times {B} {U}_1^{ 2}\to  {Her}_2/\!\!/{U}_2$ is that
constructed, in the manner of \S\ref{sec:smoothspaces}, from the
smooth functor of Lie groupoids $\ssR \times \mathsf{B} \mathsf{U}_1^{
  2}\to  \mathsf{Her}_2/\!\!/\mathsf{U}_2$ that maps morphisms
$(\mathsf{x},\mathsf{g})\in \ssR\times \mathsf{U}_1^{ 2}$ to
$(\mathsf{x}\,\mathsf{h} ,\mathsf{d}(\mathsf{g}))\in
\mathsf{Her}_2\times \mathsf{U}_2$. The morphism $BU_2\to
Her_2/\!\!/U_2$ is that constructed from the smooth functor
$\mathsf{B}\mathsf{U}_2\to \mathsf{Her}_2/\!\!/\mathsf{U}_2$ that maps morphisms  $\mathsf{g}\in \mathsf{U}_2$ to $(0,\mathsf{g})\in \mathsf{Her}_2\times \mathsf{U}_2$. The morphism $ Her_2/\!\!/U_2\to BU_1$ is the composition of the morphism $ Her_2/\!\!/U_2\to BU_2$ defining the conjugation action of $U_2$ on $Her_2$ and the morphism $Bdet:BU_2\to BU_1$ corresponding to the determinant homomorphism $\mathsf{det}:\mathsf{U}_2\to \mathsf{U}_1$. The morphism $BR:B{U}_1^{ 2}\to BU_1$ corresponds to the homomorphism 
 \begin{equation} \label{eq:f1k}
\mathsf{R}:	(e^{\mathrm{i}\uptheta}, e^{\mathrm{i}\upphi})\mapsto \exp\left\{ \mathrm{i}\begin{pmatrix}
            1& 1\end{pmatrix} b \begin{pmatrix} 1 & 0 \\ 0 & -1\end{pmatrix} b^{-1}\begin{pmatrix} \uptheta  \\ \upphi \end{pmatrix}\right\},
\end{equation}
where $b \in \mathsf{GL}_2(\ssZ)$ is the element in $\mathsf{Aut}
\, \mathsf{U}_1^{ 2}$ corresponding to the automorphism $\mathsf{b}$.
 Lastly, taking a composite of the morphism $BU_2\to Her_2/\!\!/U_2$, the morphism $f$, and the morphism $Her_2/\!\!/U_2\to BU_1$ we can construct a $BT:BU_2\to BU_1$ making the
 diagram commute.
 
Any morphism $\R\to (BU_1)^{B{U}_1^{ 2}}\simeq BU_1\times \Z \times \Z$ factors through $\ast$; consequently the composite horizontal morphism in the above diagram $\R\times B{U}_1^{ 2}\to BU_1$ must factor through the projection $p_2:\R\times B{U}_1^{ 2}\to B{U}_1^{ 2}$. The result is the commutative diagram
\[\begin{tikzcd}
	& {B{U}_1^{ 2}} \\
	{\R\times B{U}_1^{ 2}} & {B{U}_1^{ 2}} && {BU_1\mathrlap{,}} \\
	{B{U}_1^{ 2}} & {BU_2}
	\arrow["Bd"', from=3-1, to=3-2]
	\arrow["{\{0\}\times \mathrm{id}}", from=3-1, to=2-1]
	\arrow["{\{1\}\times \mathrm{id}}"', from=1-2, to=2-1]
	\arrow["BT"', from=3-2, to=2-4,dashed]
	\arrow["BR",bend left=50, from=1-2, to=2-4,dashed]
	\arrow["{p_2}", from=2-1, to=2-2]
	\arrow[from=2-2, to=2-4,dashed]
\end{tikzcd}\]
which implies the existence of a commutative triangle
\[\begin{tikzcd}
	{B{U}_1^{ 2}} && {BU_1\mathrlap{.}} \\
	& {BU_2}
	\arrow["Bd"', from=1-1, to=2-2]
	\arrow["BT"', from=2-2, to=1-3,dashed]
	\arrow["BR", from=1-1, to=1-3,dashed]
\end{tikzcd}\]
We may equivalently consider this diagram in $\mathcal{L}$, in which case it becomes
 \[\begin{tikzcd}
	{\mathsf{U}_1^{ 2}} && {\mathsf{U}_1\mathrlap{.}} \\
	& {\mathsf{U}_2}
	\arrow["\mathsf{d}"', from=1-1, to=2-2]
	\arrow["\mathsf{T}"', from=2-2, to=1-3,dashed]
	\arrow["\mathsf{R}", from=1-1, to=1-3,dashed]
\end{tikzcd}\]
Since $\mathsf{U}_1$ is abelian, this diagram must commute on the
  nose. The universality of abelianization tells us that $\mathsf{T}$
  must factor through $\mathsf{det}$, and thus $\mathsf{T}\circ
  \mathsf{d}$ must take the form $(e^{\mathrm{i}\uptheta},e^{\mathrm{i}\upphi})\mapsto e^{\mathrm{i} \mathsf{m} (\uptheta+\upphi)}$ for some $\mathsf{m}\in \ssZ$. There is, however, no choice of $b \in
\mathsf{GL}_2(\ssZ)$  in Eq.~\ref{eq:f1k} such that $\mathsf{R}$ takes this
form, so our assumption that the action extends to $Her_2/\!\!/U_2$
must be false.
We thus have a smoothness anomaly,  in the familiar
  world of unitary quantum mechanics.
\subsection{Invertible TQFTs in $d=2$\label{sec:2dTFT}}
In \cite{GRWTS23}, we gave a classification of the group actions and homotopy
fixed points on the space of framed or oriented TQFTs in $d=2$ taking
values in the bicategory whose objects are algebras over
$\ssC$, 1-morphisms are bimodules and 2-morphisms are intertwiners.

We have already seen in the case of $d=1$ that there is a freedom in
choosing how to realise smooth versions of TQFTs, namely by thinking
of them as QFTs on spacetimes equipped with either a smooth map to a
point or a metric, but with trivial dynamical evolution. We expect the
same in higher $d$, but unfortunately no explicit description of a
smooth space of such QFTs is yet available.

Neverthless, it seems reasonable
to assume that in the former case we get the same pattern of
symmetries as in \cite{GRWTS23} ({\em i.e.} we can assume that any group
is discrete) and that in the latter case we endow the groups with
their usual smooth structure. Considering the latter case, the
permutation group $\mathsf{S}_n$ can act on
$B(\mathbb{C}^\times)^{\oplus n}$ via permutation, and this action has
a natural homotopy fixed point derived from the inclusion
$\mathsf{S}_n\to (\ssC^\times)^{\oplus n}\rtimes
\mathsf{S}_n$. Following the discussion in Section \ref{sec:findact}, this
induces an $S_n$-action with homotopy fixed point on
$B^2(\mathbb{C}^\times)^{\oplus n}$ which we denote
$B^2(\mathbb{C}^\times)^{\oplus n}/\!\!/S_n$. This provides us with a
natural guess for the smooth space $X$, which we note is connected and
therefore (in this topos) path-connected.

Before delving into the nitty-gritty, let us make some remarks that
may help with the physics interpretation. Very roughly, given a $0$-form
group $G_0$, a homotopy fixed point assigns to each $g_0
\in G_0$ a natural transformation from the functor (from the bordism
category to some target category) defining the QFT to
itself. A natural transformation in a 1-category assigns to each
object in the source a morphism in the target, but in a higher
category it assigns in addition 2-morphisms to 1-morphisms and so
on. Here, a QFT assigns to the closed 1-manifold $\mathsf{S}^1$ the
$(\ssC,\ssC)$-bimodule, {\em i.e.} the $\ssC$-vector space $\ssC^n$, which we interpret
here as the state space associated to $\mathsf{S}^1$, while we assign to
$g_0$ a linear map from $\ssC^n$ to itself, which we can think of as the
usual action of the group on the state space.

Continuing, given a group object $G_1$ whose underlying smooth space is connected 
  (a `$1$-form group')
,  a homotopy fixed point assigns to each $g_1$ in $G_1$ a 2-natural transformation, which assigns to an object in the
source a $2$-morphism in the target, and so on. So here, a QFT assigns
to a closed 0-manifold ({\em i.e.} a point) the commutative $\ssC$-algebra $\ssC^n$ over $\ssC$, while
to $g_1$ we assign an intertwiner from the identity 1-morphism on
$\ssC^n$, {\em i.e.} the bimodule
$_{\ssC^n}(\ssC^n)_{\ssC^n}$, to itself.

Now for the nitty-gritty. Given a $G$-action on $X:=B^2(\mathbb{C}^\times)^{\oplus n}/\!\!/S_n$, we factorize to
obtain the diagram
\pastingVE{X}{X/\!\!/G}{BS_n}{BS_n/\!\!/G}{\terminal}{BG\mathrlap{.}}
The bottom square defines a $G$-action on $BS_n$. From \S\ref{sec:findact}, we
know that for this action we can replace $G$ with $\tau_{\leq 0}G$. Supposing
this to be a Lie group, we get homotopy fixed points given by smooth
maps from $\pi_0\mathsf{G}$ to $\mathsf{S}_n$ that are twisted homomorphisms, in the
sense of Eq.~\ref{eq:TwistedReps}. Since
the state space associated to $\mathsf{S}^1$ is $\ssC^n$, we interpret these as
twisted representations of the state space that are special in that
the structure of the TQFT forces them to act by permuting the states.

Although every homotopy fixed point of a $G$-action on $X$ induces one
on $\tau_{\leq 0} X$, it is not the case that the former are completely
determined by the latter, even when $G$ is $0$-truncated. This is most
easily seen by considering the invertible case where $n=1$ and the only
action on $\tau_{\leq 0} X = BS_1 = \terminal$ is the trivial one.
We then get an action of the form 
 \squareE{B^2\bbC^\times}{B^2\bbC^\times/\!\!/G}{\terminal}{BG\mathrlap{.}}{}{}{}{}
The space of homotopy fixed points ${\Sm}_{/BG}(BG,
B^2\bbC^\times/\!\!/G)$  is in \cite{NSS14} interpreted as the degree-two $\C^\times$-cohomology of
$BG$ with coefficients in a local system given by the
action of $G$ on $\bbC^\times$.
In general, this cohomology is not related to any traditional
cohomology theory, even when $G$ is a Lie group. However, in the case
of a trivial action, the story is somewhat simplified. For instance,
when $G$ is a compact Lie group, we have, following \cite[Thm. 6.4.38]{Sch13},
\[
	\pi_0{\Sm}_{/BG}(BG, B^2\bbC^\times\times BG)=H^2_{\mathrm{Segal}}(\mathsf{G}, \ssR_+)\times H^3_{\mathrm{Segal}}(\mathsf{G},\ssZ)\mathrlap{,}
\]
where $H^\ast_{\mathrm{Segal}}(-,-)$ represents Lie group cohomology
as defined by Segal~\cite{Seg70,Bry00}. The first factor simply corresponds to
smooth cocycles {\em i.e.} smooth maps $\chi:\mathsf{G}\times \mathsf{G} \to \ssR_+$ satisfying 
\[
	\chi(\mathsf{g}_2,\mathsf{g}_3)+\chi(\mathsf{g}_1,\mathsf{g}_2\mathsf{g}_3)=\chi(\mathsf{g}_1\mathsf{g}_2,\mathsf{g}_3)+\chi(\mathsf{g}_1,\mathsf{g}_2).
\]
We have that  $e^{\chi(\mathsf{g}_1,\mathsf{g}_2)}\in \ssC^\times$ represents a self-intertwiner
of the bimodules $_{\ssC}\ssC_{\ssC}$ associated to the
1-dimensional interval.

Sticking with the invertible case $n=1$, let us study pure 1-form
symmetries. So we take $G$ to be $1$-connected and $2$-truncated, so
that we may write it as, $B^2\hat G$, and we fruthermore take $\hat G$
to correspond to a Lie group. From the discussion in \S\ref{sec:findact} on
Eilenberg-Maclane objects, it follows that, since there is only one
group action of $G$ with homotopy fixed point on $ B\bbC^\times$,
there is also only one on $B^2\bbC^\times$, and this corresponds to the trivial action. The homotopy fixed point space is given by ${\Sm}(B G, B^2 \bbC^\times)$. There is an essentially
surjective functor ${\Sm}_\terminal(B G, B^2 \bbC^\times)\to {\Sm}(BG, B^2 \bbC^\times)$ induced by `forgetting the
point'. We have, from \cite[Prop. 7.2.2.12]{Lur09}, that ${\Sm}_\terminal(B^2\hat{G}, B^2 \bbC^\times)$
is equivalent to the set of homomorphisms $\mathsf{r}$ from $\hat{\mathsf{G}}$ to
$\ssC^\times$. Thus, every homotopy fixed point can be traced back to a 1-dimensional $\hat{\mathsf{G}}$ representation $\mathsf{r}$ with values in $\ssC^\times$. For each $\mathsf{g}\in \hat{\mathsf{G}}$ we can again have that  $\mathsf{r}(\mathsf{g})\in \ssC^\times$ is a self-intertwiner of the bimodules $_{\ssC}\ssC_{\ssC}$ associated to the 1d-interval.

\section{Closing words}
We have described how the language of $\infty$-topoi, in general, and
smooth spaces, in particular, can be used to formulate the notion of
smooth generalized symmetries of dynamical quantum field theories.
Though we have tried to put the basic concepts in place, it is clear that we have
barely scratched the surface in terms of what could be done. Indeed,
comparing to \cite{GRWTS23}, we see that we have only generalized the very
first notion there, namely that of generalized global symmetries. It
remains, for example, to
study gauge symmetries and associated anomalies, for which we
uncovered a rich story
in \cite{GRWTS23}, tied to the cobordism hypothesis. We expect the same to be true in the
smooth context.

Elevating symmetries to smooth symmetries also offers us the hope of
being able to differentiate. There are tools for doing this in the
topos of smooth spaces \cite{Sch13,SS20} and this presumably furnish us with the
means of deriving a generalized version of Noether's theorem relating
smooth symmetries and conserved charges of QFTs, as we briefly
sketched in \cite{GRWTS23}.

Part of the reason that we have not been able to do very much in terms
of applications is certainly the intrinsic mathematical difficulty of
working in an $\infty$-topos. But another part is that we currently
lack many examples of smooth spaces of QFTs. This in turn is partly
due to the lack of understanding of what higher category,
generalizing the 1-category of vector spaces and linear maps, should
be used as the target for QFTs. We have reasons to hope, however, that
progress will be made on this question in the near future; we hope
that this will allow the power of the $\infty$-categorical methods described here to be
brought to bear on physics. ``To infinity and beyond!''

\let\oldaddcontentsline\addcontentsline
\renewcommand{\addcontentsline}[3]{}
\section*{Acknowledgments} 
We thank David Ayala and Pelle Steffens for discussions, along with
the referee for the suggestion, amongst others, to consider the space $G\mathrm{Conn}(M)$. BG is supported by Science and Technology Facilities Council (STFC) consolidated grant ST/T000694/1. JTS is supported  by the U.S. National Science Foundation (NSF) grant PHY-2014071. 
\let\addcontentsline\oldaddcontentsline

\appendix
\renewcommand\thesubsection{\thesection.\arabic{subsection}}
\addtocontents{toc}{\protect\setcounter{tocdepth}{1}}
\section{Group actions on $B_\nabla K$}
    \subsection{Lie groups acting on $B_\nabla K$ \label{sec:grpSSc}}
As we have seen in \S\ref{sec:orbstab}, any 
homotopy fixed point of any
action of any $G$ corresponds to a homotopy fixed point of a corresponding action
on some connected component, which is equivalent to
$B\esh K$, and for these we can take $G$ to be $0$-truncated,
without loss of generality.

If we further assume that $G$ corresponds to a Lie group
$\mathsf{G}$, then, as we have argued in \S\ref{sec:grpSSl},
the analysis of actions on $B\esh K$ and
their homotopy fixed points collapses to the analysis of short exact
sequences of groups and their splittings.

We can also say something about the presence of smoothness anomalies
in this case, at least if we also assume that the Lie group $\mathsf{K}$ is
  abelian. This will enable us to give another concrete  example of a
  smoothness anomaly in Appendix~\ref{sec:TFTs1dTake3}, albeit in a
  case that is most likely of mathematical, rather than physical, interest.

The arguments are somewhat involved, so let us first sketch the
essence of the
idea before dotting the $i$s and crossing the $\hbar$s to make everything precise. Suppose the action of 
$G$ on $B\esh K$ corresponds to assigning the automorphism
$\mathsf{a}: \esh{\mathsf{K}} \to \esh{\mathsf{K}}$ to some element of
$\mathsf{G}$. If there is no smoothness anomaly, this must lift to
some autoequivalence $f: B_\nabla K \to B_\nabla K$, which in turn
induces an autoequivalence on the space of bundles with connection for every
manifold $\mathsf{M}$.

Now, given a principal $\mathsf{K}$-bundle with connection on a
manifold $\mathsf{M}$, we can
compute its holonomy. When the Lie group $\mathsf{K}$ is abelian, this
allows us to associate an element in $\mathsf{K}$ to a loop in
$\mathsf{M}$; since a small deformation of a loop leads
to a small change in the holonomy, we expect that this leads to a
suitably defined morphism in $\Sm$.

The holonomy depends
only on restrictions of the principal bundle with connection to 
$\mathsf{S}^1$, where the connection is
always flat. Since we have already seen that $B\esh K$ can be
interpreted as a smooth space
classifying principal bundles with flat connections, the induced
effect of $f$ on holonomy
is to post-compose the holonomy morphism of the domain bundle with
connection on $\mathsf{M}$
with the map $\mathsf{a}$. But if $\mathsf{a}$, which is
certainly smooth with respect to $\esh{\mathsf{K}}$, is not smooth with
respect to the smooth structure on $\mathsf{K}$, then it seems
unlikely that the
resulting holonomy morphism of the codomain bundle will be smooth; if
it isn't, $f$ cannot
exist and so we have a smoothness anomaly. 

Now for the gory details. Let $\mathsf{G}$ be a Lie group and
$\mathsf{K}$ be an abelian Lie group. 
 Starting from a $G$-action on $B_\nabla
  K$ admitting a homotopy fixed point, we get a $G$-action on $B\esh
  K$. Choosing a horn filler of $\ast \to BG\leftarrow \ast$ (which
  corresponds to choosing an element of $\mathsf{G}$), we get a commutative diagram
\begin{equation}\label{eq:AutoDiagram}
\begin{tikzcd}[row sep={45,between origins}, column sep={45,between origins}]
B\esh K\ar{dd} \ar[rr,"j",pos=0.5]\and    \and   B_\nabla K \ar[crossing over]{rr}  \and \and \terminal\\[\perspective]
 \and B\esh K \ar[rr,"j",pos=0.2] \ar{dl} \ar[ul,"Ba",swap] \and     \and B_\nabla K \ar[ul,"f",swap] \ar{rr} \ar{dl} \and\and \terminal\ar[""{name=0, anchor=center, inner sep=0},dl]\ar{ul} \\[-\perspective]
B\esh K/\!\!/G\ar{rr} \and  \and    B_\nabla K/\!\!/G \ar{rr} \ar[from=uu,crossing over]\and \and BG\mathrlap{,}\ar[from=uu,crossing over]\and
\arrow[{description}, shorten >=3pt, Rightarrow, from=1-5, to=0]
\end{tikzcd}
\end{equation}
in which all squares are cartesian, the bottom and front rectangles
are equivalent to one another, and in which the morphisms denoted
$Ba$ and $f$ must be equivalences, since they are pullbacks of the
equivalence $\ast \to \ast$.

Supposing that we instead start from a  a $G$-action on $B\esh K$ admitting a homotopy fixed point,
described by the homomorphism $\mathsf{e}:\mathsf{G}\times
\esh{\mathsf{K}}\to   \esh{\mathsf{K}}$, we construct the allowed
outer prisms in this diagram. All the objects in this prism are
deloopings of Lie groups, so we can replace by an equivalent
  diagram in the category ${\cal L}$ of the form
\[\begin{tikzcd}[row sep={45,between origins}, column sep={45,between origins}]
     \esh{\mathsf{K}}\ar[crossing over]{rr}  \and \and \terminal\\[\perspective]
    \and \esh{\mathsf{K}}\ar[ul,"\mathsf{a}",swap] \ar{rr} \ar{dl} \and\and \terminal\ar[""{name=0, anchor=center, inner sep=0},dl]\ar{ul} \\[-\perspective]
        \esh{\mathsf{K}}\rtimes  \mathsf{G} \ar{rr} \ar[from=uu,crossing over]\and \and \mathsf{G} \mathrlap{.}\ar[from=uu,crossing over]\and
\arrow[{description}, shorten >=3pt, Rightarrow, from=1-3, to=0]
\end{tikzcd}\]
Choosing the 2-simplex filler indicated in the diagram 
to be $\mathsf{g}\in \mathsf{G}$, the most general allowed form of
  $\mathsf{a}$, based on different manifestations of the pullback
  squares, is given by $\mathsf{b}\circ
  \mathsf{e}_\mathsf{g}^{-1}\circ \mathsf{b}^{-1}$, where $\mathsf{b}$
  is an automorphism of $\esh{\mathsf{K}}$.\footnote{This statement remains
  true if we replace $\esh{\mathsf{K}}$ by any abelian Lie group, a fact we will use for our second example of a smoothness anomaly
  in \S\ref{sec:UnitQM}.}

A smoothness anomaly will result if there exists $\mathsf{g}\in
  \mathsf{G}$ such that it is not possible to find a corresponding
  allowed form of $\mathsf{a}$ and an equivalence $f$ making a
  commutative diagram of the form 
\[\begin{tikzcd}
	{BK^\delta} & {BK^\delta} \\
	{B_\nabla K} & {B_\nabla K \mathrlap{.}}
	\arrow["Ba", from=1-1, to=1-2,dashed]
	\arrow[from=1-1, to=2-1]
	\arrow[from=1-2, to=2-2]
	\arrow["{f}"', dashed, from=2-1, to=2-2]
      \end{tikzcd}\]

   Now we simplify the problem using holonomy, which defines a morphism in $\Sm$, denoted
$\mathrm{hol}: (B_\nabla K)^{S^1}\to K$,\footnote{When $\mathsf{K}$ is
  non-abelian, we instead get a morphism to the smooth space
  $K/\!\!/K$, where the action of $\mathsf{K}$ on itself is
  by conjugation.} where $(-)^{(-)}:\Sm \times
\Sm^{op}\to \Sm$ denotes the internal hom in $\Sm$ and $S^1$ is the smooth
space corresponding to the circle manifold $\mathsf{S}^1$. This is described as follows~\cite{FSS13}. 
Let $\mathsf{X}$ be a manifold.  Then $\Sm(X,(B_\nabla K)^{S^1})$ is
the groupoid of principal $\mathsf{K}$-bundles with connection over
$\mathsf{X}\times\mathsf{S}^1$. From such a principal bundle we get a smooth
map $\mathsf{X}\to\mathsf{K}$, \emph{i.e.} an object of
  $\Sm(X,K)$, by taking for each $\mathsf{x} \in\mathsf{X}$ the
  holonomy along $\{\mathsf{x}\} \times \mathsf{S}^1$.
This construction is natural in $X$ and so gives rise to a morphism $(B_\nabla K)^{S^1}\to
K$.

Consider now the addition of the following data: A smooth
manifold $\mathsf{M}$ and a principal $\mathsf{K}$-bundle
$\mathsf{P}$ on $\mathsf{M}\times \mathsf{S}^1$ with connection. Such
a bundle is an object in $\Sm(M\times S^1, B_\nabla K)$ or equivalently
an object in $\Sm(M,  (B_\nabla K)^{S^1})$; we denote the latter object
by $P:M\to
(B_\nabla K)^{S^1}$. With this data, we can form the following commutative diagram
\[\begin{tikzcd}
	& {(BK^\delta)^{S^1}} & {(BK^\delta)^{S^1}} \\
	M & {(B_\nabla K)^{S^1}} & {(B_\nabla K)^{S^1}} \\
	& K & K \mathrlap{.}
	\arrow["{(Ba)^{S^1}}", from=1-2, to=1-3,dashed]
	\arrow[from=1-2, to=2-2]
	\arrow[from=1-3, to=2-3]
	\arrow["{f^{S^1}}"', dashed, from=2-2, to=2-3]
	\arrow["\mathrm{hol}", from=2-3, to=3-3]
	\arrow["\mathrm{hol}"', from=2-2, to=3-2]
	\arrow["P", from=2-1, to=2-2]
      \end{tikzcd}\]
From this diagram we get two morphisms $r:M\to K$ and $r^\prime:
  M\to K$ defined, respectively, as $\mathrm{hol}\circ P$ and $\mathrm{hol}\circ f^{S^1}
  \circ  P$. Since $r$ and $r^\prime$ correspond to smooth maps
  between
  manifolds, they are determined by their underlying maps of sets, or equivalently by the morphisms $\Gamma r$ and $\Gamma
    r^\prime$ in $\cal{S}$.

On
  applying $\Gamma:\Sm\to {\cal S}$ to the above diagram, we can make
  use of
  the fact that  $\Sm(S^1,BK^\delta)\to \Sm(S^1, B_\nabla K)$ is
  an equivalence. Indeed it is $-1$-truncated, since $BK^\delta \to
  B_\nabla K$ is, and it is $-1$-connected, since every bundle with
  connection over $\mathsf{S}^1$ is flat. 
 Since the space $\Sm(S^1,BK^\delta)$ corresponds to the groupoid
    of principal $\mathsf{K}^\delta$-bundles on $\mathsf{S}^1$,
    an object of $\Sm(S^1,BK^\delta)$ corresponds to a
    homomorphism $\mathsf{h}:\ssZ\to \mathsf{K}^\delta$. The
  morphism $\Gamma \mathrm{hol}:\Sm(S^1,BK^\delta)\to \Gamma K$
  corresponds to evaluating $\mathsf{h}$ at $1 \in \ssZ$ and
  the morphism $\Gamma (Ba)^{S^1}$ corresponds to post-composing with
  $\mathsf{a}$. Thus, we have a commutative diagram in ${\cal S}$ of
  the form
\[\begin{tikzcd}
	{\Sm(S^1,BK^\delta)} & {\Sm(S^1,BK^\delta)} \\
	{\Gamma K} & {\Gamma K \mathrlap{,}}
	\arrow["{\Gamma (Ba)^{S^1}}", from=1-1, to=1-2]
	\arrow["\Gamma \mathrm{hol}", from=1-2, to=2-2]
	\arrow["\Gamma a"', from=2-1, to=2-2]
	\arrow["\Gamma \mathrm{hol}"', from=1-1, to=2-1]
\end{tikzcd}\]
where $\Gamma
  a$ is the morphism in ${\cal S}$ corresponding to the map of
    sets underlying the smooth (but trivially so) homomorphism $\mathsf{a}:\esh{\mathsf{K}}\to \esh{\mathsf{K}}$. 
From this, we can deduce that in ${\cal S}$ we have $\Gamma
r^\prime=\Gamma a \circ \Gamma r $, or equivalently that
$
\mathsf{r}^{\prime \delta}=\mathsf{a} \circ\esh{\mathsf{r}}
$
as maps of sets.
If we can construct $\mathsf{r}$ and $\mathsf{a}$ such that the set map
$\mathsf{a} \circ\esh{\mathsf{r}}$ does not descend from a smooth map
$\mathsf{r}^\prime$,
then the group action we started with will
have a smoothness anomaly. We will see an example of this  in Appendix~\ref{sec:TFTs1dTake3}.

\subsection{Smooth TQFTs in $d=1$ -- Take three\label{sec:TFTs1dTake3}}

In this appendix we repeat the analysis of
  \S\ref{sec:TFTs1dTake1}, taking as $X$ the smooth space $\coprod_n
  B_\nabla GL_n$. This is a possible (though admittedly
  unlikely) candidate for  the smooth
space version of TQFTs in $d=1$. 

 We find that 
the space called $B GL_n$ there is replaced not by the
connected smooth space $BGL_n$, as one might na\"{\i}vely have guessed, but
rather by the non-connected smooth space $B_\nabla
GL_n$. This seems perhaps surprising at first glance, but the
arguments of \S\ref{sec:grpSSc} show that the upshot for physical
symmetries, at least from Simplicio's point of view, is entirely reasonable: for a group object $G$
corresponding to a Lie group $\mathsf{G}$, say, we can replace $B_\nabla
GL_n$ by $B\esh{GL_n}$ and $\mathsf{G}$ by $\pi_0\mathsf{G}$. In
other words, the symmetries of the smooth version are unchanged 
compared to those of the original: the relevant group actions are
specified by homomorphisms $\pi_0\mathsf{G} \to \mathrm{Aut} \,
\esh{\mathsf{GL}_n}$, where, as before, $\esh{\mathsf{GL}_n}$ is $\mathsf{GL}_n$ considered as a discrete group, and an associated homotopy fixed point is
given by a set map $\mathsf{r}: \pi_0\mathsf{G} \to
\esh{\mathsf{GL}_n}$ satisfying the corresponding version of Eq.~\ref{eq:TwistedReps}. The group of automorphisms of
$\esh{\mathsf{GL}_n}$, as described in \cite{Die63}, is
generated by inner automorphisms, automorphisms involving determinants,
inversion-transposition and automorphisms
involving field
automorphisms. 

So a smooth 1-d oriented $\mathsf{G}$-symmetric TQFT
 is, at least in this incarnation, a
representation of $\pi_0\mathsf{G}$, but twisted in the sense of
Eq.~\ref{eq:TwistedReps}; the morphism that forgets the $\mathsf{G}$-symmetry sends such a
twisted representation to its underlying vector space.

The general discussion on smoothness anomalies for $B_\nabla K$ for
abelian $\mathsf{K}$ in
\S\ref {sec:grpSSc} can be used to demonstrate that we can have a smoothness
anomaly in the $n=1$ case (so $\mathsf{K}=\ssC^\times$).

We begin by asserting that there exists (at least assuming the axiom
of choice) an automorphism of order three
of the discrete group $\esh{(\ssC^\times)}$.  For instance, given $\ssC^\times\simeq
  \ssR\times \mathsf{U}_1$, where $\ssR$ is the group of
  additive reals, one can form such an automorphism using an order-three
  permutation of a basis for $\esh{\ssR}$ (constructed using the axiom
  of choice), considered as a vector space over
  the rationals.

   It will be crucial for what
  follows that such an automorphism cannot lie in the same conjugacy class as
  an automorphism of $\esh{(\ssC^\times)}$ that is smooth with respect to the usual Lie group structure on $\ssC^\times$, since such an automorphism must have
  an order valued in $\{1,2,\infty\}$.

Proceeding, consider the action of $\Z/3$ on $B\esh{(\C^\times)}$ described by the homomorphism
$\mathsf{e}:\ssZ/3\times\ssC^\times\to\ssC^\times$ that sends the
generator of $\ssZ/3$ to such an automorphism of order three.
Let us assume that this action
extends to an action on $B_\nabla \C^\times$ and show that it leads to a contradiction. Following the
argument in \S\ref{sec:grpSSc}, we choose $\mathsf{M}=\ssR^2$,
and $\mathsf{P}$ to be the trivial principal $\ssC^\times$-bundle
on $\ssR^2\times\mathsf{S}^1$  with connection $1$-form $\mathsf{A}$
given,
at $((\mathsf{x},\mathsf{y}),{\uptheta},\mathsf{k})\in \ssR^2\times \mathsf{S}^1\times \ssC^\times$,
by
\[
	\mathsf{A}=\mathsf{k}^{-1}\mathrm{d}\mathsf{k}+\frac{1}{\pi}(\mathsf{x}+\mathrm{i} \mathsf{y}) \cos^2\uptheta \mathrm{d}\uptheta .
\]
The morphism $r$ described in \S\ref{sec:grpSSc}
for this choice of $\mathsf{M}$ and $\mathsf{P}$ corresponds to the
smooth surjective local diffeomorphism  $\mathsf{r}:\ssR^2\to \ssC^\times:(\mathsf{x},\mathsf{y})\mapsto
e^{\mathsf{x}+\mathrm{i} \mathsf{y}}$ implying (since $\mathsf{r}^\prime$ is also
smooth) that the automorphism $\mathsf{a}$ described in \S\ref{sec:grpSSc} is
smooth. But
the automorphism $\mathsf{a}$ is also an
automorphism of order three, being in the same conjugacy class as the
original automorphism, so cannot be
smooth, which is a contradiction. Thus the assumption that the
action extends must be false, and we have another example of a
smoothness anomaly.


\section{Proofs and other results}
\tocless\subsection{Actions equipped with homotopy fixed points as actions in the arrow topos \label{proof:Actions-HFPs}}
\noindent\emph{\S\ref{sec:fix}, \pagelink{page:Actions-HFPs}{p. \pagenumber{page:Actions-HFPs}}.}

	The identity morphism $\mathrm{id}_{BG}$ of $BG$, as an object
        in ${\cal O}_{\X}:=\Fun(\Delta^1,X)$, is connected and
       inherits a basepoint from $BG$ in an obvious way.
Equipped as such, $\mathrm{id}_{BG}$ describes a group object in
  ${\cal O}_{\X}$ that we denote $\Omega \mathrm{id}_{BG}$. The topos $({\cal O}_{\X})_{/\mathrm{id}_{BG}}$ can then be viewed  as the category of $\Omega \mathrm{id}_{BG}$-actions. 
	
	The canonical functor  ${\cal O}_{\X_{/BG}}\to ({\cal O}_{\X})_{/\mathrm{id}_{BG}}$ is an equivelence~\cite[Prop. 4.2.12]{Cis19}. We want to prove the following theorem.
	
	\newtheorem*{thm}{Theorem}

	\begin{thm}
	An object of $({\cal O}_{\X})_{/\mathrm{id}_{BG}}$  is in the essential image of the canonical functor ${\cal O}_{\X_{/BG}}\to ({\cal O}_{\X})_{/\mathrm{id}_{BG}}$ restricted to the full subcategory $(\X_{/BG})_\terminal \subseteq {\cal O}_{\X_{/BG}}$ if and only if it corresponds to a $\Omega\mathrm{id}_{BG}$-action on an object of the form $\terminal\to X$.
	\end{thm}

	\begin{proof}
	The canonical functor ${\cal O}_{\X_{/BG}}\to ({\cal O}_{\X})_{/\mathrm{id}_{BG}}$ takes objects in the full-subcategory $(\X_{/BG})_\terminal$ to objects in $({\cal O}_{\X})_{/\mathrm{id}_{BG}}$ of the form
	\[\begin{tikzcd}
	BG & E \\
	BG & BG\mathrlap{,}
	\arrow[from=1-1, to=1-2]
	\arrow[from=1-2, to=2-2]
	\arrow[from=1-1, to=2-2]
	\arrow[from=1-1, to=2-1]
	\arrow[from=2-1, to=2-2]
\end{tikzcd}\]
where the bottom 2-simplex is the degenerate $2$-simplex on $\mathrm{id}_{BG}$. Furthermore, it is easy to see that every object in $({\cal O}_{\X})_{/\mathrm{id}_{BG}}$ of this form is the image of some object in $(\X_{/BG})_\terminal$. We denote by $({\cal O}_{\X})_{/\mathrm{id}_{BG}}^{deg}$ the full subcategory of $({\cal O}_{\X})_{/\mathrm{id}_{BG}}$ spanned by such objects .

What remains is to show that an object of $({\cal
  O}_{\X})_{/\mathrm{id}_{BG}}$  is equivalent to an object in
$({\cal O}_{\X})_{/\mathrm{id}_{BG}}^{deg}$ if and only if it
describes a group action on some $x:\ast\to X$ in $ \X_\terminal\subseteq {\cal
  O}_{\X}$. The `only if' direction holds since limits are calculated pointwise~\cite[Cor.~5.1.2.3 ]{Lur09}. We thus focus on the `if' direction.

Let $\mathcal{M}:\Delta^2\times \Delta^1\to \X$ be a diagram in $\X$. We use the following notion for the vertices of~$\mathcal{M}$,
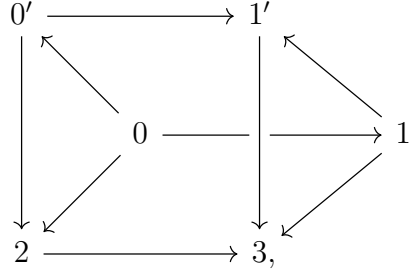
\[\begin{tikzcd}[row sep={45,between origins}, column sep={45,between origins}]
  {0^\prime} \ar[crossing over]{rr}  \and \and {1^\prime} \\[\perspective]
    \and 0 \ar[ul] \ar{rr} \ar{dl} \and\and 1 \ar[""{name=0, anchor=center, inner sep=0},dl]\ar{ul} \\[-\perspective]
       2  \ar{rr} \ar[from=uu,crossing over]\and \and 3 \mathrlap{,}\ar[from=uu,crossing over]\and
\end{tikzcd}\] 
and use subscripts to indicate the subdiagrams obtained by restricting
to the specified vertices. A 1-simplex in  $({\cal
  O}_{\X})_{/\mathrm{id}_{BG}}$ is a diagram $\mathcal{M}$ subject to
the condition that $\mathcal{M}_{23}=\mathrm{id}_{BG}$. Such a diagram
$\mathcal{M}$ is a morphism in  $({\cal O}_{\X})_{/\mathrm{id}_{BG}}$
from  $\mathcal{M}_{0123}$ to $\mathcal{M}_{{0^\prime}{1^\prime}23}$
(considered as objects in $({\cal O}_{\X})_{/\mathrm{id}_{BG}}$) and
so it makes sense to ask when it is an equivalence. It turns out that
it is an equivalence if  $\mathcal{M}_{0{0^\prime}}:\Delta^1\to \X$
and $\mathcal{M}_{1{1^\prime}}:\Delta^1\to \X$, considered as
morphisms in $\X$ (through the image of the morphism $0\to 1$ of $\Delta^1$), are equivalences.

Let $\sigma:\Delta^1\times \Delta^1\to \X$ be an object $({\cal
  O}_{\X})_{/\mathrm{id}_{BG}}$ representing an
$\Omega\mathrm{id}_{BG}$-action on $x:\ast \to X$. We prove, by
explicit construction, the existence of an equivalence $\mathcal{M}$
from an object in $({\cal O}_{\X})_{/\mathrm{id}_{BG}}^{deg}$ to
$\sigma$. We construct $\mathcal{M}$ as follows. For $\mathcal{M}$  to
be a 1-simplex in $({\cal O}_{\X})_{/\mathrm{id}_{BG}}$ we require
$\mathcal{M}_{23}=\mathrm{id}_{BG}$. For $\mathcal{M}$ to be a
morphism into $\sigma$ we require
$\mathcal{M}_{{0^\prime}{1^\prime}23}=\sigma$. For $\mathcal{M}$ to be
a morphism from an object in  $({\cal
  O}_{\X})_{/\mathrm{id}_{BG}}^{deg}$ we require $\mathcal{M}_{023}$
to be the degenerate simplex on $\mathrm{id}_{BG}$. We now use
that $\sigma$ is a $\Omega \mathrm{id}_{BG}$-action on $x$ to conclude
that $\mathcal{M}_{{0^\prime}2}$ is an equivalence. It follows
  that a filler exists for the $\Lambda^2_2$ horn of
$\mathcal{M}_{0{0^\prime}2}$; we choose such a filler, which also
fixes $\mathcal{M}_{0{0^\prime}}$. For any choice of filler $\mathcal{M}_{0{0^\prime}}$ is an equivalence, as a morphism in $\X$. For $\mathcal{M}$
itself to be an equivalence (as a morphism in $({\cal O}_{\X})_{/\mathrm{id}_{BG}}$) we take
$\mathcal{M}_{1}=\mathcal{M}_{{1^\prime}}$ and
$\mathcal{M}_{1{1^\prime}}$ to be the identity of
$\mathcal{M}_{1}$. We now fill in the remaining data. We take
sequentially $\mathcal M_{0{0^\prime}{1^\prime}}$, $\mathcal
M_{01{1^\prime}}$, $\mathcal{M}_{1{1^\prime}3}$, $\mathcal
M_{0{0^\prime}23}$, $\mathcal{M}_{0{0^\prime}{1^\prime}3}$, and
$\mathcal{M}_{01{1^\prime}3}$ as fillers of the horns $\Lambda^2_1$,
$\Lambda^2_2$, $\Lambda^2_1$, $\Lambda^3_2$, $\Lambda^3_1$, and
$\Lambda^3_2 $ . These fix, respectively,   $\mathcal
M_{0{1^\prime}}$,  $\mathcal{M}_{01}$, $\mathcal M_{13}$,  $\mathcal
M_{0{0^\prime}3}$, $\mathcal{M}_{0{1^\prime}3}$, and
$\mathcal{M}_{013}$. Note that we can fill the $\Lambda^2_2$ horn
of $\mathcal{M}_{01{1^\prime}}$ since $\mathcal{M}_{1{1^\prime}}$ is
an equivalence, as a morphism in $\X$. 

The existence of the equivalence $\mathcal{M}$ proves our theorem.
\end{proof}
The conclusion drawn from this theorem is that $G$-actions equipped with homotopy fixed points can be thought of as $\Omega\mathrm{id}_{BG}$-actions on objects in $\X_\terminal$. In addition to this,  such an $\Omega\mathrm{id}_{BG}$-action has an underlying object $x:1\to X$ if  its preimage in  $({\X_{/BG}})_\terminal$ lands on $x$ under the functor $({\X_{/BG}})_\terminal\to \X_{\terminal}$ induced by $b_{BG}^\ast$.

\tocless\subsection{Relation between $m$-connected covers of
  $X/\!\!/G$ for connected group objects \label{proof:OrbitSpace}}
\noindent\emph{\S\ref{sec:orbstab}, \pagelink{page:OrbitSpace}{p. \pagenumber{page:OrbitSpace}}.}

\begin{thm}
 If $BG$ is $p$-connected, then any section $s:BG\to X/\!\!/G$ through $x:\ast \to X$ has a $(m-2)$-connected/$(m-2)$-truncated factorization through $\tge{m}{[x]}{X/\!\!/G}\to X/\!\!/G$ for $0\le m\le p+1$.
\end{thm}

\begin{proof}
We have $(m-2)$-connected/$(m-2)$-truncated factorizations $[x]:\ast\to \tge{m}{[x]}{X/\!\!/G}\to X/\!\!/G$ and $s:BG\to \tge{m}{[x]}{X/\!\!/G}\to X/\!\!/G$. Since we also have a factorization of $[x]$ through the basepoint of $BG$ and $s$ we have a commutative square formed by the solid lines of 
\[\begin{tikzcd}[column sep=2cm]
	\terminal & {\tge{m}{s}{X/\!\!/G}} \\
	{\tge{m}{[x]}{X/\!\!/G}} & {X/\!\!/G\mathrlap{,}}
	\arrow["", from=1-2, to=2-2]
	\arrow[""', from=2-1, to=2-2]
	\arrow["", from=1-1, to=1-2]
	\arrow[""', from=1-1, to=2-1]
	\arrow["f",dashed, from=2-1, to=1-2]
\end{tikzcd}\]
which by the $(m-2)$-connected/$(m-2)$-truncated factorization system defines a canonical map $f:\tge{m}{[x]}{X/\!\!/G}\to\tge{m}{s}{X/\!\!/G}$ as shown. Our theorem is equivalent to the statement that $f$ is an equivalence under the given conditions. 

Since the bottom map in the above square is $(m-2)$-truncated, $f$ is itself $(m-2)$-truncated. Thus $f$ is an equivalence if and only if it is $(m-2)$-connected. We now show, through a series of implications, that $f$ is $(m-2)$-connected if  $BG$ is $p$-connected for some $p\geq m-1$. Since $\ast \to \tge{m}{[x]}{X/\!\!/G}$ is $(m-2)$-connected, from~\cite[Prop.~6.5.1.16]{Lur09} and the top left triangle in the above diagram, $f$ is $(m-2)$-connected if and only if  the morphism $\ast \to {\tge{m}{s}{X/\!\!/G}}$ is. But, $\ast \to {\tge{m}{s}{X/\!\!/G}}$ is a composition of $\ast \to BG$ and the $(m-2)$-connected morphism $BG\to  {\tge{m}{s}{X/\!\!/G}}$. Thus, from~\cite[Prop.~6.5.1.16]{Lur09}, $\ast \to {\tge{m}{s}{X/\!\!/G}}$ is $(m-2)$-connected if $\ast \to BG$ is $(m-2)$-connected. From~\cite[Prop.~6.5.1.20]{Lur09},  $\ast \to BG$ is $(m-2)$-connected if and only if $BG$ is $(m-1)$-connected. But $BG$ is $(m-1)$-connected if it is $p$-connected for some $p\geq m-1$. 
\end{proof}
 We remark that this argument could equally be viewed
 as being made in $\X_{/BG}$ as in $\X$.

\tocless\subsection{Truncating actions on connected objects equipped with homotopy fixed points  \label{proof:ConnectedTruncatedX}}
\noindent\emph{\S\ref{sec:findact}, \pagelink{page:ConnectedTruncatedX}{p. \pagenumber{page:ConnectedTruncatedX}}.}

\begin{thm}If $X$ is both $n$-truncated and $0$-connected, then every $G$-action on $X$ equipped with a homotopy fixed point is in the essential image of the functor  $(\X_{/B\tau_{\leq n-1}G})_\terminal\to (\X_{/BG})_\terminal$.
\end{thm}
\begin{proof}
Given a $G$-action on $X$ equipped with a homotopy fixed point, we can
form a commutative diagram in $\X$ of the form
\[\begin{tikzcd}
	\ast & BG & {B\tau_{\le n-1} G} \\
	X & {X/\!\!/G} & {\tau_{\le n}(X/\!\!/G)} \\
	\ast & BG & {B\tau_{\le n-1} G\mathrlap{,}}
	\arrow[from=1-2, to=2-2]
	\arrow[two heads, from=1-1, to=1-2]
	\arrow[from=1-1, to=2-1]
	\arrow[two heads, from=2-1, to=2-2]
	\arrow[from=2-2, to=3-2]
	\arrow[two heads, from=3-1, to=3-2]
	\arrow[from=2-1, to=3-1]
	\arrow[from=2-3, to=3-3]
	\arrow[two heads, from=3-2, to=3-3]
	\arrow[two heads, from=2-2, to=2-3]
	\arrow[from=1-3, to=2-3]
	\arrow[two heads, from=1-2, to=1-3]
\end{tikzcd}\]
where the left-hand side squares are cartesian, and the right-hand
squares commute because of naturality. We have used that  $\tau_{\leq n} BG \simeq B \tau_{\leq n-1} G$. We must show that the
right-hand squares are, in addition, cartesian when $X$ is
$n$-truncated and $0$-connected, since then, under the functor
$(\X_{/B\tau_{\leq n-1}G})_\terminal\to (\X_{/BG})_\terminal$, the
object of $(\X_{/B\tau_{\leq n-1}G})_\terminal$ described by the
$2$-simplex formed by the right-hand vertical arrows would be mapped into the object of $(\X_{/BG})_\terminal$  described by the $2$-simplex formed by the middle vertical arrows. 

Now, the right-hand total rectangle is manifestly cartesian. The map $X\to \ast$ is $0$-connected and thus $\ast \to X$ is an effective epimorphism. Thus so is $ BG\to X/\!\!/G$, and so is its truncation $ B\tau_{\le n-1}G\to \tau_{\le n}(X/\!\!/G)$ by successive use of~\cite[Prop. 6.5.1.16 (5)]{Lur09}. Thus, by reverse pasting, \cite[Lem. 3.7.3]{ABFJ20},  if we can show that the top right square is cartesian, the bottom right square will be too.

By~\cite[Rem. 3.64]{Ras20} the top right square will be cartesian if we can show that $BG\to X/\!\!/G$  is $(n-1)$-truncated. To see that it is, first note that $\ast \to X$ is $(n-1)$-truncated, which follows from the fact that looping increases truncatedness. Then note that, since $X\to X/\!\!/G$ is an effective epimorphism, we can use \cite[Prop. 6.2.3.17]{Lur09} to conclude that $BG\to X/\!\!/G$ must be $(n-1)$-truncated. 
\end{proof}

\tocless\subsection{A recipe for finding certain group actions \label{proof:ConnectedActions}}
\noindent\emph{\S\ref{sec:findact}, \pagelink{page:ConnectedActions}{p. \pagenumber{page:ConnectedActions}}.}

It was stated in the main text that we can sometimes find group actions in an arbitrary
topos in much the same way as in the topos of spaces, provided that
certain technical assumptions are satisfied.

We begin by spelling these
assumptions out.
Firstly, we must assume that the object $X$ on which we act can be written as
$\coprod_{i\in A} X_i$, where $A$ is an indexing set and all $X_i$ are
$0$-connected. Secondly, we must assume that every $0$-connected object in
the topos (or at least every object that appears in the recipe below)
has only one point up to homotopy. Thirdly, we must assume that 
 the homotopy quotient $X/\!\!/G$ corresponding to the action of $G$ on $X$
takes the form $\coprod_{\alpha} (\coprod_{i\in A_\alpha}
X_i/\!\!/G)$, where $\{A_\alpha\}$ is a partition of $A$ and all $
(\coprod_{i\in A_\alpha} X_i/\!\!/G)$ are $0$-connected.
In the topos of spaces, these assumptions hold (for any object
and action thereon), so the
recipe gives us an algorithm to find all possible actions on any
object.

Now we describe the recipe.

\paragraph*{Step 1:} The first step is to choose a partition
$\{A_\alpha\}$ of the set $A$ subject to the condition that, for each
$i, j\in A_\alpha$, we have $X_i\simeq X_j$. We can see the need
  for this condition as follows. We want to find  $G$-actions on $X_\alpha:=\coprod_{i\in A_\alpha} X_i$ for each $A_\alpha$ such that $X_\alpha/\!\!/G$ is connected. Suppose we are given such an action $X_\alpha/\!\!/G$. As suggested in the main text, we can take its $0$-connected/$0$-truncated factorization, giving the diagram 
\pastingVE{X_\alpha}{X_\alpha/\!\!/G}{\coprod_{i\in A_\alpha} \ast_i}{\tau_{\le 0}X_\alpha/\!\!/G}{\ast}{BG\mathrlap{.}}
 (We have put subscripts on the terminal objects $\ast$, to indicate which $X_i$ they descend from.) By choosing a $j\in A_\alpha$, we can form the diagram  
\[\begin{tikzcd}
	{X_j} & {X_\alpha} & {X_\alpha/\!\!/G} \\
	{\terminal_j} & {\coprod_{i\in A_\alpha} \terminal_i} & {\tau_{\le 0}X_\alpha/\!\!/G\mathrlap{,}}
	\arrow[from=2-1, to=2-2]
	\arrow[from=1-1, to=1-2]
	\arrow[from=1-2, to=2-2]
	\arrow[from=1-1, to=2-1]
	\arrow[from=2-2, to=2-3,two heads]
	\arrow[from=1-2, to=1-3,two heads]
	\arrow[from=1-3, to=2-3]
\end{tikzcd}\]
in which the right square is cartesian by definition and the left
square is cartesian by the fact that in any topos we have descent over coproducts~\cite{Rezk19}.
By pasting, the rectangle formed from the two squares is also
cartesian. We now make use of our second assumption, namely that
connected objects only have one point up to homotopy. Under this
assumption, since $\tau_{\le 0}X_\alpha/\!\!/G$ is connected, the
 diagram above shows that we must have $X_i\simeq X_j$ for each $i,j\in A_\alpha$ if  $X_\alpha/\!\!/G$ is to be connected.  

\paragraph*{Step 2:} The second step is to choose a $G$-action on
$\tau_{\le 0} X_\alpha\simeq \coprod_{i\in A_\alpha} \terminal_i$ for
which $\tau_{\le 0} X_\alpha/\!\!/G$ is connected.

Choosing a $j\in A_\alpha$, we can now form the diagram 
\[
\begin{tikzcd}
	X_j \ar[r] \ar[d]& X_\alpha \ar[r,dashed,two heads]\ar[d] & ? \ar[d,dashed]\\
	\ast_j \ar[r]& \coprod_{i\in A_\alpha} \terminal_i \ar[r,two heads]& \tau_{\le 0} X_\alpha/\!\!/G\mathrlap{,}
\end{tikzcd}
\]
in which the left-hand square is cartesian, and the rectangle and
right-hand square want to be cartesian too. Here $\tau_{\le 0}
X_\alpha/\!\!/G$ is connected, and pointed by $\ast_j$, thus we may
associate to it a group object $\Omega \tau_{\le 0} X_\alpha/\!\!/G$. 

\paragraph*{Step 3:} For all $\alpha$ and a chosen $j\in A_\alpha$, choose an action of $\Omega \tau_{\le 0} X_\alpha/\!\!/G$ on $X_j$. This gives us an object $E$, which sits in the commutative diagram 
 \[\begin{tikzcd}
	{X_j} & Y & {E} \\
	{\terminal_j} & {\coprod_i \terminal_i} & {\tau_{\le 0}X_\alpha /\!\!/G} \\
	& \terminal& BG\mathrlap{.}
	\arrow[from=1-2, to=2-2]
	\arrow[from=2-1, to=2-2]
	\arrow[from=1-1, to=1-2]
	\arrow[from=1-1, to=2-1]
	\arrow[from=2-2, to=2-3,two heads]
	\arrow[from=1-3, to=2-3]
	\arrow[from=1-2, to=1-3,two heads]
	\arrow[from=2-3, to=3-3]
	\arrow[from=2-2, to=3-2]
	\arrow[from=3-2, to=3-3,two heads]
\end{tikzcd}\]
where all squares are cartesian, which determines $Y$. In fact, since
it does not matter which $j$ we chose, $Y$ must be $\coprod_{i\in
  A_\alpha} X_i$ by the universality of coproducts~\cite{Rezk19}. Thus
every $E$ is in fact a group action on $X_\alpha$, so we are entitled to call it $X_\alpha/\!\!/G$.

\paragraph*{Step 4:} To construct the final action on $X$, we take the coproduct in $\X_{/BG}$ of $\coprod_{\alpha} X_\alpha /\!\!/G$.

\tocless\subsection{Absence of smoothness anomalies for TFTs in the smooth space of QM \label{proof:ExtendToQM}}
\noindent\emph{\S\ref{sec:1dTFTTake2}, \pagelink{page:ExtendToQM}{p. \pagenumber{page:ExtendToQM}}.}

\begin{thm}
Let $G$ be a group object in $\Sm$ corresponding to
a Lie group $\mathsf{G}$. For every $G$-action on $BGL_n$ that admits
a homotopy fixed point, there is a $G$-action on
$X=Mat_n/\!\!/GL_n$ with homotopy fixed point $s$ through $h=0$, such
that $\tge{1}{s}{X/\!\!/G}\to BG$ corresponds to the original
$G$-action on $BGL_n$.
\end{thm}
\begin{proof}
A  $G$-action on $BGL_n$ admitting a homotopy fixed point is specified by a $\mathsf{G}$-action  $\mathsf{e}: \mathsf{G}\times \mathsf{GL}_n\to \mathsf{GL}_n$ in the traditional sense.
The corresponding $G$-action as an object in $\Sm_{/BG}$ is the morphism $B\hat G\to BG$ corresponding to the projection of the semi-direct product $\hat{\mathsf{G}}:=\mathsf{GL}_n\rtimes_\mathsf{e} \mathsf{G}$ onto $\mathsf{G}$.

There exists an $\hat G$-action on $Mat_n$ corresponding to the smooth
map $\hat{\mathsf{G}}\times \mathsf{Mat}_n\to
\mathsf{Mat}_n:((\mathsf{M},\mathsf{g}),\mathsf{h})\mapsto
\mathsf{M}d\mathsf{e}_\mathsf{g}(\mathsf{h}) \mathsf{M}^{-1}$, where
$d\mathsf{e_g}:\mathsf{Mat}_n\to \mathsf{Mat}_n$ is the differential
of $\mathsf{e_g}:\mathsf{GL}_n\to \mathsf{GL}_n$. This action on
$Mat_n$ has a section through $h=0$. We can combine the $G$-action on
$BGL_n$ that yields $\hat G$ and the $\hat G$-action on $Mat_n$, along with their sections, to form the commutative diagram  
\[\begin{tikzcd}
	& \ast & BG \\
	& {BGL_n} & {B\hat G} \\
	{Mat_n} & {Mat_n/\!\!/GL_n} & {Mat_n/\!\!/\hat G} \\
	\ast & {BGL_n} & {B\hat G} \\
	& \ast & BG\mathrlap{,}
	\arrow[from=4-3, to=5-3]
	\arrow[from=5-2, to=5-3]
	\arrow[from=4-2, to=4-3]
	\arrow[from=4-2, to=5-2]
	\arrow[from=3-2, to=4-2]
	\arrow[from=3-3, to=4-3]
	\arrow[from=3-2, to=3-3]
	\arrow[from=2-3, to=3-3]
	\arrow[from=2-2, to=3-2]
	\arrow[from=1-2, to=1-3]
	\arrow[from=2-2, to=2-3]
	\arrow[from=1-3, to=2-3]
	\arrow[from=1-2, to=2-2]
	\arrow[from=3-1, to=3-2]
	\arrow[from=3-1, to=4-1]
	\arrow[from=4-1, to=4-2]
\end{tikzcd}\]
in which all squares are cartesian. By pasting, all rectangles are
cartesian too, so we can read off that $Mat_n/\!\!/\hat G\to BG$ is an
action on $Mat_n/\!\!/GL_n$ with section $s$ through $h=0$. Moreover,
since the morphism $B\hat G\to Mat_n/\!\!/\hat G$ is $(-1)$-truncated
and the morphism $BG\to B\hat G$ is $(-1)$-connected,  the factorization $s:BG\to B\hat G\to Mat_n/\!\!/\hat G$ exhibits $B\hat G$ as $\tge{1}{s}{Mat_n/\!\!/\hat G}$.
\end{proof}
\definecolor{newLinkCC}{RGB}{100, 100, 100}  
\hypersetup{
    linkcolor=newLinkCC
}

\bibliographystyle{alphaurl}
\bibliography{smooth_generalized_symmetry}
\end{document}